\pdfoutput=1
\documentclass[11pt,a4paper]{article}

\usepackage[english]{babel}

\usepackage{amsfonts,amssymb,amsmath,amsthm,mathtools,thmtools}
\usepackage{fullpage,graphicx,bm,xspace}

\usepackage{newpxtext}
\usepackage[semibold]{sourcesanspro}
\usepackage{eulerpx}

\usepackage[cal=euler,calscaled=1.08]{mathalpha}

\usepackage[%
activate={true,nocompatibility},
selected=true,
tracking=true,
factor=1100,
babel=true
]%
{microtype}

\usepackage[%
pagebackref=true,
colorlinks=true,
citecolor=DarkGreen,
linkcolor=Navy
]%
{hyperref}

\usepackage[svgnames]{xcolor}

\usepackage{enumitem}
\setlist{  
  listparindent=\parindent,
  parsep=0pt,
}

\usepackage{bbm}

\usepackage{algorithm}
\usepackage[noend]{algpseudocode}

\usepackage[capitalise, noabbrev, nameinlink]{cleveref}

\usepackage
[%
labelformat=default, 
labelsep=colon, 
textformat=simple, 
font={small,sf}, 
labelfont=bf 
]%
{caption}

\usepackage{sectsty}
\allsectionsfont{\sffamily}

\makeatletter
 \renewenvironment{abstract}{%
   \small%
   \begin{center}%
     {\bfseries \sffamily\abstractname\vspace{-.5em}\vspace{\z@}}%
   \end{center}%
   \quotation
 }
\makeatother

\usepackage[framemethod=tikz]{mdframed}
\declaretheoremstyle[
spaceabove=\parsep, spacebelow=\parsep,
headfont={\sffamily\color{Navy}},
bodyfont=\normalfont,
notefont=\sffamily, notebraces={(}{)},
postheadspace=0.5em,
mdframed={
  linecolor=Navy!50,
  linewidth=1,
  innertopmargin=1pt,
  innerbottommargin=6pt,
  innerleftmargin=6pt,
  innerrightmargin=6pt,
  roundcorner=1pt,
}
]{thmstyle}

\declaretheoremstyle[
spaceabove=\parsep, spacebelow=\parsep,
headfont={\sffamily\color{DarkGreen}},
bodyfont=\normalfont,
notefont=\sffamily, notebraces={(}{)},
postheadspace=0.5em,
mdframed={
  linecolor=DarkGreen!50,
  linewidth=1,
  innertopmargin=1pt,
  innerbottommargin=6pt,
  innerleftmargin=6pt,
  innerrightmargin=6pt,
  roundcorner=1pt,
}
]{defstyle}

\declaretheorem[%
name=Theorem,
parent=section,
style=thmstyle
]%
{theorem}
\declaretheorem[%
name=Lemma,
sibling=theorem,
style=thmstyle
]%
{lemma}
\declaretheorem[%
name=Corollary,
sibling=theorem,
style=thmstyle
]%
{corollary}
\declaretheorem[%
name=Proposition,
sibling=theorem,
style=thmstyle
]%
{proposition}
\declaretheorem[%
name=Claim,
sibling=theorem,
style=thmstyle
]%
{claim}
\declaretheorem[%
name=Fact,
sibling=theorem,
style=thmstyle
]%
{fact}
\declaretheorem[%
name=Definition,
sibling=theorem,
style=defstyle
]%
{definition}

\numberwithin{equation}{section}

\newcommand{\PROB}[2][]{\Pr_{#1} \left[ #2 \right]}
\newcommand{\Prob}[2][]{\Pr_{#1} \bigl[ #2 \bigr]}

\newcommand{\set}[1]{\{ #1 \}}
\newcommand{\Set}[1]{\bigl\{ #1 \bigr\}}

\newcommand{\eqperiod}{\enspace .}
\newcommand{\eqcomma}{\enspace ,}
\newcommand{\ie}{i.e.,\ }

\newcommand{\lbQ}{\kappa}

\newcommand{\disjointunion}{\overset{.}{\cup}}

\DeclareMathOperator{\poly}{poly}

\DeclareMathOperator*{\img}{img}

\DeclareMathOperator*{\vc}{vc}

\DeclareMathOperator{\size}{Size}

\newcommand{\rect}{Q}

\newcommand{\ind}{\mathbbm{1}}
\newcommand{\indic}[1]{\ind_{\set{#1 \text{ is a clique}}}}

\newcommand{\eps}{\varepsilon}

\newcommand{\etanew}{\eta}

\DeclareMathOperator{\core}{core}

\newcommand{\lexvc}{W}

\newcommand{\compatible}{large\xspace}

\newcommand{\calD}{\mathcal{D}}

\newcommand{\calF}{\mathcal{F}}
\newcommand{\calG}{\mathcal{G}}
\newcommand{\calH}{\mathcal{H}}

\newcommand{\calP}{\mathcal{P}}
\newcommand{\calQ}{\mathcal{Q}}

\newcommand{\calT}{\mathcal{T}}

\DeclareMathOperator*{\E}{\mathbb{E}}
\newcommand{\R}{\mathbb{R}}
\newcommand{\N}{\mathbb{N}}

\crefname{case}{Case}{Cases}
\Crefname{case}{Case}{Cases}

\crefname{property}{Pro\-per\-ty}{Pro\-per\-ties}
\Crefname{property}{Pro\-per\-ty}{Pro\-per\-ties}

\renewcommand{\refeq}[1]{(\ref{#1})}

\usepackage{thmtools,thm-restate}

\declaretheorem[%
sibling=theorem,
name=Theorem,
style=thmstyle
]%
{restatabletheorem}
\declaretheorem[%
sibling=theorem,
name=Lemma,
style=thmstyle
]%
{restatablelemma}

\makeatletter
\renewenvironment{thmt@restatable}[3][]{%
  \thmt@toks{}
  \stepcounter{thmt@dummyctr}
  \long\def\thmrst@store##1{%
    \@xa\gdef\csname #3\endcsname{%
      \@ifstar{%
        \thmt@thisistheonefalse\csname thmt@stored@#3\endcsname
      }{%
        \thmt@thisistheonetrue\csname thmt@stored@#3\endcsname
      }%
    }%
    \@xa\long\@xa\gdef\csname thmt@stored@#3\@xa\endcsname\@xa{%
      \begingroup
      \ifthmt@thisistheone
      \else
        \@xa\protected@edef\csname the#2\endcsname{%
          \thmt@trivialref{thmt@@#3}{??} (Restated)}%
        \ifcsname r@thmt@@#3\endcsname\else
          \G@refundefinedtrue
        \fi
        \@xa\let\csname c@#2\endcsname=\c@thmt@dummyctr
        \@xa\let\csname theH#2\endcsname=\theHthmt@dummyctr
        \let\label=\@gobble
        \let\ltx@label=\@gobble
        \def\thmt@restorecounters{}%
        \@for\thmt@ctr:=\thmt@innercounters\do{%
          \protected@edef\thmt@restorecounters{%
            \thmt@restorecounters
            \protect\setcounter{\thmt@ctr}{\arabic{\thmt@ctr}}%
          }%
        }%
        \thmt@trivialref{thmt@@#3@data}{}%
      \fi
      \ifthmt@restatethis
        \thmt@restatethisfalse
      \else
        \csname #2\@xa\endcsname\ifx\@nx#1\@nx\else[{#1}]\fi
      \fi
      \ifthmt@thisistheone
         \thmt@rst@storecounters{#3}%
        \label{thmt@@#3}%
      \fi
      ##1%
      \csname end#2\endcsname
      \ifthmt@thisistheone\else\thmt@restorecounters\fi
      \endgroup
    }
    \csname #3\@xa\endcsname\ifthmt@thisistheone\else*\fi
    \@xa\end\@xa{\@currenvir}
  }
  \thmt@collect@body\thmrst@store
}{%
}
\makeatother

\newcommand{\su}[1]{}
\newcommand{\ap}[1]{}
\newcommand{\kr}[1]{}

\usepackage{fancyhdr}

\title{Clique Is Hard on Average for Sherali-Adams\\ with Bounded
  Coefficients%
  \thanks{%
    This is the full-length version of a paper
    with the 
    title \emph{``Clique Is Hard on Average for Unary Sherali-Adams''}
    that appeared in the \emph{Proceedings of the 64th Annual IEEE
      Symposium on Foundations of Computer Science (FOCS '23)}.}}

\author{%
\makebox[.23\linewidth]{Susanna F. de Rezende}\\
\textsl{Lund University}%
\and
\makebox[.23\linewidth]{Aaron Potechin}\\
\textsl{University of Chicago}%
\and
\makebox[.23\linewidth]{Kilian Risse}\\
\textsl{EPFL}%
}

\date{\today}

\begin{document}
\maketitle

\begin{abstract}
  We prove that Sherali-Adams with polynomially bounded coefficients
  requires proofs of size $n^{\Omega(d)}$ to rule out the existence of
  an $n^{\Theta(1)}$-clique in Erd\H{o}s-R\'{e}nyi random graphs whose
  maximum clique is of size $d\leq 2\log n$. This lower bound is tight
  up to the multiplicative constant in the exponent.
  We obtain this result by introducing a technique inspired by
  pseudo-calibration which may be of independent interest. The
  technique involves defining a measure on monomials that precisely
  captures the contribution of a monomial to a refutation. This
  measure intuitively captures progress and should have further
  applications in proof complexity.
\end{abstract}


\pagestyle{fancy}
\fancyhead{}
\fancyfoot{}
\renewcommand{\headrulewidth}{0pt}
\renewcommand{\footrulewidth}{0pt}
\fancyfoot[C]{\sffamily\thepage}

\pagenumbering{roman}

\thispagestyle{empty}
\newpage

\setcounter{tocdepth}{2}{\sffamily\tableofcontents}
\newpage

\pagenumbering{arabic}
\setcounter{page}{1}

\newcommand{\NP}{\textbf{NP}}
\newcommand{\Pclass}{\textbf{P}}
\newcommand{\Wclass}{\textbf{W}[1]}

\section{Introduction}
The problem of identifying a maximum clique in a given graph, that is,
finding a fully connected subgraph of maximum size, is one of the
fundamental problems of theoretical computer science. Already
mentioned by Cook~\cite{Cook71ComplexityTheoremProving} in his seminal
paper introducing the theory of \NP-complete problems, it was one of
the first combinatorial problems proven \NP-hard by
Karp~\cite{Karp72Reducibility}. Building on the PCP theorem this
result was later strengthened to even rule out polynomial time
algorithms that approximate the maximum clique size within a factor
of~$n^{1-\eps}$~\cite{Hastad99Clique,Zuckerman07LinearDegreeExtractors},
unless $\Pclass = \NP$.

A related problem is $k$-clique: given an $n$-vertex graph, determine
whether it contains a clique of size $k$. This problem can be solved
in time $O(n^k)$ by iterating over all subsets of vertices of size $k$
and checking whether one of them is a clique. This naïve algorithm is
essentially the fastest known; a clever use of fast matrix
multiplication~\cite{NP85CplxSubgraph} allows a slight improvement
upon the constant in the exponent but no algorithms with a sublinear
dependence on $k$ in the exponent are known.

If we only assume $\Pclass \neq \NP$, it is unknown whether there are
faster algorithms for $k$-clique. However, improving upon the linear
dependence on $k$ in the exponent would disprove the exponential time
hypothesis~\cite{CHKX04LinearFPTreductions} and getting rid of the
dependence on $k$ 
altogether would imply that the class of fixed parameter tractable
problems collapses to
$\Wclass$~\cite{DF95FPTandCompletenessII}. Hence, if one is willing to
make the strong assumption that the exponential time hypothesis holds,
then the naïve algorithm has essentially optimal running time in the
worst-case.

Besides studying $k$-clique in the worst-case, 
one may consider it in the average-case setting. Suppose
the given graph is an Erd\H{o}s-Rényi graph with edge probability
around the threshold of containing a $k$-clique. Does $k$-clique
require time $n^{\Omega(k)}$ on such graphs? Or, even less
ambitiously, is there an algorithm running in time $n^{o(k)}$ that
decides the $n^\eps$-clique problem on such graphs?
It is unlikely that the hardness of such average-case questions can be
based on worst-case hardness assumptions such as
$\textbf{P} \neq \textbf{NP}$ or the exponential time
hypothesis~\cite{BT06}. They are, in fact, being used as hardness
assumptions themselves: the \emph{planted clique conjecture} states
that $n^{1/2-\eps}$-clique requires time $n^{\Omega(\log n)}$ on
Erd\H{o}s-Rényi graphs with edge probability~$1/2$.

In order to obtain \emph{unconditional} lower
bounds -- that do not rely on hardness assumptions --
for such average-case questions, we focus
on limited models of computation.
This approach has turned out to be quite fruitful and several results
of this form have emerged over the past few decades.
For Boolean circuits, Rossman~\cite{Rossman08, Rossman10} proved two
remarkable results: he showed that monotone circuits, i.e., circuits
consisting of $\vee$ and $\wedge$ gates only, as well as circuits of
constant depth require size $\Omega(n^{k/4})$ to refute the existence
of a $k$-clique in the average-case setting.

Instead of studying circuits, it is also possible to 
approach this problem through the
lens of proof complexity. Very broadly, proof complexity studies
certificates of unsatisfiability of propositional formulas. As we
cannot argue about certificates of unsatisfiability in general we
consider certificates of a certain form, or in terms of proof
complexity, refutations in a given proof system.
For instance, if we prove that any certificate in a proof system $P$
that witnesses that a given $n$-vertex graph contains no $k$-clique
requires length $n^{\Omega(k)}$ on average,
then we immediately obtain
average-case $n^{\Omega(k)}$ running time lower bounds for any
algorithm whose trace can be interpreted as a proof in the system $P$. 
It is often the case that state-of-the-art algorithms can be captured
by seemingly simple proof systems, as was shown to be the case
for clique algorithms~\cite{ABRLNR21}.

It is often the case that weak proof systems are sensitive to the 
precise encoding of combinatorial principles.
The $k$-clique formula is no exception: it is somewhat straightforward
to prove almost optimal $n^{\Omega(k)}$ resolution size lower bounds
for the less usual binary encoding of the $k$-clique formula
\cite{LPRT17ComplexityRamsey} and these lower bounds can even be
extended to an $n^{\Omega(k)}$ lower bound for the Res($s$) proof
system for constant $s$~\cite{DGGM20}.  For the more natural unary
encoding not much is known.  There are essentially optimal
$n^{\Omega(k)}$ average-case size lower bounds for regular resolution
\cite{ABRLNR21,Pang21} and tree-like resolution
\cite{BGLR12Parameterized,Lauria18}.  For resolution, there are
two 
average-case lower bounds that hold in different regimes: for
$n^{5/6} \ll k \le n/3$, Beame et al.~\cite{BIS07} proved an
average-case $\exp(n^{\Omega(1)})$ size lower bound and for
$k\leq n^{1/3}$, Pang~\cite{Pang21} proved a $2^{k^{1-o(1)}}$ lower
bound.
It is a long standing open problem, mentioned, e.g., in
~\cite{BGLR12Parameterized}, to prove an unconditional $n^{\Omega(k)}$
resolution size lower bound for the unary encoding -- even in the
worst case.

Little is known about the average-case hardness of the
$k$-clique formula in the semi-algebraic setting. There are optimal
degree lower bounds for $k \leq n^{1/2 - \eps}$ 
for the Sum-of-Squares proof system
\cite{MPW15SumOfSquaresPlantedClique, BHKKMP16clique, Pang21-sos}, but 
there are no non-trivial lower bounds on size.
For Nullstellensatz, however,
if restricted to not use dual variables,
then size lower bounds follow by a simple syntactic
argument~\cite{Margulies08Thesis}.  Prior to our work no other size
lower bounds were known for algebraic or semi-algebraic proof systems.

\subsection{Our Result}

In this work we establish that Sherali-Adams~\cite{AS94,DM13} with
polynomially bounded coefficients requires size $n^{\Omega(D)}$ to
refute the $n^{1/100}$-clique formula on random graphs whose maximum
clique size is of size $D \leq 2 \log n$. Qualitatively this
establishes the planted clique conjecture for Sherali-Adams with
polynomially bounded coefficients. This is the first size lower bound
on the clique formula for a semi-algebraic proof system.

\begin{theorem}[Informal]
  \label{thm:main}
  For all integers $n \in \N^+$ and $ D \le 2 \log n$, if
  $G \sim \calG(n,n^{-2/D})$ is an Erd\H{o}s-R\'enyi random graph,
  then it holds asymptotically almost surely that Sherali-Adams with
  polynomially bounded coefficients requires size at least
  $n^{\Omega(D)}$ to refute the claim that~$G$ contains a clique of
  size $k$, for any $k \leq n^{1/67}$.
\end{theorem}

Note that Sherali-Adams with polynomially bounded coefficients is
stronger than unary Sherali-Adams and is incomparable to
resolution~\cite{GHJMPRT22}. Our result further applies to the
SubCubeSums proof system~\cite{FMSV23MaxSAT} as
our proof strategy gives a lower bound on the sum of the magnitude of
the coefficients of a Sherali-Adams refutation, ignoring Boolean
axioms.

Let us stress that the size lower bound holds regardless of the degree
of the refutation.  This is a somewhat unique feature of our technique
-- all other lower bound strategies for Sherali-Adams and
Sum-of-Squares are tailored to proving degree lower bounds, which, if
strong enough, imply size lower bounds by the size-degree
relation~\cite{AH18SosTradeoff}.  Since the clique formula has
refutations of degree $D$ we cannot expect to obtain size lower bounds
through this connection for $D\leq \sqrt{n}$. We therefore introduce a
new technique, inspired by pseudo-calibration~\cite{BHKKMP16clique},
that is more refined -- for any monomial $m$, of arbitrary degree, we
determine a lower bound on the size of the smallest
bounded-coefficient Sherali-Adams derivation of $m$.

\subsection{Organization}

The rest of this paper is organized as follows.  In \cref{sec:prelim}
we introduce some basic terminology to then outline our proof strategy
in \cref{sec:lb-overview} where we also attempt to convey some
intuition. With the motivation at hand from \cref{sec:lb-overview} we
then go on to define the central combinatorial concept of a
\emph{core} of a graph in \cref{sec:cores-bounds} and a notion of
pseudorandomness in \cref{sec:random}.  We proceed in
\cref{sec:outline} to prove the main theorem for any graph satisfying
our notion of pseudorandomness, but postpone the proof of one of the
main lemmas to \cref{sec:good-rect}.  Finally, in
\cref{sec:random-proof}, we show that a random graph satisfies our
pseudorandomness property and conclude with some open problems in
\cref{sec:conclusion}.


\section{Preliminaries}
\label{sec:prelim}

Natural logarithms (base $\mathrm{e}$) are denoted by $\ln$, whereas
base $2$ logarithms are denoted by $\log$.  For integers $n \ge 1$ we
introduce the shorthand $[n] = \set{1, 2, \ldots, n}$ and sometimes
identify singletons $\set{u}$ with the element $u$. Let
$\binom{S}{\ell}$ denote the set of subsets of $S$ of size $\ell$ and,
for a given a random variable $X$ and an event $P$, we denote by
$\ind_{P}(X)$ the indicator random variable that is $1$ if $P$ holds
and $0$ otherwise.

\subsection{Semantic Sherali-Adams}
\label{sec:SA}

Let $\calP = \set{p_1 = 0, \ldots, p_m = 0}$ be a polynomial system of
equations over Boolean variables $x_1, \ldots, x_n$ and their twin
variables $\bar{x}_1, \ldots, \bar{x}_n$. Denote by
\begin{align}
  B_n =
  \set{x_i(1-x_i) \mid i \in [n]}
  \cup
  \set{\bar{x}_i(1-\bar{x}_i) \mid i \in [n]}
  \cup
  \set{1 - \bar{x}_i - {x}_i \mid i \in [n]}
\end{align}
the corresponding Boolean axioms and negation axioms and let $I_{B_n}$
denote the ideal generated by $B_n$, that is, $I_{B_n}$ consists of
all polynomials of the form $\sum_{j} r_j q_j$ where the $r_j$ are
arbitrary polynomials and $q_j \in B_n$. For an ideal $I$ and
polynomials $p$ and $q$ we write $p \equiv q \mod I$ if $p-q \in I$.

A \emph{semantic Sherali-Adams refutation of $\calP$} is a sequence of
polynomials $(g_1, \ldots, g_m, f_0)$ such that $f_0$ is of the form
\begin{equation}
  f_0 =
  \sum_{\substack{A,B\subseteq [n]\\\alpha_{A,B} \ge 0}}
  \alpha_{A,B} \prod_{i\in A}x_i \prod_{i\in B}\bar{x}_i
\end{equation}
and it holds that
\begin{equation}\label{eq:SArefutation}
  \sum_{j\in [m]}g_jp_j + f_0 \equiv -1 \mod I_{B_n} \eqperiod
\end{equation}
The \emph{size} of a refutation $\pi$, denoted by $\size(\pi)$, is the
sum of the size of the binary encodings of the non-zero coefficients
on the left hand side of~\refeq{eq:SArefutation} when all polynomials
are expanded as a sum of monomials (without cancellations), while the
\emph{coefficient size} of $\pi$ is the sum of the magnitudes of the
coefficients of the aforementioned monomials.

We note that this definition differs from that of the usual Sherali-Adams proof 
system~\cite{AS94,DM13} where Boolean axioms and negation axioms
are written out explicitly and the size is measured taking also these axioms
into account. We are not
the first to disregard the size contribution of these axioms:
most size lower bounds for Sherali-Adams also apply in this setting.
The question of whether these two size measures are polynomially related 
was raised explicitly in~\cite{FHRSV24}, where the above system is referred
to as \emph{succinct} Sherali-Adams, and remains open.

To verify that semantic Sherali-Adams is a Cook-Reckhow proof system, 
we need to check that a semantic Sherali-Adams refutation is verifiable in 
time polynomial in the size of the refutation. This was originally shown
in~\cite{FHRSV24}, but
we provide a simpler and more direct proof of this fact.

\begin{proposition}
  Let $\calP$ be a polynomial system of equations over $n$ Boolean
  variables and their twin variables. A semantic Sherali-Adams refutation $\pi$
  of $\calP$ can be verified in time
  $O\bigl(\size(\pi)^2 \cdot \poly(n)\bigr)$.
\end{proposition}

\begin{proof}
  Denote the coefficient of a monomial $m$ in a polynomial
  $p \in \calP$ by $\gamma_{p}(m)$. Further, given a monomial $m$, let
  $\ind_m$ be the vector of length $2^n$, indexed by assignments
  $\rho \in \set{0, 1}^n$, that corresponds to the truth table of $m$:
  the entry with index $\rho$ is~$1$ if the monomial $m$ evaluates
  to~$1$ under $\rho$ and~$0$ otherwise. We denote the all-ones vector
  $\ind_1$ by $\ind$, and the all-zero vector $\ind_0$ by $\varmathbb{0}$.

  With this notation at hand we may equivalently define a semantic
  Sherali-Adams refutation of~$\calP$ as a sequence of rational
  vectors $(\beta_1, \ldots, \beta_{|\calP|}, \alpha)$ indexed by
  multilinear monomials with $\alpha \geq 0$ and
  \begin{align}
    \sum_{p\in \calP}
    \sum_{m_1,m_2} 
    \beta_{p}(m_1) \cdot \gamma_{p}(m_2) \cdot \ind_{m_1 \cdot m_2} - 
    &\sum_{m} \alpha(m) \cdot \ind_m - \ind =  \varmathbb{0} \eqperiod
  \end{align}
  View the left hand side as a vector $\pi$ indexed by assignments
  $\rho \in \set{0,1}^n$. Since the vector $\pi$ is rational,
  we have that $\pi$ is the $0$-vector if and
  only if the inner product is $0$, that is,
  $\langle \pi, \pi \rangle = 0$. This can be efficiently checked by
  expanding the inner product
  \begin{align}
    \begin{split}
      \langle \pi, \pi \rangle
      &=
        \Big(
        \sum_{p\in \calP} \sum_{\substack{m_1,m_2}} 
        \beta_{p}(m_1) \cdot
        \gamma_{p}(m_2) \cdot
        \ind_{m_1 \cdot m_2}
        \Big)^2
        +
        \Big(\sum_{m} \alpha(m) \cdot \ind_m\Big)^2
        +
        \langle\ind,\ind\rangle\\
      &\quad{}-
        2
        \sum_{p \in \calP}
        \sum_{m_1,m_2}
        \sum_m
        \alpha(m) \cdot
        \beta_{p}(m_1) \cdot
        \gamma_{p}(m_2)
        \cdot
        \langle\ind_{m_1\cdot m_2},\ind_{m} \rangle\\
      &\quad{}-
        2
        \sum_{p \in \calP}
        \sum_{m_1,m_2}
        \beta_{p}(m_1)\cdot
        \gamma_{p}(m_2)
        \cdot
        \langle\ind_{m_1\cdot m_2},\ind \rangle
        +
        2
        \sum_m
        \alpha(m)
        \cdot
        \langle\ind_{m},\ind \rangle \eqperiod
    \end{split}
  \end{align}
  Observe that in the above expression all the inner products can be
  efficiently computed. Hence checking whether $\pi$ is indeed a valid semantic
  Sherali-Adams refutation boils down to verifying that the above sum
  of bounded rationals is equal to $0$.
\end{proof}

As the distinction between semantic Sherali-Adams and Sherali-Adams is
not essential for what follows, we refer to semantic Sherali-Adams
simply as Sherali-Adams going forward.
A Sherali-Adams refutation $\pi$ of $\calP$ is a \emph{Sherali-Adams
  refutation with $f(n)$-bounded-coefficients} if the magnitude of all
coefficients is bounded by $f(n)$ and we call $\pi$ a
\emph{Sherali-Adams refutation with $\poly$-bounded-coefficients} if
for some constant $c > 0$ it holds that $\pi$ is a Sherali-Adams
refutation with $n^c$-bounded-coefficients. Let us record the
following observation.

\begin{proposition}\label{prop:bounded-coeff}
  If Sherali-Adams requires coefficient size $s$ to refute $\calP$,
  then Sherali-Adams with $f(n)$-bounded-coefficients requires size
  $s/f(n)$ to refute $\calP$.
\end{proposition}
\begin{proof}
  As every coefficient in an $f(n)$-bounded-coefficient refutation is
  bounded by $f(n)$, there need to be at least $s/f(n)$ monomials with
  a non-zero coefficient.
\end{proof}

\emph{Unary Sherali-Adams} is a subsystem of Sherali-Adams where all
coefficients of monomials are either $+1$ or $-1$ and the
right-hand-side of~\cref{eq:SArefutation} is any negative integer
\begin{equation}\label{eq:uSArefutation}
  \sum_{j\in [m]}g_jp_j + f_0 \equiv -M \mod I_{B_n}\eqcomma
\end{equation}
where $f_0$ is again a non-negative sum of monomials.

\begin{proposition}\label{prop:coeff-unary}
  If Sherali-Adams requires coefficient size $s$ to refute
  $\calP$, then unary Sherali-Adams requires size at least
  $s$ to refute $\calP$.
\end{proposition}
\begin{proof}
  We can transform any unary Sherali-Adams refutation of size $s$,
  summing to an integer $-M$, to a Sherali-Adams refutation of
  coefficient size at most $s$ by dividing the left hand side by
  $M \ge 1$.
\end{proof}

\subsection{Graph Theory}

Before defining the $k$-clique formula, 
we introduce some terminology and notation that we use
throughout the paper.

Unless stated otherwise, $G$
denotes a $k$-partite graph
with partitions $V_1, \ldots, V_k$ of size $n$ each. We call a
partition $V_i$ a \emph{block} and, for $S \subseteq [k]$, denote by
$V_S$ the vertices in blocks in $S$, that is,
$V_S = \bigcup_{i \in S} V_i$. For disjoint sets $W_1, \ldots, W_s$ we
let a \emph{tuple} $t= (w_1, \ldots, w_s)$ be a sequence of vertices
satisfying $w_i \in W_i$ for all $i \in [s]$. All tuples we consider
are defined with respect to the partition $V_1, \ldots, V_k$, though,
at times, may only be defined over a subset of the blocks, that is,
not all tuples are of size $k$. For a tuple $t = (v_1, \ldots, v_k)$
and a set $S \subseteq [k]$ we denote the projection of $t$ onto $S$
by $t_S = (v_i\,\mid\,i \in S)$. An \emph{$s$-tuple} is a tuple of
size $s$ and sometimes it is convenient for us to think of a tuple as
a set of vertices. We take the liberty to interchangeably identify a
tuple as a sequence as well as a set and hope that this causes no
confusion.

A set $Q$ of tuples is a \emph{rectangle} if for some $S\subseteq [k]$
the set $Q$ can be written as the
Cartesian product of sets $U_i \subseteq V_i$ for
$i \in S$, i.e., $Q = \bigtimes_{i \in S} U_i$;
in other words, 
$Q$ contains
all tuples $t = (u_i\,\mid\,i \in S)$ satisfying $u_i \in U_i$ for all
$i \in S$. Rectangles, unless explicitly stated, consist of $k$-tuples
only, that is, $Q = \bigtimes_{i \in [k]} U_i$. 
Given a rectangle $Q$ and a set
$S \subseteq [k]$ we let $Q_S$ be the projection of $Q$ onto the
blocks in $S$: if $Q = \bigtimes_{i \in [k]} U_i$, then
$Q_S = \bigtimes_{i \in S} U_i$ and, in particular, we have
$Q_i = U_i$ for $i \in [k]$.

While $G$ always denotes a large graph, the graphs $H$ and $F$ denote
small graphs: throughout the paper $H$ and $F$ are graphs on $k$
labeled vertices. Usually these graphs have a small vertex cover and
graphs denoted by $F$ furthermore have many isolated vertices. For a
graph $H$ we denote the minimum vertex cover by $\vc(H)$ and sometimes
refer to $H$ as a \emph{pattern graph}, whereas $F$ is usually a
\emph{core graph} (see \cref{sec:cores-bounds}). We denote by $\calH$
the set of graphs on $k$ labeled vertices and for a parameter
$i\in \N^+$ let $\calH_i \subseteq \calH$ be the family of graphs with
a minimum vertex cover of size at most $i$, that is, all graphs
$H \in \calH_i$ satisfy $\vc(H) \le i$.

We record here two simple lemmas about graphs that will come in handy 
throughout the paper.

\begin{lemma}
  \label{lem:count-H}
  There are at most $2^{c \log k + b(c -(b+1)/2)} \le 2^{c(b+\log k)}$
  graphs $H$ over $k$ vertices with a vertex cover of size $b$ and
  $\bigl\lvert V\bigl(E(H)\bigr)\bigr\rvert \le c$.
\end{lemma}

\begin{proof}
  We first choose the $b$ vertices from the $k$ vertices that form the
  vertex cover. Then, from the remaining $k-b$ vertices, we choose
  $c-b$ vertices that may be incident to an edge. We can add edges
  that are incident to the vertex cover and the other $c-b$ vertices
  and thus get that there are at most
  \begin{align}
    \binom{k}{b}
    \binom{k-b}{c-b}
    2^{\binom{b}{2}}
    2^{b(c-b)}
    \le 2^{ c\log k + b(c - (b+1)/2)}
  \end{align}
  many such graphs.
\end{proof}

Recall that a maximal matching of $H$ is a matching that cannot be
extended in $H$.

\begin{proposition}
  \label{clm:match}
  Any maximal matching in a graph $H$ is of size at least
  $\lceil \vc(H)/2\rceil$.
\end{proposition}

\begin{proof}
  Since $M$ is maximal, all edges of $H$ are incident to $V(M)$. Thus
  the set $V(M)$ is a vertex cover of $H$.
\end{proof}

The distribution of random graphs we consider in this paper
is a $k$-partite version of the Erd\H{o}s-Rényi distribution.
For a fixed set $V$ of $n$ vertices and a real number $0 \le p \le 1$, 
the Erd\H{o}s-Rényi distribution $\calG(n,p)$ 
is the distribution of random graphs on vertex set $V$ where every potential
edge $e = \set{u,v}$, for vertices $u\neq v$, 
is sampled independently with probability~$p$.
As was done in~\cite{BIS07,ABRLNR21}, we work with the block model,
which is defined as follows. 
Given $k$ blocks $V_1, \ldots, V_k$ of size $n$ and a real number
$0 \le p \le 1$, we denote by $\calG(n,k,p)$ the distribution over
graphs on the vertex set $V_{[k]}$ defined by sampling each edge
$e = \set{u,v}$ independently with probability $p$ if $u$ and $v$ are
in distinct blocks. We sometimes refer to such pairs of vertices
$\set{u,v}$ as \emph{potential edges}. Edges within the same block are
never included and hence $\calG(n,k,p)$ is a distribution over
$k$-partite graphs.

\subsection{Clique Formula}
\label{sec:formula}
Below we present an encoding of the $k$-clique formula on $k$-partite graphs
as a system of polynomial equations. 

Given a $k$-partite graph $G$ with blocks $V_1, \ldots, V_k$ of size
$n$ we define the \emph{$k$-clique formula over~$G$} stating that $G$ has
a $k$-clique (with one vertex per block) as follows. The
formula is defined over $2kn$ variables: each vertex $v \in V_{[k]}$
is associated with two variables $x_v$ and $\bar x_v$. 
The intended meaning is that $x_v$ is $1$ if $v$ is in the identified $k$-clique
and $0$ otherwise, and $\bar x_v = 1 - x_v$.
Since we define the Sherali-Adams proof system (see \cref{sec:SA}) 
to operate modulo the ideal generated by the Boolean axioms $y(1-y)$
and the negation axioms $1 - \bar x_v - x_v$, we do not need to explicitly include
such axioms in the formula.

For each block $V_i$ we introduce the \emph{block axiom}
$\sum_{v \in V_i} x_v - 1 = 0$ stating that precisely one vertex from each
block is chosen and for each pair of vertices
$\set{u, v} \not\in E(G)$ in distinct blocks we introduce the
\emph{edge axiom} $x_ux_v = 0$ that ensures that non-neighbors are not
simultaneously selected. We note that we could also include edge
axioms for pairs of vertices in the same block but we choose not to since
these are easily derivable from the block axioms.

It should be evident that this formula is satisfiable modulo the Boolean
and the negation axioms if and only if
there is a $k$-tuple $t$ such that the vertex induced subgraph $G[t]$ is a
clique.

We make two remarks about this choice of encoding.
The first is that our lower bound strategy is completely agnostic to
the encoding of the block axioms. 
We could equally well have considered the polynomial translation
of the CNF encoding, or the binary encoding, or any 
encoding of the constraints that one vertex of each block should 
be chosen to be part of the clique.

The second remark is that, as argued in Beame, Impagliazzo and Sabharwal~\cite{BIS07}, 
the reason for choosing to define the formula over $k$-partite graphs
is that proving a lower
bound for this encoding implies a lower bound for other natural encodings. 
In particular, if we consider the $k$-clique formula defined 
over a (not necessarily $k$-partite) graph $G=(V,E)$ on $kn$ vertices,
as the formula containing the axiom $\sum_{v\in V} x_v = k$ and 
edge axioms $x_ux_v = 0$ for each pair of vertices
$\set{u, v} \not\in E$, then 
for any equal-sized $k$-partition $V = V_1\,\dot\cup\,\cdots\,\dot\cup\,V_k$
it holds that
a lower bound for the $k$-clique formula on $G$ with
this partition implies the same lower bound for the non-$k$-partite
formula.

\begin{proposition}[\cite{BIS07}]\label{lem:switch-distibutions}
  Let $k,n \in \N^+$ be integer and let $G$ be a graph on $kn$
  vertices. Then the minimum Sherali-Adams coefficient size to refute
  the $k$-clique formula over $G$ is bounded from below by the
  coefficient size required to refute the $k$-clique formula defined
  with respect to any equal-sized $k$-partition of $G$.
\end{proposition}

This proposition was proven in \cite{BIS07} for resolution size,
and it is straightforward to see that it holds
for Sherali-Adams coefficient size.
Indeed, it is enough to observe that the non-$k$-partite $k$-clique
formula can be easily derived from the $k$-partite $k$-clique formula.

\section{Main Theorem and Proof Overview}
\label{sec:lb-overview}

The main result of this paper is a tight, up to constants in the
exponent, Sherali-Adams coefficient size lower bound for $k$-clique
formulas over Erd\H{o}s-R\'enyi random
graphs. 
\begin{theorem}[Main theorem]
  \label{thm:main-formal}
  Let $k$ and $D$ be functions of $n$ such that $D \le 2 \log n$ and
  $k \leq n^{1/66}$. If ${G\sim \calG(n, k, n^{-2/D})}$, then
  asymptotically almost surely Sherali-Adams requires coefficient size
  $n^{\Omega(D)}$ to refute the $k$-clique formula over $G$.
\end{theorem}
Note that \cref{thm:main} follows directly from \cref{thm:main-formal}
along with \cref{lem:switch-distibutions,prop:bounded-coeff}, where we
assume that the coefficient size is bounded by a polynomial $n^c$
where $c$ is a constant independent of $D$. The main result of the
conference version of this paper~\cite{DPR23} follows from
\cref{thm:main-formal} and
\cref{lem:switch-distibutions,prop:coeff-unary}.

In the rest of this section we outline our proof strategy. We intend
to come up with a so-called \emph{pseudo-measure} which lower bounds
the coefficient size of a Sherali-Adams refutation. Before we get
ahead of ourselves let us define what a pseudo-measure is. 

\begin{definition}[Pseudo-measure]\label{def:pseudo-measure}
  Let $\delta > 0$ and $\calP$ be a set of polynomials over the
  polynomial ring $\R[x_1, \ldots, x_n, \bar x_1, \ldots, \bar x_n]$.
  A linear function
  $\mu\colon \R[x_1, \ldots, x_n, \bar x_1, \ldots, \bar x_n] \rightarrow
  \R$, mapping polynomials to reals, is a
  \emph{$\delta$-pseudo-measure for $\calP$} if for all monomials $m$
  and all polynomials $p \in \calP$ it holds that
  \begin{enumerate}
  \item $\mu(1) = 1$, \label[property]{it:one}
  \item $|\mu(m \cdot p)| \leq \delta$, and \label[property]{it:axiom}
  \item $\mu(m)\geq -\delta$. \label[property]{it:non-neg}
  \end{enumerate}
\end{definition}

A concept related to the notion of a pseudo-measure has previously
appeared in \cite{PZ22php-ns} for the Nullstellensatz proof system
over the reals. We also note that in a recent work Hub\'a\v{c}ek,
Khaniki and Thapen~\cite{HKT24} defined a similar notion for
Sherali-Adams with bounded degree.  We have the following simple
proposition.

\begin{proposition}
  \label{prop:pseudo}
  There is a $\delta$-pseudo-measure for $\calP$ if and only if any
  Sherali-Adams refutation of $\calP$ requires coefficient size
  $1/\delta$.
\end{proposition}

\begin{proof}
  Given a monomial $m$, let $\ind_m$ be the vector
  of length $2^n$, indexed by assignments $\rho \in \set{0, 1}^n$,
  that corresponds to the truth table of $m$, 
  that is, on index $\rho$ the vector $\ind_m$ is~$1$ if 
  the monomial $m$ evaluates to~$1$ 
  on assignment $\rho$ and~$0$ otherwise. 
  We denote the all-ones vector $\ind_1$ by $\ind$.

  Let $\calP$ be a polynomial system of equations over Boolean variables
  $x_1, \ldots, x_n$ and their twin variables
  $\bar{x}_1, \ldots, \bar{x}_n$. We denote the coefficient of a
  monomial $m$ in a polynomial $p \in \calP$ by $\gamma_p(m)$.
  We may write a linear program that searches for a Sherali-Adams
  refutation of $\calP$ of minimum coefficient size as
  \begin{equation}
    \begin{array}{ll@{}ll}
      \text{minimize}
      & \displaystyle \sum_{m} \alpha({m})
        \hspace{5pt} +
        \sum_{\substack{p\in \calP}} \sum_{\substack{m}}
        \big(\beta^+_{p}(m) + \beta^-_{p}(m)\big)\\
      \text{subject to}
      & \displaystyle \sum_{p\in \calP} \sum_{\substack{m_1,m_2}} 
        \big(\beta^+_{p}(m_1) - \beta^-_{p}(m_1)\big)
        \cdot \gamma_{p}(m_2) \cdot \ind_{m_1 \cdot m_2} - 
      &\displaystyle\sum_{m} \alpha(m) \cdot \ind_m =  \ind
      \\
      &\displaystyle\alpha
        , \displaystyle\beta^+
        , \displaystyle\beta^-
        \geq 0 \eqperiod
    \end{array}
  \end{equation}
  We appeal to duality to obtain the following dual program to the
  above. Denote by $\mu$ the vector of dual variables of dimension
  $2^n$ and introduce for a monomial $m$ the shorthand
  $\mu(m) = \mu^T \ind_m$ and for a polynomial
  $q = \sum_{m}\gamma(m) \cdot m$ let
  $\mu(q) = \sum_{m} \gamma(m) \cdot \mu(m)$. The dual to the above linear
  program may be expressed as
  \begin{equation}
    \begin{array}{ll@{}ll}
      \text{maximize}
      & \displaystyle \mu(1)\\[10pt]
      \text{subject to}
      &\displaystyle
        \mu(m) \geq - 1
      & \forall m\\[7pt]
      & \displaystyle 
        \left| \mu(m\cdot p) \right| \leq 1 
        \hspace{10pt}
          & \forall 
            m, \forall p\in \calP \eqperiod
    \end{array}
  \end{equation}
  Thus the notion of a pseudo-measure for $\calP$ is indeed the dual
  object of a Sherali-Adams refutation of $\calP$ with small
  coefficient size. By appealing to strong duality and normalizing by
  $\mu(1)$ the claim follows.
\end{proof}

\subsection{Our Pseudo-Measure}
\label{sec:pseudo-measure}

In what follows we define our pseudo-measure $\mu$ for the $k$-clique
formula. We may think of $\mu$ as a progress measure: it assigns to
each monomial a real value which can be thought of as the contribution
of this monomial towards the refutation of the $k$-clique
formula. Thus, intuitively, we would like to associate each monomial
with the fraction of potentially satisfying assignments that it rules
out. In order to define this a bit more formally, let us introduce the
set of potentially satisfying assignments.

We say that an assignment $\alpha$ is \emph{potentially satisfying}
for the $k$-clique formula if there is a graph $G$ such that the
$k$-clique formula defined over $G$ is satisfied by $\alpha$. This set
of assignments can be easily characterized: if we associate each
$k$-tuple $t$ with the assignment $\rho_t$ that sets all variables
$x_u$ to $1$ if $u \in t$ and to $0$ otherwise, then the set of
potentially satisfying assignments of the $k$-clique formula is
\begin{align}
  \set{
  \rho_t
  \mid
  t
  \in
  V_1 \times V_2 \times \cdots \times V_k
  }
  \eqperiod
\end{align}

We say that a monomial $m$ \emph{rules out} an assignment $\rho$ if
$\rho(m) = 1$. As there is a one-to-one correspondence between
potentially satisfying assignments and tuples, it is convenient to
think of the tuples that a monomial rules out. We thus associate each
monomial $m$ with the set
\begin{align}
  Q(m) = \set{t \mid \rho_t(m) = 1}
\end{align}
of ruled out $k$-tuples. Note that $Q(1)$ is the set of all tuples,
that is, $Q(1) = V_1 \times V_2 \times \cdots \times V_k$ and the set
$Q(x_ux_v)$ associated with an edge axiom $x_ux_v$ consists of all
$k$-tuples that contain the vertices $u$ and $v$.

More generally, it is not too hard to see that the set of ruled out
tuples of a monomial is a rectangle and that for each rectangle $Q$
there is at least one monomial $m$ such that $Q$ is the set of tuples
ruled out by $m$. We thus often discuss rectangles and it is
implicitly understood that if a statement holds for all rectangles,
then it also holds for all monomials. Finally, observe that if a
monomial $m$ satisfies $m = m_1\cdot m_2$, then
$Q(m) \subseteq Q(m_1)$.

For intuition we will now discuss two na\"ive, and fatally flawed,
attempts to define a pseudo-measure.
For our first attempt,
we simply associate each monomial with the
fraction of ruled out tuples, that is we map a monomial $m$ to
\begin{align}
  \frac{|Q(m)|}{|Q(1)|} = n^{-k} \cdot |Q(m)| \eqperiod
\end{align}
This measure maps the monomial $1$ to $1$ and is clearly non-negative
and hence satisfies \cref{it:one,it:non-neg} of
\cref{def:pseudo-measure} for any $\delta > 0$. Furthermore, again for
any $\delta > 0$, it satisfies \cref{it:axiom} of
\cref{def:pseudo-measure} for the block axioms. Only the edge axioms
cause trouble: the rectangle $Q(x_ux_v)$ associated with the edge
axiom $x_ux_v$ is a $n^{-2}$ fraction of all tuples. As such,
according to \cref{prop:pseudo}, this pseudo-measure may only gives us
an $n^2$ coefficient size lower bound---not quite what we are after.

We may try to remedy this by not associating a monomial $m$ with
\emph{all} tuples in $Q(m)$ but rather only with a subset of $Q(m)$
that depends on the graph $G$. One very naïve attempt would be to
associate $m$ with the number of $k$-cliques that it rules out, that
is, we may associate a monomial $m$ with the normalized measure
\begin{align}\label{eq:def1-mu}
  \tilde{\mu}(m) = n^{-k}\sum_{t \in Q(m)} 2^{\binom{k}{2}} \indic{t}(G) \eqperiod
\end{align}
This definition, at least \emph{in expectation} over
$G\sim\calG(n,k,1/2)$, satisfies all properties of a pseudo-measure:
the monomial $1$ is mapped to $1$, the axioms are all mapped to $0$
and the measure is non-negative. 

The obvious problem is that all graphs we consider do \emph{not}
contain a $k$-clique and hence everything (including the monomial $1$)
is mapped to $0$. Put differently, the main problem is that the random
variable $\tilde{\mu}(1)$ over $G \sim \calG(n,k,1/2)$ has too large
variance. We reduce this variance as pioneered by Barak et
al.~\cite{BHKKMP16clique}: we expand \cref{eq:def1-mu} in the Fourier
basis and truncate the resulting expression. A careful choice of the
truncation, along with some significant effort, allows us to argue
that, on the one hand, the measure associated with $1$ is large (\ie
the variance is reduced) while, on the other hand, the measure
associated with edge axioms is still tightly concentrated around $0$
and that the measure is essentially non-negative (\ie the variance did
not increase significantly). This measure thus constitutes a valid
pseudo-measure as defined in \cref{def:pseudo-measure} up to
normalization.
In order to state the precise definition of our pseudo-measure $\mu$
we need some notation.

If $p$ denotes the probability that a potential edge $e$ is present in
the graph, then the character $\chi_e(G)$ is defined by
\begin{align}
  \chi_e(G) =
  \begin{cases}
    \frac{1-p}{p} & \text{if } e  \in E(G)\\
    -1 &\text{otherwise;}
  \end{cases}
\end{align}
and for a set of potential edges $E$ we let
$\chi_E(G) = \prod_{e \in E} \chi_e(G)$.  It is convenient for us to
work with the above (non-standard) basis as it allows us to easily
cancel characters in case an edge is missing.  Observe that for
$p=1/2$ this is the usual $\pm 1$ Fourier basis. First time readers
are advised to keep this case in mind for the remainder of the
article.
Let us record some useful facts.

\begin{proposition}
  Let $k, \ell, n \in \N^+$, let $p > 0$ and $G \sim
  \calG(n,k,p)$. For any potential edge $e$, sampled with probability
  $p$, it holds that $\E_G[\chi_e(G)] = 0$ and
  $\E_G[\chi^{2\ell}_e(G)] = (1-p)(1+(\frac{1-p}{p})^{2\ell-1})$. In
  particular, for $\ell=1$, we have $\E_G[\chi^2_e(G)] = (1-p)/p$. A
  useful bound is $\E_G[\chi^{2\ell}_e(G)] \le p^{-2\ell}$. Finally,
  observe that for any tuple $t$ we have that
  $\sum_{E \subseteq \binom{t}{2}} \chi_E(G)$ is $p^{-\binom{k}{2}}$
  if $t$ is a clique and $0$ otherwise.
\end{proposition}

To concisely state our pseudo-measure we need some further
notation. We consider sums of tuples and want to treat edge sets that
are equal up to the mapping onto a $k$-tuple as the same. More
precisely, if we have two $k$-tuples
$t=(v_1, \ldots, v_k),t'=(v_1', \ldots, v_k')$ and edge sets
$E \subseteq \binom{t}{2}$ and $E' \subseteq \binom{t'}{2}$ such that
$\set{v_i,v_j} \in E$ if and only if $\set{v_i',v_j'} \in E'$, then we
want to identify $E$ and $E'$ as the same edge set. To this end we
consider \emph{pattern graphs $H$} (similar to the \emph{shape graphs}
in the terminology of \cite{BHKKMP16clique}) over the vertex set
$[k]$. For a tuple $t=(v_1, \ldots, v_k)$ and a pattern graph $H$ we
let $H(t)$ be the edge set that contains the edge $\set{v_i, v_j}$ if
and only if the edge $\set{i,j}$ is present in $H$. See
\cref{fig:H-to-G} for an illustration.
\begin{figure}
  \centering
  \includegraphics{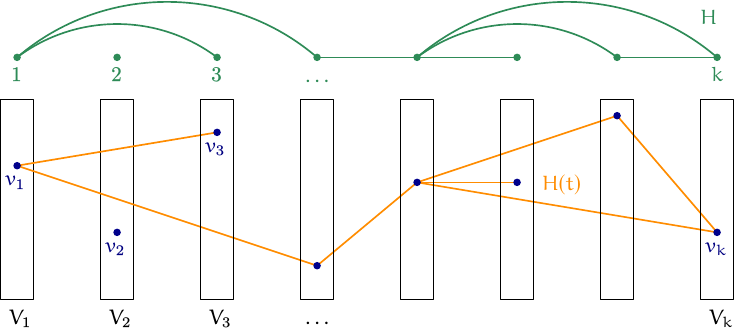}
  \caption{A pattern graph $H$ mapped onto a tuple $t=(v_1, \ldots, v_k)$}
  \label{fig:H-to-G}
\end{figure}

With this notation at hand we define our pseudo-measure as
\begin{align}\label{eq:def-mu}
  \mu(m) = \mu_d\bigl(Q(m)\bigr) = 
  n^{-k}
  \sum_{t \in Q(m)}
  \sum_{\substack{H\\ \vc(H)\le d}}
  \chi_{H(t)}(G) \eqcomma
\end{align}
where the second sum is over all graphs $H$ over $[k]$ vertices with
vertex cover at most $d$, and $d = \etanew D$ is a small constant
$\etanew > 0$ times the maximum clique size of $G$. 

Observe that Boolean axioms, the negation axioms and the block axioms
multiplied by an arbitrary monomial are all mapped to $0$ by
$\mu$. Hence it remains to prove that the measure $\mu$ maps the
constant $1$ monomial to a large value, that $\mu$ is small on
subrectangles of edge axioms, \ie any edge axiom multiplied by a
monomial is mapped to a small value, and that all monomials are
mapped to an approximately non-negative value.

By inspecting the second moment of $\mu(1)$ it is not too hard to see
that there is quite a bit of freedom on how to choose the truncation
in the definition of $\mu$ while maintaining the property that
$\mu(1) = 1 \pm n^{-\Omega(1)}$ asymptotically almost surely. However,
ensuring that the edge axioms are associated with small measure is
more delicate. Here we heavily rely on our choice to truncate
according to the minimum vertex cover. More specifically we rely on
two crucial properties of graphs $H$ satisfying $\vc(H) = d$: firstly,
we use the fact that such graphs contain a matching of size
$\lceil d/2 \rceil$ (see \cref{clm:match}) and, secondly, that the
family of these graphs satisfies a monotonicity property which leads
to a useful partition of this family. For more details about this
partition we refer to \cref{sec:cores-bounds}.
Let us mention that it is conceivable that one could 
increase the bound on $k$ for which our results hold by truncating
according to the size of the maximum matching. As we do not know how
to define the above mentioned partition with respect to the maximum
matching we truncate according to the minimum vertex cover.

In the following sections we try to present some intuition as to why
$\mu_d$ is a pseudo-measure up to normalization, that is, why it
satisfies \cref{def:pseudo-measure} where \cref{it:one} is relaxed to
approximately $1$. In \cref{sec:whole-space} we verify
that $G$ sampled from $\calG(n, k, 1/2)$ asymptotically
almost surely satisfies
$\mu(1) = \mu_d(\bigtimes_{i\in [k]} V_i) = 1 \pm n^{-\Omega(1)}$. As
mentioned, this follows by a straightforward second moment argument.
In \cref{sec:axioms-sketch} we outline why any subrectangle $Q$
of an edge axiom satisfies $|\mu_d(Q)| \le n^{-\Omega(d)}$. This proof
motivates the definitions in
\cref{sec:cores-bounds,sec:random}. Finally, in
\cref{sec:nonneg-sketch}, we provide some high-level overview of how
to prove that any rectangle $Q$ is mapped to an approximately
non-negative value, that is, it holds that
$\mu_d(Q) \geq - n^{-\Omega(d)}$. This is the most technically
challenging part of the paper.

\subsection{Expected Behavior of Our Pseudo-Measure}
\label{sec:whole-space}

The measure $\mu_d(Q)$ of any rectangle $Q$ satisfies 
\begin{align}
  \E_{G}[\mu_d(Q)]
  =
  n^{-k}
  \sum_{t \in Q}
  \sum_{H \in \calH_d}
  \E_G[\chi_{H(t)}(G)]
  =
  n^{-k}
  \sum_{t \in Q}
  \E_G[\chi_{\emptyset(t)}(G)]
  = n^{-k}|Q| \eqperiod
\end{align}
In particular, as $Q(1) = \bigtimes_{i \in [k]} V_i$, it holds that
$\E_G[\mu(1)] = 1$. In what follows we show that, for $p = 1/2$, the
measure is somewhat concentrated around the expected value. The
concentration, though, is far from enough to perform a union bound
over all rectangles to argue that the measure behaves as expected on
all rectangles simultaneously.

We show that the measure concentrates by an application of Chebyshev's
inequality. To this end we analyze the second moment: for $p=1/2$ we
have
\begin{align}
  \E_G[\mu_d^2(Q)] 
    &=
      n^{-2k} \cdot
      \sum_{H \in \calH_d} \sum_{t,t' \in Q}
      \E_G[\chi_{H(t)}(G)\chi_{H(t')}(G)]\\
    &=
      n^{-2k} \cdot
      \sum_{H \in \calH_d}
      \sum_{
      \substack{t,t' \in Q:\\
      t^{}_{V(E(H))} = t'_{V(E(H))}}
      }
      \E_G[\chi_{H(t)}(G)\chi_{H(t')}(G)]\displaybreak[0]\\ 
    &=
      n^{-2k} \cdot
      \sum_{H \in \calH_d}
      \big|
      \Set{
      (t,t') \colon t,t'\in Q \text{ and } t^{}_{V(E(H))} = t'_{V(E(H))}
      }
      \big|\displaybreak[0]\\
    &=
      n^{-2k} \cdot
      \sum_{H \in \calH_d}
      |Q_{V(E(H))}| \cdot |Q_{[k] \setminus V(E(H))}|^2 \\
    &\le\label{eq:wholespace-overview}
      n^{-k}|Q| \cdot
      \Big(1 + 
      \sum_{\substack{H \in \calH_d\\ H \neq \emptyset}}
      n^{-|V(E(H))|}
    \Big)\eqperiod
\end{align}
A careful application of \cref{lem:count-H} allows us to bound the
number of pattern graphs $H$ we sum over in
\eqref{eq:wholespace-overview} to conclude that
$\E[\mu_d^2(Q)] = |Q|n^{-k}\big(1 \pm n^{-\Omega(1)}\big)$, as long as
$k$ and $d$ are small. By virtue of Chebyshev's inequality we then
conclude that $\mu(1) = 1 \pm n^{-\Omega(1)}$ asymptotically almost
surely.

A natural attempt to prove that the measure is mostly non-negative is
to analyze higher moments in the hope that these are closely
concentrated around the (positive) expected value. The fundamental
difficulty in analyzing the pseudo-measure $\mu_d$ is that we have to
analyze exponentially many rectangles simultaneously. Since there is
such a large number of rectangles, for each input graph $G$, there
will be some rectangles where the value of $\mu_d$ differs
considerably from the expected value.

For example, the measure on a rectangle $Q$ with only a few vertices
$Q_i$ in some block $V_i$ heavily depends on the behavior of the edges
incident to the vertices in $Q_i$. Hence, if $Q_i$ is small enough, we
expect large deviations from the expected value. A slightly
simplified, though more concrete, example of this phenomenon goes as
follows: let $v_1 \in V_1$ and $v_2 \in V_2$, let $Q$ be the rectangle
that consists of all tuples that contain $v_1$ as well as $v_2$, and
let $H$ be the graph with the single edge $\set{1,2}$. In this setting
the sum $\sum_{t \in Q}{\chi_{H(t)}(G)}$ heavily depends on whether
the edge $\set{v_1, v_2}$ is present in $G$: if the edge is present,
then the sum is equal to $n^{k-2}\frac{1-p}{p}$ and, if the edge is
not present, then it is equal to $-n^{k-2}$. This indicates that on
some rectangles the measure heavily depends on a few edges and we can
thus not hope to naïvely prove concentration of the measure over all
rectangles.

This slightly simplified example can be generalized to show that for a
fixed $H$ there is always a small number of rectangles where the value
contributed by $H$ is much larger than expected.  Part of the
technical challenge of the proof is to identify these bad rectangles
and to handle them separately.

\subsection{Edge Axioms Should Have Small Measure}
\label{sec:axioms-sketch}

We now explain the main ideas for bounding the magnitude of the
measure of edge axioms. Recall that all other axioms are mapped to $0$
by $\mu$ and we are thus just left to show that the value of the edge
axioms is closely concentrated around $0$. 

For every pair of vertices $\set{u,v} \notin E(G)$ in distinct blocks
we have an edge axiom $p_{uv} = x_ux_v$ stating that at least one of
$x_u$ and $x_v$ are set to $0$. Let $Q$ be a subrectangle of
$Q(p_{uv})$. Note that for every such rectangle $Q$ there is a
monomial $m$ such that $Q=Q(m \cdot p_{uv})$ and hence these are the
correct rectangles to consider if we want to prove \cref{it:axiom} of
\cref{def:pseudo-measure}. In other words, if we manage to show for
all such $Q$ that $|\mu_d(Q)| \le n^{-\Omega(d)}$, then it follows
that for all monomials $m$ it holds that
$|\mu_d(m \cdot p_{uv})| \le n^{-\Omega(d)}$, as needed.

We first show that for a fixed pair of vertices
$\set{u,v} \not\in E(G)$, with good probability, all such
subrectangles $Q$ have small absolute measure. By a union bound over
all missing edges we then conclude that all subrectangles $Q$ of an
edge axiom satisfy $|\mu_d(Q)| \le n^{-\Omega(d)}$. Let us fix an edge
$\set{u, v} \notin E(G)$.

If $Q$ is empty, then there is nothing to prove as $\mu_d(Q)$ is
trivially $0$. Hence we may assume that $Q$ is non-empty, that is, $Q$
has at least one vertex per block and hence each tuple in $Q$ contains
both $u$ and $v$. Let $i\neq j \in [k]$ such that $u \in V_i$ and
$v \in V_j$. For $e = \set{i,j}$ we may write
\begin{align}
  \mu_d(Q)
  &= n^{-k}
    \sum_{t \in Q}
    \sum_{H \in \calH_d}\chi_{H(t)}(G)\\
  &= n^{-k}
    \sum_{t \in Q}
    \Big(
      \sum_{\substack{H \in \calH_d\\e\notin H}}
      \chi_{H(t)}(G)+
      \sum_{\substack{H \in \calH_d\\ e\in H}}
      \chi_{H(t)}(G)
    \Big)\\
  &= n^{-k} \label{eq:ax-sum-of-boundary}
    \sum_{t \in Q}
    \sum_{
    \substack{
      H:\,
      \vc(H)= d,\\
      \vc(H \cup \set{e})=d+1}
      }
    \chi_{H(t)}(G) \eqcomma
\end{align} 
where the last equality follows from the fact that every tuple
$t \in Q$
contains $u$ and $v$ and thus, if $e \notin H$, then
$
\chi_{H(t)}(G) =
- \chi_{H(t)}(G) \cdot \chi_{\set{u,v}}(G) =
- \chi_{(H \cup \set{e})(t)}(G)
$
as $\set{u,v} \not\in E(G)$. 

The naïve approach to bounding $|\mu_d(Q)|$ is to try to bound the
magnitude of $\sum_{t \in Q} \chi_{H(t)}(G)$ for each $H$ separately
and to then multiply this bound by the number of graphs $H$ we sum
over. Recall from \cref{lem:count-H} that there are about $2^{dk}$
graphs with a minimum vertex cover of size~$d$.  As the magnitude of
$\sum_{t \in Q} \chi_{H(t)}(G)$ typically has value
$\Omega(n^{-d}|Q|)$, for large rectangles $Q$, even with the optimal
bound $|\sum_{t \in Q} \chi_{H(t)}(G)| \leq O(n^{-d}|Q|)$, we can only
show a bound of
\begin{align}
  |\mu_d(Q)|
  \le
  n^{-k}
  \sum_{
  \substack{H:\,
  \vc(H)= d,\\
  \vc(H \cup \set{e})=d+1}
  }
  O\big(|Q| n^{-d}\big)
  \le O\big(2^{dk} n^{-d}\big) =
  O\big(\exp(d(k-\log n))\big)\eqperiod
\end{align}
Note that for $k = \poly(n)$ much larger than both $d$ and $\log n$
the bound is at best $\exp\bigl(O(k)\bigr)$. We require a bound of the
form $|\mu_d(Q)| \le n^{-\Omega(d)}$, which is much smaller than
$\exp\bigl(O(k)\bigr)$.

Instead of bounding the magnitude of $\sum_{t \in Q} \chi_{H(t)}(G)$
for each $H$ separately, we partition the relevant set of graphs into
different families and proceed to bound the magnitude of
$\sum_{H \in \calH(F,E^\star_F)}\sum_{t \in Q} \chi_{H(t)}(G)$ for
each such family $\calH(F,E^\star_F)$. More precisely, we have
families of graphs indexed by graphs $F$ with at most $3d$
non-isolated vertices of the form
\begin{align}
\calH(F, E^\star_F) = \set{H \mid E(H) = E(F) \cup E,~\text{where}~E
  \subseteq E^\star_F} \eqcomma
\end{align}
that partition the set of graphs $H$ satisfying $\vc(H)=d$ and
$\vc(H \cup \set{e}) = d+1$. Using these families we can bound
the magnitude of $\mu_d(Q)$ by
\begin{align}
  |\mu_d(Q)|
  &=
    n^{-k}
    \Big|
    \sum_{t \in Q}
    \sum_{
      \substack{H:\, \vc(H)= d,\\ \vc(H\cup
      \set{e})=d+1}
      }
    \chi_{H(t)}(G)
    \Big|\\
  &\le
    n^{-k}
    \sum_{F}
    \Big|
    \sum_{t \in Q}
    \sum_{H \in \calH(F,E^\star_F)}
    \chi_{H(t)}(G)
    \Big|\\
  &=
    n^{-k}
    \sum_{F}
    \Big|
    \sum_{t \in Q}
    \chi_{F(t)}(G)
    \sum_{E \subseteq E^\star_F}
    \chi_{E(t)}(G)
    \Big|
    \eqperiod\label{eq:ax-sketch}
\end{align}
Observe that the innermost sum is, up to normalization, the indicator
function of whether the edge set $E^\star_F(t)$ is present in $G$. In
fact the innermost sum, with the appropriate definition of
$E^\star_F$, is simply a statement about the common neighborhood sizes
of different subsets of $t$ in $G$.  We will need to argue that for
random graphs, with high probability, all such sets behave as expected
and the innermost sums are therefore bounded.

Furthermore, since each graph $F$ has at most $3d$ with incident
edges, there are fewer such graphs: according to \cref{lem:count-H}
at most $2^{3d(d + \log k)}$. Since $k \le n^{1/66}$ and
$d \leq 2 \eta \log n$, for some small constant $\eta$, it holds that
there are at most $2^{d(d + \log k)} \lesssim n^{d/50}$ many such
graphs $F$. Thus, an upper bound of $n^{k-\Omega(d)}$ on the 
absolute value of 
two innermost sums in \cref{eq:ax-sketch}
can now be
used to obtain the claimed bound $|\mu_d(Q)| \le n^{-\Omega(d)}$.
This completes the proof sketch for bounding the measure on edge axioms.

In \cref{sec:cores-bounds} we formally define these \emph{core} graphs
$F$ and the families $\calH(F, E^\star_F)$. In \cref{sec:random} we
introduce the pseudorandomness property of graphs we rely on in order to
bound the two innermost sums in \cref{eq:ax-sketch}. In
\cref{sec:axioms-short} we formally prove that the measure on
subrectangles of axioms is bounded in absolute value and lastly, in
\cref{sec:random-proof}, we verify that random graphs indeed satisfy
the necessary pseudorandomness properties.

\subsection{Rectangles Should Be Approximately Non-Negative}
\label{sec:nonneg-sketch}

To show that all rectangles $Q$ have essentially
non-negative measure, the main idea is to decompose $Q$ into a collection $\calQ$ of
rectangles satisfying the following properties.
\begin{enumerate}
\item The collection $\calQ$ is small, that is,
  $|\calQ| \le n^{O(d)}$.
\item Each rectangle $Q \in \calQ$ is either
  \begin{enumerate}
  \item very small: $|Q| \le n^{(1-\eps) k}$ and hence $|\mu_d(Q)|$ is
    negligible,
  \item a subrectangle of an axiom and thus, as argued in
    \cref{sec:axioms-sketch}, $|\mu_d(Q)|$ is bounded, or
  \item \label{it:good-rectangles} all common neighborhoods in $Q$ are
    of expected size and therefore
    \begin{align}
    \mu_d(Q) \approx |Q|/n^k > 0 \eqperiod
    \end{align}
  \end{enumerate}
\end{enumerate}
In other words, $\calQ$ contains some rectangles that have negligible
measure and a collection of larger rectangles on which the measure
behaves as expected. As the latter rectangles have strictly positive
measure we may conclude that our pseudo-measure is essentially
non-negative on all rectangles.

We bound the measure on small rectangles by summing the maximum
possible magnitude of any character appearing in the definition of our
pseudo measure.

\begin{lemma}\label{lem:rect-small}
  Any rectangle $Q$ satisfies
  $|\mu_{d}(Q)| \le O\big(|Q|n^{-k}k^dp^{-dk}\big)$.
\end{lemma}
\begin{proof}
  We bound $\mu_d(Q)$ by counting the number of pattern graphs $H$ we sum over
  multiplied by the maximum magnitude of each such character. We have that
  \begin{align}
    |\mu_d(Q)|
    &\le
      n^{-k}
      \sum_{i=0}^d
      \sum_{j=i}^{ik}
      \Big|
      \sum_{t \in Q}
      \sum_{\substack{H:\\ \vc(H) = i\\ |E(H)| = j}}
      \chi_{H(t)}(G)
      \Big|\\
    &\le
      |Q|
      \cdot
      n^{-k}
      \sum_{i=0}^d
      \binom{k}{i}
      \sum_{j=0}^{ik}
      \binom{ik}{j}\Big(\frac{1-p}{p}\Big)^j\\
    &=
      |Q|
      \cdot
      n^{-k}
      \sum_{i=0}^d
      \binom{k}{i}
      p^{-ik}
    \le
      O
      \big(
      |Q|
      n^{-k}
      k^d
      p^{-dk}
      \big)\eqcomma
  \end{align}
  as claimed.
\end{proof}

We implement the above proof outline in
\cref{sec:partition-rect}. Proving that our pseudo-measure
concentrates around a positive value on rectangles as described in
\cref{it:good-rectangles} is the most delicate part of our proof. In
fact, the above proof outline is somewhat inaccurate in that the value the
pseudo-expectation concentrates around is not simply $|Q|/n^k$ but
further depends on the number of small blocks in the rectangle $Q$. We
refer to \cref{def:good} for the precise definition of these
rectangles and to \cref{lem:rect-good} for the claimed concentration
inequality. \Cref{sec:good-rect} is dedicated to the proof of
\cref{lem:rect-good}.

\section{Cores}
\label{sec:cores-bounds}

In this section we introduce the notion of a core of a pattern graph,
which will be used ex\-ten\-sive\-ly throughout the rest of the
paper. Our notion of a core seems to be loosely connected to the
notion of a \emph{vertex cover kernel} as used in parameterized
complexity (see, e.g., the survey by Fellows et
al.~\cite{FJKRW18vcKernel}).

\subsection{Cores and Boundaries}

Recall that when bounding the measure of subrectangles of axioms
$A_e$, we were left with sums over graphs $H$ such that $\vc(H)= d$
and $\vc(H\cup \set{e})=d+1$ (see \cref{eq:ax-sum-of-boundary}).  Such
graphs motivate the following definition of sets of graphs in the
\emph{boundary} of an edge.

\begin{definition}[Boundary]\label{def:boundary}
  Let $i \in \N$, $H$ be a graph and $e \in \binom{V(H)}{2}$ be an
  edge. The graph $H$ is in the \emph{$(i, e)$-boundary}, denoted by
  $\calH_i(e)$, if and only if $\vc(H) = i$ and
  $\vc(H \cup \set{e}) = i+1$. Furthermore, we say that $H$ is in the
  \emph{$e$-boundary} if and only if $H$ is in an $(i, e)$-boundary
  for some $i \in \N$.
\end{definition}

As mentioned in the proof sketch bounding the edge axioms, we cannot
bound each $H$ in the $e$-boundary separately (there are too many
pattern graphs $H$) so we partition such graphs according to
\emph{cores} as explained below.

\begin{definition}[Core]\label{def:core}
  A vertex induced subgraph $F$ of $H$ is a \emph{core} if any minimum
  vertex cover of $F$ is also a vertex cover of $H$.
\end{definition}

Ultimately we are interested in cores that are induced by small vertex
sets. It turns out that, in general, we cannot hope for cores of a
graph $H$ that are induced by fewer than $3\cdot \vc(H)$ many
vertices: as the graph $H$ that consists of $\vc(H)$ vertex disjoint
paths of length $2$ has only a single core $F = H$, the best we can
hope for are cores of size $3\cdot \vc(H)$.

The notions of cores and $(i,e)$-boundaries interact nicely in the
following sense.

\begin{proposition}
  \label{prop:core-boundary}
  A core of a graph $H$ is in the $(i, e)$-boundary if and only if $H$
  is.
\end{proposition}
\begin{proof} Let $F$ be a core of $H$.  We first argue that if a core
  $F$ of the graph $H$ is in the $(i, e)$-boundary, then so is
  $H$. Indeed, by definition it holds that $\vc(F) = \vc(H) =
  i$. Moreover, $F$ being in the $(i, e)$-boundary implies that the
  minimum vertex cover of $F \cup \set{e}$ has size $i+1$, and
  therefore the minimum vertex cover of $H \cup \set{e}$ must also be
  $i+1$ since $F$ is a subgraph of $H$.

  It remains to argue that if $H$ is in the $(i, e)$-boundary, then so
  is the core $F$.  By definition of core, $\vc(F) = \vc(H) = i$.
  Suppose, for the sake of contradiction, that $F$ is not in the
  $(i, e)$-boundary and thus $\vc(F \cup \set{e}) = i$. Let $\lexvc$
  be a minimum-sized vertex cover of $F \cup \set{e}$. Since
  $|\lexvc| = i$, it holds that $\lexvc$ is also a minimum-sized vertex
  cover of $F$ and thus, by definition of core, $\lexvc$ is also a
  vertex cover of $H$. But this contradicts the assumption that $H$ is
  in the $(i, e)$-boundary since $\lexvc$ also covers the edge $e$ and
  hence is a vertex cover of size $i$ of $H \cup \set{e}$.
\end{proof}

Recall that $\calH$ is the set of graphs on $k$ labeled vertices. We
consider a map $\core$ from $\calH$ to small cores that satisfies
certain properties as described below.

\begin{theorem}
  \label{lem:compression}
  There is a map $\core$ that maps graphs $H \in \calH$ to a core of
  $H$ with the following properties. For every graph $F$ in the image
  of $\core$ we have that $|V(E(F))| \leq 3\cdot \vc(F)$ and that
  there exists an edge set
  $E^\star_F \subseteq V\bigl(E(F)\bigr) \times \big([k] \setminus
  V\bigl(E(F)\bigr)\big)$ such that $\core(H) = F$ if and only if
  $E(H) = E(F) \cup E$ for $E \subseteq E^\star_F$.
\end{theorem}

We prove \cref{lem:compression} in \cref{sec:core-proofs} below.  From
now on we only consider the cores given by the map $\core$ as in
\cref{lem:compression}. With a slight abuse of nomenclature we say
that $\core(H)$ is \emph{the core} of $H$.
Note that for a graph $F$ in the image of $\core$ we have that
$\core^{-1}(F) = \calH(F, E^\star_F) = \set{H \mid E(H) = E(F) \cup
  E,~\text{for}~E \subseteq E^\star_F}$, as introduced in
\cref{sec:axioms-sketch}. 

\subsection{Proof of \cref{lem:compression}}
\label{sec:core-proofs}

We first explain how to construct a map $\core$ and then prove that it
satisfies the required properties. In order to construct the mapping
we require an order on subsets of vertices of $H$: consider every
subset of vertices as a sequence of vertices sorted in ascending order
and say that a set $U$ is lexicographically smaller than a set $V$ if
the ascending sequence $(u_1, \ldots, u_s)$ of $U$ is
lexicographically smaller than the sequence $(v_1, \ldots, v_t)$ of
$V$, that is, for $i = 1, \ldots, \min\set{s,t}$ compare $u_i$ with $v_i$: if
$u_i < v_i$ (respectively $u_i > v_i$), then $U$ is lexicographically
smaller (respectively larger) than $V$; if $u_i = v_i$ continue with
$i = i+1$; if the prefix of length $\min\set{s,t}$ is equal, then $U$
is lexicographically smaller (larger) than $V$ if $s < t$
(respectively $s > t$).

Equivalently $U$ is lexicographically smaller than $V$ if the smallest
element in the symmetric difference of $U$ and $V$ is contained in
$U$. From this alternate definition we immediately obtain the
following property of the order.

\begin{fact}\label{fact:lex}
  Let $W$ and $V$ be sets such that $W$ is lexicographically smaller
  than $V$ and $W \not\subset V$. Then, for any element $w$, the set
  $W \cup \set{w}$ is lexicographically smaller than $V$.
\end{fact}

We extend this notion in the natural manner to vertex covers and say
that a minimum vertex cover $\lexvc$ is the \emph{lexicographically
  first minimum vertex cover of $H$} if $\lexvc$ is the
lexicographically smallest minimum vertex cover of $H$.

\begin{algorithm}[t]
  \caption{Computes the core of $H$}
  \label{alg:core}
  \begin{algorithmic}[1]
    \Procedure{Core}{$H$}

    \State $\lexvc \gets$ lex first minimum vertex cover of $H$
    
    \State $U_1 \gets$ lex first maximal set in $[k] \setminus \lexvc$
    with matching $M_1 \subseteq H$ from $U_1$ to $\lexvc$ of size
    $|U_1|$

    \State $U_2 \gets$ lex first maximal set in
    $[k] \setminus (\lexvc \cup U_1)$ with matching $M_2 \subseteq H$
    from $U_2$ to $\lexvc$ of size $|U_2|$

    \State \textbf{return} $H[\lexvc \cup U_1 \cup U_2]$
    
    \EndProcedure
  \end{algorithmic}
\end{algorithm}

The map $\core$ is constructed as follows (see \cref{alg:core} for an
algorithmic description). Given a graph $H$ with lexicographically
first minimum vertex cover $\lexvc$ define the following two sets
$U_1$ and $U_2$ of vertices. Let $U_1 \subseteq [k] \setminus \lexvc$
be the lexicographically first maximal (with respect to set inclusion)
set of vertices with a matching $M_1\subseteq H$ from $U_1$ to
$\lexvc$ that covers all of $U_1$, that is, $U_1 \subseteq V(M_1)$ and
$M_1 \subseteq U_1 \times \lexvc$. Similarly let
$U_2 \subseteq [k] \setminus (\lexvc \cup U_1)$ be the
lexicographically first maximal set of vertices with a matching $M_2$
from $U_2$ to $\lexvc$ of size $|U_2|$. The core of $H$ is defined to
be $\core(H) = H[\lexvc \cup U_1 \cup U_2]$. An illustration is
provided in \cref{fig:core}.

Let us record some simple observations

\begin{claim}\label{clm:prop-core}
  For $H$, $\lexvc, U_1, U_2$ and matchings $M_1,M_2$ defined as
  above, it holds that
  \begin{enumerate}
  \item $\lexvc \cap V(M_1) \supseteq \lexvc \cap V(M_2)$,
    \label{it:m1-m2}
  \item any edge
    $e \in H \setminus H[\lexvc \cup U_1]$
    is incident to $\lexvc \cap V(M_1)$, and
    \label{it:edge-outside-1}
  \item any edge
    $e \in H \setminus H[\lexvc \cup U_1 \cup U_2]$
    is incident to $\lexvc \cap V(M_2)$.
    \label{it:edge-outside}
  \end{enumerate}
\end{claim}
\begin{proof}
  \Cref{it:m1-m2} follows by maximality of $U_1$. For
  \cref{it:edge-outside-1}, let
  $e = \set{v,w} \in H \setminus H[\lexvc \cup U_1]$. As $\lexvc$ is a
  vertex cover of $H$ we may assume that $w \in \lexvc$ and hence, as
  $e \not \in H[\lexvc \cup U_1]$, the vertex $v$ is not in
  $\lexvc \cup U_1$. Hence by maximality of $U_1$ it holds that
  $w \in \lexvc \cap V(M_1)$. The same argument establishes
  \cref{it:edge-outside}.
\end{proof}

Since $|U_2| \leq |U_1| \leq |\lexvc|$, it follows that
$\bigl|V\bigl(E(F)\bigr)\bigr| \leq 3\,|\lexvc| \leq 3 \vc(H)$.
To argue that this map satisfies the other required properties, we
first prove that the size of the minimum vertex cover of
$H[\lexvc \cup U_1]$ is the same as that of $H$. Clearly
$\vc(H) \ge \vc\bigl(H[\lexvc \cup U_1]\bigr)$ so it remains to prove
the opposite inequality.

\begin{figure}
  \centering
  \includegraphics{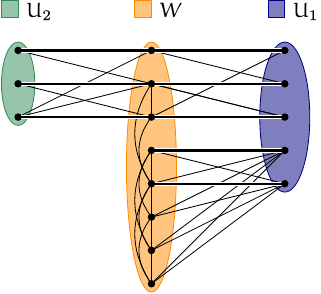}
  \caption{A core with edges in $M_1$ and $M_2$ highlighted}
  \label{fig:core}
\end{figure}

\begin{lemma}
  \label{cl:core-vc}
  For $H$, $\lexvc$ and $U_1$ defined as above, it holds that
  $\vc(H) \le \vc\bigl(H[\lexvc \cup U_1]\bigr)$.
\end{lemma}
\begin{proof}
  Let $M_1$ be the matching from $U_1$ to $\lexvc$ that covers $U_1$
  and let $R$ be the vertices in $\lexvc$ that are not matched by
  $M_1$, that is, $R = \lexvc \setminus V(M_1)$.
  Towards contradiction suppose that the graph $H[\lexvc \cup U_1]$ has a
  vertex cover of size $\vc(H)-1$.

  By \cref{clm:prop-core}, \cref{it:edge-outside-1}, the only edges in
  $H$ which do not appear in the graph $H[\lexvc \cup U_1]$ are edges
  between a vertex in $\lexvc \cap V(M_1)$ and a vertex outside of
  $\lexvc \cup U_1$.

  Let $\lexvc'$ be a vertex cover of $H[\lexvc \cup U_1]$ of size
  $\vc(H)-1$ that maximizes $|\lexvc'\cap V(M_1)|$.  If $\lexvc'$ is a
  vertex cover of $H$ we have reached a contradiction. Otherwise,
  there exists a vertex $u \in \lexvc \cap V(M_1)$ that is not
  contained in $\lexvc'$ and an edge $e_0 = \set{u,w} \in E(H)$ for
  $w \in [k] \setminus (\lexvc \cup U_1)$.

  We now want to argue that either we can construct a vertex cover of
  $H[\lexvc \cup U_1]$ of size $\vc(H)-1$ that contradicts the fact
  that $\lexvc'$ maximized $|\lexvc'\cap V(M_1)|$, or we can construct
  a strict superset of $U_1$ with a matching to $\lexvc$,
  contradicting the maximality of $U_1$.
 
  We iteratively construct two sets,
  $A \subseteq \lexvc\setminus \lexvc'$ and
  $B\subseteq \lexvc'\cap U_1$ as follows. We start by including $u$
  in $A$. We then iteratively include in $B$ all vertices $v$ in $U_1$
  that are matched by $M_1$ to some vertex in $A$, that is, vertices
  $v$ such that $\set{v,v'} \in M_1$ for some $v'\in A$; and include
  in $A$ all vertices of $\lexvc$ not covered by $\lexvc'\setminus
  B$. Note that $|A|\geq |B|$ and that we keep the invariant that
  $B \subseteq \lexvc'$ since the edges in $M_1$ must be covered by
  some vertex.
  We consider two cases.
  \begin{description}[font=\mdseries]
  \item[Case $|A| = |B|$:] In this case,
    $\big(\lexvc'\cup A \big) \setminus B$ is a vertex cover of
    $H[\lexvc \cup U_1]$ of size $\vc(H)-1$ contradicting the fact that
    $\lexvc'$ maximized $|\lexvc'\cap V(M_1)|$.
    
  \item[Case $|A| > |B|$:] This can only happen if there is a vertex
    $v$ in $R\cap A$. In this case, we can define an augmenting path
    from $v$ to $w$, alternating between edges not in $M_1$ and edges
    in $M_1$. This implies we can define a matching that matches all
    of $U_1$ to vertices in $\lexvc$ as well as the vertex $w$. This
    is in contradiction with the maximality of $U_1$.\qedhere
  \end{description}
\end{proof}

To prove that $H[\lexvc \cup U_1 \cup U_2]$ is a core of $H$ we must show
that any minimum vertex cover of $H[\lexvc \cup U_1 \cup U_2]$ is a vertex
cover of $H$.

\begin{corollary}
  \label{cor:vc-iff}
  Any minimum vertex cover of $H[\lexvc \cup U_1 \cup U_2]$ is also a
  vertex cover of $H$.
\end{corollary}
\begin{proof}
  By \cref{cl:core-vc} we have that
  \begin{align}
    \vc(H) =
    \vc\bigl(H[\lexvc \cup U_1]\bigr) =
    \vc\bigl(H[\lexvc \cup U_1 \cup U_2]\bigr) \eqperiod
  \end{align}
  Therefore, any minimum vertex cover $\lexvc'$ of
  $H[\lexvc \cup U_1 \cup U_2]$ has no vertex in $U_2$ (otherwise,
  $\lexvc'\setminus U_2$ would be a smaller vertex cover of
  $H[\lexvc \cup U_1]$). As $\lexvc'$ needs to cover the edges of the
  matching $M_2$ from $U_2$ to $\lexvc$, the vertex cover $\lexvc'$
  contains all vertices in $V(M_2) \cap \lexvc$.

  According to \cref{clm:prop-core}, \cref{it:edge-outside}, only
  vertices in $V(M_2) \cap \lexvc$ may have edges incident in $H$ that
  are not present in $H[\lexvc \cup U_1 \cup U_2]$. Hence $\lexvc'$ is
  a vertex cover of $H$; the statement follows.
\end{proof}

It remains to argue that for every $F \in \img(\core)$ there is a set
$E^\star_F$ such that $\core(H) = F$ if and only if $H = F \cup E$ for
$E \subseteq E^\star_F$. Let $E^\star_F = \bigcup_{H \in
  \core^{-1}(F)} H \setminus F$.  By definition, if $\core(H) = F$,
then there is an~$E \subseteq E^\star_F$ such that $H = F \cup E$. We
establish the reverse direction by the following two claims.

In \cref{sec:explicit} we provide an alternate proof that
characterizes the set $E^\star_F$ explicitly.

\begin{claim}\label{cl:idempotent}
  Let $\core(H) = F$ and let $\lexvc, U_1, U_2$ (respectively
  $\lexvc', U_1', U_2'$) denote the sets as in \cref{alg:core} when
  run on $H$ (respectively on $F$). It holds that $\lexvc = \lexvc'$,
  $U_1 = U_1'$ and $U_2 = U_2'$.
\end{claim}

\begin{proof}
 By definition of the map $\core$ given by \cref{alg:core}, we have that
 $F = H[\lexvc \cup U_1 \cup U_2]$.
  \Cref{cor:vc-iff} implies that $\lexvc$ is  the
  lexicographically first minimum vertex cover of $F$
  (otherwise $\lexvc$ would not be the
  lexicographically first minimum vertex cover of $H$)
  and hence $\lexvc = \lexvc'$.

  Clearly any matching $M_1$ in $H$ from $\lexvc$ to $U_1$ is also
  in $F= H[\lexvc \cup U_1 \cup U_2]$, 
  and any matching $M_1'$ in $F$ from $\lexvc$ to $U_1'$
  is also in $H\supseteq F$, and thus it must hold $U_1 = U_1'$.
  Similarly, any matching $M_2$ in $H[[k]\setminus U_1]$ 
  from $\lexvc$ to $U_2$ is also
  in $F[[k]\setminus U_1]= H[\lexvc \cup U_2]$, 
  and any matching $M_2'$ in $F[[k]\setminus U_1]$ from $\lexvc$ to $U_2'$
  is also in $H[[k]\setminus U_1]\supseteq F[[k]\setminus U_1]$, 
  and thus it must hold $U_2 = U_2'$.
\end{proof}

Note that \cref{cl:idempotent} implies that $\core(F)=F$ and that the
sets $\lexvc, U_1$ and $U_2$ as in \cref{alg:core} run on any two
graphs $H$ and $H'$ such that $\core(H)=\core(H')$ are identical.

\begin{lemma}
  For $F \in \img(\core)$ and for all $E \subseteq E^\star_F$ it holds that
  $\core(F \cup E) = F$.
\end{lemma}

\begin{proof}
  Let $E', E''$ be subsets of~$E^\star_F$ such that
  $\core(F \cup E') = \core(F \cup E'') = F$. It suffices to show that
  if
  $E \subseteq E'\cup E''$, then it holds that $\core(F \cup E) = F$.
  Let $\lexvc, U_1, U_2$ ($\lexvc', U_1', U_2'$) be the sets as in
  \cref{alg:core} when run on $F \cup E $ (respectively on
  $F \cup E'$). By \cref{cl:idempotent} the sets $\lexvc', U_1'$ and
  $U_2'$ could be equivalently defined as the sets from
  \cref{alg:core} when run on $F$ or $F \cup E''$.

  We first show that $\lexvc = \lexvc'$.  
  By \cref{cl:idempotent} we have that $\lexvc'$ is  the
  lexicographically first minimum vertex cover of $F$.
  It suffices to argue that 
  $\lexvc'$ is a vertex cover of $F \cup E$ since
  any vertex cover of $F \cup E$ is a vertex cover of $F$.
  But this is easy to see since $\lexvc'$ is a vertex cover of $F \cup E'$ 
  and of $F \cup E''$ and thus it is a vertex cover of $F \cup E' \cup E''$,
  and hence also of $F\cup E \subseteq F \cup E' \cup E''$.
 
  Suppose for the sake of contradiction that $U_1 \neq U_1'$ and let
  $u$ be the smallest vertex in the symmetric difference of $U_1$ and
  $U_1'$.
  Note that any matching $M_1'$ in $F\cup E'$ from $\lexvc$ to $U_1'$
  is also in~$F$. Hence either $U_1' \subseteq U_1$ (since $U_1$ is a
  maximal set with a matching $M_1 \subseteq F \cup E$ into $\lexvc$)
  or $U_1'$ is lexicographically larger than $U_1$ (as $U_1$ is the
  lexicographically first such set). In both cases it holds that
  $u\in U_1$.

  Let ${U}=(U_1 \cap U_1') \cup \{u\}$ and observe that since
  $u \not \in U_1'$, it holds that either $U_1' \subsetneq U$ or
  $U_1'$ is lexicographically larger than $U$. Let $M$ be a matching
  in $F \cup E$ from $U$ to $\lexvc$ that covers $U$ (which exists as
  ${U}\subseteq U_1$) and denote by $e$ the edge in $M$ adjacent to
  $u$. Since $e \in F\cup E \subseteq F \cup E' \cup E''$ it must be
  the case that $e\in F\cup E'$ or $e\in F\cup E''$. Without loss of
  generality suppose that $e\in F\cup E'$ and thus
  $M\subseteq F\cup E'$. Note that this contradicts the choice of
  $U_1'$: if $U_1' \subsetneq U$, then $U_1'$ is not maximal and, if
  $U_1'$ is lexicographically larger than $U$, then by
  \cref{fact:lex}, it is not the lexicographically first maximal set
  that can be matched to $\lexvc$.  We conclude that $U_1 = U_1'$.

  A similar argument can be used to show that $U_2 = U_2'$ and thus
  $\core(F \cup E) = F$ as claimed.
\end{proof}

\section{Well-Behaved Graphs}
\label{sec:random}

In this section, we define the notion of \emph{well-behaved} graphs,
which is based on two combinatorial properties of graphs related to
common neighborhoods of small tuples, and two analytic properties that
bound certain character sums. In \cref{sec:random-proof} we then show
that random graphs are asymptotically almost surely well-behaved. In
the following sections we prove that our measure satisfies the
required conditions to obtain our Sherali-Adams coefficient size lower
bound for any well-behaved graph.

Let us start by introducing the concepts needed to define
well-behaved graphs. 
We say a rectangle $Q$ is \emph{$s$-small} if $|Q_i| \leq s$ for all
$i\in[k]$ and, given a set $A\subseteq[k]$, a rectangle $Q$ is said to
be \emph{$(s,A)$-\compatible} if $|Q_i| > s$ for all $i\in A$.  For
any set $\calD$ we say that a function $f : \calD \rightarrow \R^+$ is
\emph{$r$-bounded} if $f(x) \leq r$ for all $x\in \calD$.

We require some terminology and notation from graph theory.  The
neighborhood of a vertex $v \in V$ in a graph $G=(V,E)$ is
$N(v) = N_G(v) = \set{u \mid \set{u,v} \in E}$ and the neighborhood of
a set of vertices $U \subseteq V$ is
$N(U) = N_G(U) = \set{v \not\in U \mid \exists u \in U: \set{u,v} \in
  E}$. For a set $W \subseteq V$ the neighborhood of a vertex $v$ in
$W$ is $N(v, W) = N(v) \cap W$ and similarly for a set $U$ we let the
neighborhood of $U$ in $W$ be $N(U, W) = N(U) \cap W$. The common
neighborhood of $U$ is $N^\cap(U) = \bigcap_{u \in U}N(u)$ and the
common neighborhood of $U$ in $W$ is
$N^\cap(U, W) = N^\cap(U) \cap W$. This notation is naturally extended
to a tuple $t$ by considering $t$ as a set of vertices.

The next two definitions are purely combinatorial. They are similar to
definitions that have appeared in previous papers on $k$-clique
\cite{BIS07,BGLR12Parameterized,ABRLNR21}. Recall that throughout the
paper graphs denoted by $G$ are $k$-partite with partitions
$V_1, \ldots, V_k$ of size $n$ each.

\begin{definition}[Bounded common neighborhoods]
  A graph $G$ has \emph{$(\beta,p)$-bounded common neighborhoods from
    $Q = \bigtimes_{i \in A} Q_i$ to $R\subseteq V(G)$} if it holds
  that for all $B\subseteq A$ and all $t \in Q_B$
  \begin{align*}
    |N^{\cap}(t,R)| \in  (1\pm \beta ) p^{|t|}|R| \eqperiod
  \end{align*}
  A graph $G$ has \emph{$(\beta,p,d)$-bounded common neighborhoods in
    every block} if for all $A\subseteq [k]$ of size at most $d$ and
  all $i \in [k]\setminus A$, $G$ has $(\beta,p)$-bounded common
  neighborhoods from $V_A$ to $V_i$.
\end{definition}


While it turns out that random graphs do have bounded common
neighborhoods, the graph induced by a rectangle may certainly have
tuples with ill-behaved common neighborhoods: we may for example have
an isolated vertex in a rectangle. The following definition roughly
states that, while there may be tuples with ill-behaved neighborhoods
in a rectangle, there is a large sub-rectangle which has bounded
common neighborhoods.

\begin{definition}[Bounded error sets]
  \label{def:bounded}
  A graph $G$ has $(s, w, \beta, p, d)$-\emph{bounded error sets} if
  for all rectangles $Q = \bigtimes_{i \in [k]} Q_i$ satisfying
  $|Q_i| \geq s$ or $|Q_i| = 0$ it holds that there exists a small set
  of vertices $W \subseteq V(G)$, $|W| \le w$, such that for all
  $S \subseteq [k]$ of size at most $d$ it holds that all tuples
  $t \in \bigtimes_{i \in S} (Q_i\setminus W) $ satisfy
  \begin{align*}
    \bigl|N^\cap(t, Q_j\setminus W)\bigr| \in
    (1\pm \beta)p^{|t|}\bigl|Q_j \setminus W\bigr|
  \end{align*}
  for all $j\in [k]\setminus S$. We refer to $W$ as the \emph{error
    set of $Q$}.
\end{definition}

Recall from the edge axiom proof sketch in \cref{sec:axioms-sketch}
that we require bounds of the form $n^{k-\Omega(\vc(F))}$ on the
absolute value of certain character sums. It turns out that, in order
to prove that monomials are mapped to an essentially non-negative
value, we need tighter (depending on $|Q|$) as well as ``localized''
versions of these bounds. For conciseness we introduce the following
terminology.

\begin{definition}[Bounded character sums]
  Let $s \in \N^+$, $B \subseteq [k]$, $Q_B = \bigtimes_{i \in B} Q_i$
  and $F$ be a core graph. A graph $G$ has \emph{$s$-bounded character
    sums over $Q_B$ for $F$} if it holds that
  \begin{align*}
    \Big|
    \sum_{t \in Q_B}
    \sum_{H \in \calH(F, E^\star_F[B])}
    \chi_{H[B](t)}(G)
    \Big|
    \le  s
    \eqperiod
  \end{align*}
\end{definition}

We are now ready to state the pseudorandomness property of graphs that
allows us to prove average-case Sherali-Adams coefficient size lower
bounds for the $k$-clique formula. As
\cref{it:bounded-char-general,it:bounded-char-tight} are somewhat
difficult to parse we give an informal description upfront.

\Cref{it:bounded-char-general} states that all character sums over the
families $\calH(F, E^\star_F)$ are of bounded magnitude if the
rectangle considered has large minimum block size. Smaller rectangles
are unfortunately not as well-behaved. However, for certain
rectangles, we can guarantee something similar:
\cref{it:bounded-char-tight} states that if the common neighborhood of
small tuples in a rectangle are bounded, then the mentioned character
sums can still be bounded.

First time readers may, for now, choose to skip the formal definition
of \cref{it:bounded-char-tight}. It might be more insightful to first
read \cref{sec:outline,sec:good-rect} and return to
\cref{it:bounded-char-tight} once it is used.

\begin{definition}[Well-behaved graphs]
  \label{def:well-behaved}
  We say that a $k$-partite graph $G$ with partitions of size $n$ is
  \emph{$D$-well-behaved} if, for $p=n^{-2/D}$, the following
  properties hold:
  \begin{enumerate}
  \item \label[property]{it:expected-neigh} $G$ has $(1/k, p, D/4)$-bounded common
    neighborhoods in every block.
  \item \label[property]{it:bounded-error} There exists a constant
    $C \in \R^+$ such that $G$ has $(2s, s, 1/k, p, \ell)$-{bounded
      error sets} for all $\ell \le D/4$ and
    $s \ge C k^4 \ell \ln n/p^{2\ell}$.
  \item \label[property]{it:bounded-char-general} For any core $F$
    satisfying $\vc(F) \le D/4$, and any
    $\big(n/2, V\bigl(E(F)\bigr)\big)$-\compatible~ rectangle $Q$ it
    holds that $G$ has $s$-bounded character sums over $Q_{[n]}$ for
    $F$, where
    \begin{align*}
      s =
      6
      \cdot
      p^{-{|E(F)|}} 
      \cdot
      n^{k-\lambda \vc(F)/4} \eqcomma
    \end{align*}
    for  any $\lambda < 1 - \log(k)/\log(n)$.
  \item \label[property]{it:bounded-char-tight}
    
    Let $F$ be a core satisfying $\vc(F) \le D/4$, let
    $\Lambda \geq 20\,k \log n$, denote by $B \subseteq [k]$ a set of
    vertices and let $A = V\bigl(E(F)\bigr) \cap B$. Then for any
    rectangle $Q$ which is $(4\Lambda)$-small and
    $(\Lambda, B)$-\compatible{} the following holds. If $G$ has
    $(3/k,p)$-bounded common neighborhoods from $Q_A$ to $Q_i$ for
    every $i\in B \setminus A$, then $G$ has $s$-bounded character
    sums over $Q_B$ for $F$, where
    \begin{align*}
      s =
      O\big(
      p^{-{|E(F[B])|}} 
      \cdot
      \left({\Lambda}/{10\,k\log n}\right)^{-\vc(F[B])/4}
      \cdot
      |Q_B|
      \big)\eqperiod
    \end{align*}
  \end{enumerate}
\end{definition}

In what follows we often state that a graph $G$ is $D$-well-behaved in
which case it is implicitly understood that $G$ is $k$-partite with
partitions of size $n$. In \cref{sec:random-proof} we prove that a
graph $G$ sampled from the distribution $\calG(n,k,n^{-2/D})$ is
asymptotically almost surely $D$-well-behaved.

\begin{restatable}{restatabletheorem}{Gwellbehaved}
  \label{lem:Gwellbehaved}
  If $n$ is a large enough integer, $k \in \N^+$ and $D \in \R^+$
  satisfy $4\leq D \le 2 \log n$ and $k \le n^{1/5}$, then
  $G\sim \mathcal{G}(n, k, n^{-2/D})$ is asymptotically almost surely
  $D$-well-behaved.
\end{restatable}

\section{Clique Is Hard on Well-Behaved Graphs}
\label{sec:outline}

In this section we prove that our measure $\mu_d$ is an
$n^{-\Omega(D)}$-pseudo-measure for the $k$-clique formula, if the
formula is defined over a $D$-well-behaved graph $G$.

\begin{theorem}
  \label{thm:mu}
  There are constants $\eta, c > 0$ and $D_0 \in \N$ such that the
  following holds for large enough $n \in \N$ and all $D$ satisfying
  $D_0 < D \le 2\log n$. If $D \le k \le n^{1/66}$, $d = \eta D$ and
  $G$ is a $D$-well-behaved $k$-partite graph with $n$ vertices per
  block, then the normalized measure $\mu(m) = \mu_d(m)/\mu_d(1)$ is
  an $n^{-cD}$-pseudo-measure for the $k$-clique formula over $G$.
\end{theorem}

From \cref{thm:mu} along with \cref{lem:Gwellbehaved} and
\cref{prop:pseudo} we obtain \cref{thm:main-formal}.

In order 
to prove that the measure $\mu$ satisfies the properties of a
pseudo-measure as listed in \cref{def:pseudo-measure}, we show that
$\mu_d$ maps any axiom multiplied by a monomial to approximately~$0$
and that all monomials are associated with an essentially non-negative
value. Finally, we argue that $\mu_d(1) \ge 1 - n^{-\Omega(1)}$.

Recall that the clique formula consists of block axioms
$\sum_{v \in V_i} x_v - 1 = 0$ for each block $V_i$ and edge axioms
$x_ux_v = 0$ for every non-edge in the graph.
By linearity of $\mu_d$ over the tuples it holds for any monomial
$m$ 
that $\mu_d\big(m(\sum_{v \in V_i} x_{v} - 1)\big) =
0$.
The lemma below implies that for edge axioms $p_{uv}=x_ux_v$ 
it holds that $|\mu_d(mp_{uv})| \le n^{-cD}$, for any monomial $m$. 
As mentioned in \cref{sec:nonneg-sketch}, we also rely on this lemma
to prove that the measure is essentially non-negative. Since that
proof requires a careful choice of parameters we need to state the
lemma with some precision.

\begin{restatable}{restatablelemma}{axiomsmall}
  \label{lem:small-measure}
  Let $G$ be a $D$-well-behaved graph, let $n$ be a large enough
  integer and let $d = \eta D \le 2 \eta \log n$ for some constant
  $\eta > 0$. 
  It holds that all edge axioms $p_{uv}$
  and all rectangles $Q\subseteq Q({p_{uv}})$ satisfy
  $|\mu_d(Q)| \le O\big(\big(n^{\lambda /4 - 12 \eta}/(2k)^3\big)^{-d}\big)$
  for any $\lambda < 1 - \log(k)/\log(n)$.
\end{restatable}

Note that by choosing $\lambda = 1/2$, and 
considering $k\leq n^{1/66}$ and $\eta>0$ small enough,
\cref{lem:small-measure} implies that any
subrectangle of an edge axiom satisfies $|\mu_d(Q)| \le n^{-cD}$ for
some small enough constant $c$.
We postpone the proof of \cref{lem:small-measure} to \cref{sec:axioms-short}.

In addition to the bound on the magnitude of the measure on the axioms
we also need to prove that the measure is essentially non-negative.
We state this formally below and defer the proof to
\cref{sec:partition-rect}.

\begin{restatable}{restatablelemma}{nonneg}
  \label{lem:nonneg}
  There are constants $\eta, c > 0$ such that if $G$ is a
  $D$-well-behaved graph, $n$ is large enough,
  $d = \eta D \le 2\eta \log n$ and $D \le k \le n^{1/66}$, then any
  rectangle $Q$ satisfies $\mu_d(Q) \ge - n^{-cD}$.
\end{restatable}

In \cref{sec:whole-space} we argued that with high probability
$\mu_d(1)$ is approximately $1$ if $G$ is a random graph and
$p=1/2$. We now show that this holds for
any $D$-well-behaved graph.

\begin{lemma}
  \label{lem:whole-space-again}
  There are constants $\eta, c > 0$ such that for $n$ large enough,
  $k \le n^{1/20}$, $D \le 2 \log n$ and $d=\eta D$ it holds that if
  $G$ is a $D$-well-behaved graph, then $\mu_d(1) \ge 1 -n^{-c}$.
\end{lemma}

\begin{proof}
  This is a direct consequence of the definition of a $D$-well-behaved
  graph and \cref{lem:compression}. Recall the map $\core$ as defined
  in \cref{lem:compression} and the families of graphs
  \begin{align}\label{eq:F-family}
    \calH(F, E^\star_F) = \set{H \mid E(H) = E(F) \cup E, \text{~for~}
    E \subseteq E^\star_F} \eqcomma
  \end{align}
  defined for core graphs $F \in \img(\core)$ such that
  $\vc(F) \le d$. Recall that each such core graph $F$ satisfies that
  $|V\bigl(E(F)\bigr)| \leq 3 \vc(F)$ and hence
  $|E(F)| \leq 3\bigl(\vc(F)\bigr)^2 \leq 3d \vc(F)$.

  By appealing to \cref{it:bounded-char-general} of a $D$-well-behaved
  graph, that is, \cref{it:bounded-char-general} of
  \cref{def:well-behaved}, with~$\lambda = 4/5$ we can conclude that
  for every $F \in \img(\core)$ it holds that
  \begin{align}\label{eq:bound-core}
    n^{-k}
    \Big|
    \sum_{t \in Q(1)}
    \sum_{H\in \calH(F,E^\star_F) }\chi_{H(t)}(G)
    \Big|
    \leq 6\,p^{-|E(F)|}n^{-\vc(F)/5}
    \leq n^{-\vc(F)/6}
    \eqcomma 
  \end{align}
  where we used the bound
  $p^{-|E(F)|} \leq p^{-3d\vc(F)} \leq n^{6 \eta \vc(F)}$, which holds
  since $p = n^{-2/D}$ and $d = \eta D \leq 2\eta \log n$, and,
  furthermore, relied on the assumption that $\eta$ is a small enough
  constant.

  Recall that the families
  $\set{\calH(F, E^\star_F) \mid F \in \img(\core)}$ as defined in
  \cref{eq:F-family} partition the set $\calH_d$ of graphs of vertex
  cover at most $d$. We may hence write
  \begin{align}
    \mu_d(1)
    &=
      n^{-k}
      \sum_{H \in \calH_d}
      \sum_{t \in Q(1)}
      \chi_{H(t)}(G)\\
    &= 1 + n^{-k}
      \sum_{\substack{H \in \calH_d\\ H \neq \emptyset}}
    \sum_{t\in Q(1)}
    \chi_{H(t)}(G)\\
    &\geq 1 - n^{-k}
      \sum_{i = 1}^d
      \sum_{\substack{F \in \img(\core)\\ \vc(F) = i}}
    \Big|
		\sum_{t \in Q(1)}
    \sum_{H\in \calH(F,E^\star_F) }
    \chi_{H(t)}(G)
    \Big| \eqperiod
  \end{align}
  Since each graph $F \in \img(\core)$ satisfies that
  $|V\bigl(E(F)\bigr)| \leq 3 \vc(F)$, by appealing to
  \cref{lem:count-H}, we obtain the bound
  $|\set{F \in \img(\core)\mid\vc(F) = i}| \leq 2^{3i(d + \log
    k)}$. Combining this bound along with the bound from
  \cref{eq:bound-core} we may conclude that
  \begin{align}
    \mu_d(1)
    &\geq 1 - n^{-k}
		\sum_{i = 1}^d
		\sum_{\substack{F \in \img(\core)\\ \vc(F) = i}}
    \Big|
		\sum_{t \in Q(1)}
    \sum_{H\in \calH(F,E^\star_F) }
    \chi_{H(t)}(G)
    \Big| \\
    &\geq 1 - \sum_{i=1}^d 2^{3i(d + \log k)} n^{-i/6}\\
    &\geq
      1- n^{-c}
    \eqcomma
  \end{align}
  for some small constant $c > 0$. The final inequality relies on the
  assumptions $d \le 2 \eta \log n$, that $\eta$ is a small enough
  constant and that $k \le n^{1/20}$. This concludes the proof.
\end{proof}

This completes the proof of \cref{thm:mu} modulo
\cref{lem:small-measure} and \cref{lem:nonneg}. In
\cref{sec:axioms-short} we prove \cref{lem:small-measure} and the
proof of \cref{lem:nonneg} is provided in \cref{sec:partition-rect}.

\subsection{Axioms Have Small Measure} 
\label{sec:axioms-short}

In this section we show that any subrectangle of an edge axiom has
small measure in absolute value. We rely on the following technical
lemma.

\begin{lemma}\label{lem:bound-any-rectangle}
  If $G$ is a $D$-well-behaved graph, then for any core graph $F$ and
  any rectangle~$Q$ it holds that
  \begin{align*}
    \Big|
    \sum_{t\in Q}
    \sum_{H\in\calH(F,E^\star_F)}
    \chi_{H(t)}(G)
    \Big|
    \leq
    6 \cdot
    2^{|A|}\cdot
    p^{-|E(F)|}\cdot
    n^{k - \lambda \vc(F)/4}
    \eqcomma
  \end{align*}
  where $A = V\bigl(E(F)\bigr)$ and $\lambda < 1 - \log(k)/\log(n)$.
\end{lemma}

\begin{proof}
  Let $F$ be a core graph, let $A=V\bigl(E(F)\bigr)$ and
  $s= 6 \cdot p^{-|E(F)|} \cdot n^{k - \lambda \vc(F)/4}$.  By
  \cref{it:bounded-char-general} of \cref{def:well-behaved} we have
  that if $Q$ is $(n/2,A)$-\compatible~, that is, if $Q$ satisfies
  $|Q_i| \geq n/2$ for all $i\in A$, then
  $\big|\sum_{t\in Q}\sum_{H\in\calH(F,E^\star_F)}\chi_{H(t)}(G)\big|
  \leq s$.

  Given any rectangle $Q$ (not necessarily $(n/2,A)$-\compatible), let
  $T\subseteq A$ be the set of blocks of $Q$ such that $|Q_i| <
  n/2$. By a simple inclusion-exclusion argument, we obtain that
  \begin{align}
    Q =
    \sum_{S\subseteq T} (-1)^{|S|}
    \Big(
    \bigtimes_{i \in S}
    (V_i\setminus Q_i)
    \Big)
    \times
    \Big(
    \bigtimes_{i \in T\setminus S} V_i
    \Big)
    \times
    \Big(
    \bigtimes_{i \in [k]\setminus T} Q_i
    \Big) \eqperiod
  \end{align}

  For $S\subseteq T$, denote by $Q^S$ the rectangle
  $\Big(\bigtimes_{i \in S} (V_i\setminus Q_i)\Big) \times
  \Big(\bigtimes_{i \in T\setminus S} V_i\Big) \times
  \Big(\bigtimes_{i \in [k]\setminus T} Q_i\Big) $. Note that $Q^S$ is
  $(n/2,A)$-\compatible~and therefore by
  \cref{it:bounded-char-general} of \cref{def:well-behaved}, $G$ has
  $s$-bounded character sums over $Q^S$ for $F$.
  This implies that
  \begin{align}
    \Big|
    \sum_{t\in Q}
    \sum_{H\in\calH(F,E^\star_F)}\chi_{H(t)}(G)
    \Big| 
    \leq
    \sum_{S\subseteq T}
    \Big|
    \sum_{t\in Q^S}
    \sum_{H\in\calH(F,E^\star_F)}
    \chi_{H(t)}(G)\Big|
    \leq
    2^{|A|} \cdot s \eqcomma
  \end{align}
  as claimed.
\end{proof}

We are now ready to prove \cref{lem:small-measure} restated for
convenience.

\axiomsmall*

\begin{proof}
  Fix an edge $\set{u, v} \notin E(G)$, let $i,j \in [k]$ such that
  $u \in V_i$ and $v \in V_j$, consider the edge axiom
  $p_{uv} = x_{u}x_{v}$ and let $Q \subseteq Q({p_{uv}})$ be an
  arbitrary subrectangle of this edge axiom. Recall that $\calH_d$
  denotes the set of graphs $H$ satisfying $\vc(H) \leq d$, and as
  explained in \cref{sec:axioms-sketch}, recall that every tuple
  $t \in Q$ contains the vertices $u$ and $v$. Thus for
  $e = \set{i,j}$ we may cancel a character $\chi_H$ satisfying
  $e \not\in H$ with the character $\chi_{H \cup \set{e}}$ to obtain
  that
  \begin{align}
    \mu_d(\rect)
    &= n^{-k}
      \sum_{t \in \rect}
      \sum_{H \in \calH_d}\chi_{H(t)}(G)\\
    &= n^{-k}
      \sum_{t \in \rect}
      \sum_{
      H \in \calH_d(e)
      }
      \chi_{H(t)}(G) \eqcomma
  \end{align}
  where $\calH_d(e)$, as defined in \cref{def:boundary}, denotes the
  set of graphs in the $(d,e)$-boundary, that is, all graphs $H$ such
  that $\vc(H) = d$ and if we add the edge $e$ to $H$, then the size
  of the minimum vertex cover increases. Let the map $\core$ be as
  guaranteed by \cref{lem:compression}. Recall that according to
  \cref{prop:core-boundary} the graph $\core(H)$ is in the
  $(d,e)$-boundary $\calH_d(e)$ if and only if $H$ is. Hence the sets
  \begin{align}
    \Set{
    \core^{-1}(F) = 
    \calH(F, E^\star_F)
    \mid
    F \in \calH_d(e) \text{~and~} F \in \img(\core) 
    }
  \end{align}
  partition the $(d,e)$-boundary and we may thus write
  \begin{align}
    |\mu_d(\rect)|
    &=
      \Big|
      n^{-k} 
      \sum_{t \in \rect}
      \sum_{H\in \calH_d(e)} \chi_{H(t)}(G)
      \Big|\\
    &=
      \Big|
      n^{-k} 
      \sum_{\substack{F \in \calH_d(e)\\
      F \in \img(\core)}}
      \sum_{t \in \rect}
      \sum_{H\in\calH(F,E^\star_F)}\chi_{H(t)}(G)
      \Big|\\
    &\leq
      n^{-k} 
      \sum_{\substack{F \in \calH_d(e)\\ F \in \img(\core)}}
      \Big|
      \sum_{t \in \rect}
      \sum_{H\in\calH(F,E^\star_F)}\chi_{H(t)}(G)
      \Big|\eqperiod
  \end{align}

  By \cref{lem:bound-any-rectangle} each inner part can be bounded by
  $6 \cdot 2^{|V(E(F))|} \cdot p^{-|E(F)|} \cdot n^{k - \lambda
    \vc(F)/4}$. Note that $ \vc(F) = d $ and, according to
  \cref{lem:compression}, it holds that $|V\bigl(E(F)\bigr)| \leq
  3d$. Hence since $d \le 2\eta \log n$ it holds that
  $p^{-|E(F)|} \leq p^{-3d^2} = n^{6\eta d}$ and we may thus conclude
  that
  \begin{align}
    |\mu_d(\rect)|
    &\le
      \sum_{\substack{F \in \calH_d(e)\\ F \in \img(\core)}}
      6 \cdot 2^{3d} \cdot n^{-d(\lambda/4 - 6\eta)}\\
    &\le
      2^{3d(d + \log k)} \cdot
      6 \cdot
      2^{3d} \cdot
      n^{-d(\lambda/4 - 6\eta)}\\
    &\le
      6 \cdot
      \big(n^{\lambda /4 - 12 \eta}/(2k)^{3}\big)^{-d}
      \eqcomma
  \end{align}
  where we used \cref{lem:count-H} to bound the number of core graphs
  and relied, again, on the assumption $d \le 2 \eta \log n$. This
  concludes the proof of \cref{lem:small-measure}.
\end{proof}

\subsection{All Rectangles Are Approximately Non-Negative}
\label{sec:partition-rect}

In this section we prove that the measure is essentially non-negative
on all monomials modulo a concentration inequality whose proof we
postpone to \cref{sec:good-rect}. For convenience we restate the
precise claim.

\nonneg*

Recall from the proof sketch given in \cref{sec:nonneg-sketch} that we
intend to decompose any given rectangle into a small family $\calQ$ of
rectangles such that each rectangle $Q \in \calQ$ either
\begin{enumerate}
\item contains few tuples,\label{it:small-rect}
\item is a subrectangle of an edge axioms, or\label{it:sub-axiom}
\item is a so-called \emph{good rectangle}. \label{it:good-rect}
\end{enumerate}
Since rectangles as described in \cref{it:small-rect,it:sub-axiom}
have negligible measure (see \cref{lem:rect-small,lem:small-measure})
essentially all the measure is concentrated on these good
rectangles. Hence if we can show that the measure concentrates around
a strictly positive value on such good rectangles, then the statement
follows. Let us introduce these \emph{good rectangles}.

Before defining these rectangles formally, let us give an informal
description. A good rectangle $Q$ consists of two parts. The first
part is very small: on a few blocks the rectangle $Q$ only consists of
single vertices. Each vertex in this small part is adjacent to
\emph{all} other vertices in $Q$. Equivalently, on this small part we
have a clique and the remaining vertices in $Q$ are in the common
neighborhood of this clique.

On the other blocks, where $Q$ does not consist of a single vertex, we
require that these blocks are large, of size at least $s =
\poly(n)$. In addition we also require that all common neighborhoods
are bounded on this large part. An illustration of a good rectangle
can be found in \cref{fig:good-rect}. The formal definition follows.

\begin{definition}[Good rectangle]
  \label{def:good} Let $G$ be a $k$-partite graph and let
  $s,\beta,p,d \in \R^+$ and $R \subseteq [k]$. A rectangle
  $Q = \bigtimes_{i \in [k]} Q_i$ is \emph{$(s, \beta, p, d, R)$-good
    for $G$} if it satisfies the following properties.
  \begin{enumerate}
  \item If $i \in R$, then $Q_i = \set{v_i}$; otherwise
    $ |Q_i| \geq s$.
  \item For all $i \in R$ it holds that
    $N(v_i) \supseteq \bigcup_{j \neq i} Q_j$.
  \item \label[property]{it:good-bounded} For all
    $S \subseteq [k] \setminus R$ of size at most $d$ and for all
    $i \not\in R \cup S$, $G$ has $(\beta,p)$-bounded common
    neighborhoods from $Q_S$ to $Q_i$.
    \label{item:good-neigh}
  \end{enumerate}
\end{definition}

\begin{figure}
  \centering
  \includegraphics{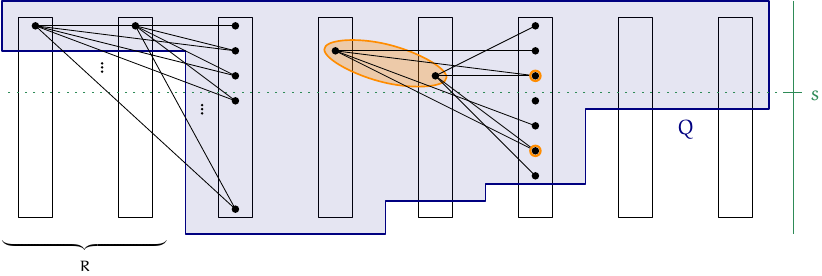}
  \caption{The rectangle $Q$ is a good rectangles as the vertices in
    $R$ have all vertices as neighbors, the blocks outside $R$ are
    large and small tuples on these blocks have common neighborhoods
    of expected size}
  \label{fig:good-rect}
\end{figure}

On good rectangles the measure is tightly concentrated around the
expected value. In \cref{sec:good-rect} we prove the following
concentration bound.

\begin{restatable}{restatablelemma}{goodRectangle}
  \label{lem:rect-good}
  For constants $\eps > 0$ and $\eta < 1/50$, for $n, k, d \in \N$ and
  $p = n^{-2/D}\le 1/2$ satisfying $d \le \eta D$ and $D \le k \le n$
  the following holds. If $s  \ge k^{13}n^{48\eta + \eps}\log n$
  and $G$ is a $D$-well-behaved graph, then any
  $(s, 1/k, p, d, R)$-good rectangle $Q$ for $G$ with $|R| = \ell < d$
  satisfies
  \begin{align*}
    \mu_d(Q) =
    p^{-\ell(k - (\ell + 1)/2)}\,
    |Q|\,
    n^{-k}
    \big(1 \pm O(n^{-\eps/8})\big) \eqperiod
  \end{align*}
\end{restatable}

In the remainder of this section we prove \cref{lem:nonneg}, assuming
\cref{lem:rect-good}.
As outlined in \cref{sec:nonneg-sketch}, we intend to decompose any
rectangle $Q$ into a small family $\calQ$ of rectangles such that each
rectangle in $\calQ$ either contains few tuples, is a subrectangle of
an edge axiom or is a good rectangle. The following lemma summarizes
our claim.

\begin{lemma}
  \label{lem:decompose}
  Let $G$ be a $D$-well-behaved graph,
  let
  $p = n^{-2/D}$, $d\leq D/4$ and
  $s \ge Ck^4 d \ln n / p^{2d}$
  for some large enough constant $C$.
  Then any rectangle $Q_0$ can be partitioned into a set of rectangles
  $\calQ$ of size $|\calQ| \le 2kn(2s)^{d}$
  such that each $Q \in \calQ$ satisfies that either
  \begin{enumerate}
  \item $Q$ is small:
    $|Q| \le O\big((n\cdot p^d)^{k - d}\big)$,
  \item $Q$ is a subrectangle of an edge axiom, or
  \item $Q$ is $(s, 1/k, p, d, R)$-good for $G$, where
    $R \subseteq [k]$ satisfies $|R| < d$.
  \end{enumerate}
\end{lemma}

Before proving \cref{lem:decompose}, let us first show how
\cref{lem:nonneg} follows. The idea of the proof is to apply
\cref{lem:decompose} to a given rectangle $Q_0$ to obtain a collection
$\calQ$ of rectangles. It holds that
$\mu_d(Q_0) = \sum_{Q \in \calQ} \mu_d(Q)$. By \cref{lem:rect-small}
there is a $\delta > 0$ such that all small rectangles $Q \in \calQ$
satisfy $|\mu_d(Q)| \le n^{-\delta D}$ and similarly by
\cref{lem:small-measure} the same holds for $Q \in \calQ$ that are a
subrectangle of an edge axiom. Further, by our choice of parameters,
the size of $\calQ$ is small---we may think of it as $n^{\delta
  D/2}$. We can thus lower bound
\begin{align}
  \mu_d(Q_0) = \sum_{Q \in \calQ} \mu_d(Q) \ge
  -
  n^{-\delta D/2}
  +
  \sum_{
  \substack{
  Q \in \calQ\\
  Q \text{~is good}
  }
  }
  \mu_d(Q)  \eqperiod
\end{align}
\Cref{lem:rect-good} states that the remaining good rectangles in the
above sum have strictly positive value. Thus
$\mu_d(Q_0) \ge -n^{-\delta D/2}$ as claimed. In what follows we
verify that this indeed holds for our choice of parameters.

\begin{proof}[Proof of \cref{lem:nonneg}]
  Let $Q_0$ be any rectangle. Our goal is to show that $\mu_d(Q_0) \ge - n^{-cD}$,
  for a sufficiently small constant $c$.
  Let $D \le k \le n^{1/66}$ be as in the statement of the lemma
  and choose
  $\lambda = 1 - \eps - \log(k)/\log(n)$ for sufficiently small
  constants $\eps > 0$ and $\eta > 0$ such that for
  $s = k^{13} n^{48\eta + \eps} \log n$ it holds that $s \leq n^{\lambda /4 - 12 \eta -
    \eps}/k^3$.  Let $d = \eta D \le 2\eta \log n$ and $p = n^{-2/D}$.
  Note that for our choice of parameters it holds that
  $ s = \omega(k^4 n^{4\eta} \log^2 n)$, hence
  $s = \omega(k^4 d \ln n / p^{2d})$, and we may thus apply
  \cref{lem:decompose} with $d = \eta D$ to the rectangle $Q_0$ to
  obtain a family $\calQ$ of size at most $|\calQ| \le 2kn (2s)^d$.

  By \cref{lem:small-measure}, any subrectangle of an axiom has
  measure bounded by
  $O\big(\big(n^{\lambda /4 - 12 \eta}/(2k)^3 \big)^{-d}\big)$.
  Moreover, according to \cref{lem:rect-small} each small rectangle
  $Q \in \calQ$ has measure of magnitude at most
  \begin{align}
    |\mu_d(Q)|
    &\le
      O\big(|Q|n^{-k}k^dp^{-dk}\big)
      =
      O(n^{-d/2})
      \eqcomma
  \end{align}
  which is even smaller than the bound on axioms.  Since
  $s \leq n^{\lambda /4 - 12 \eta -\eps}/k^3 $, we conclude that the
  measure of all small rectangles and all subrectangles of axioms in
  $\calQ$ add up, in magnitude, to at most
  $|\calQ| \cdot O\big(\big(n^{\lambda /4 - 12 \eta}/(2k)^3
  \big)^{-d}\big) \leq n^{-cD}$, for a small enough constant $c$.

  Hence the measure of $Q_0$ is mostly on the good rectangles of
  $\calQ$ and on these rectangles we know that it is closely
  concentrated around a strictly positive value. Indeed, we can apply
  \cref{lem:rect-good} to any rectangle $Q$ which is
  $(s,1/k,p,d,R)$-good for $G$ to conclude that
  \begin{displaymath}
    \mu_d(Q) =
    p^{-\ell(k - (\ell + 1)/2)} \,
    |Q| \,
    n^{-k}
    \big(
    1\pm
    O(n^{-\eps/8})
    \big) > 0 \eqperiod
  \end{displaymath}
  This concludes the proof of \cref{lem:nonneg} modulo
  \cref{lem:rect-good} and \cref{lem:decompose}.
\end{proof}

Let us proceed to prove \cref{lem:decompose}.

\begin{proof}[Proof of \cref{lem:decompose}]
  Let us describe a recursive decomposition procedure that can be
  applied to any rectangle $Q = \bigtimes_{i \in [k]} Q_i$.
  
  If either $Q$ is small, a subrectangle of an axiom or
  $(s, 1/k, p, d, R)$-good for some $R \subseteq [k]$, then return
  $Q$. Otherwise decompose in the following recursive fashion.
  \begin{enumerate}
  \item \label[case]{case:non-neigh} If there is a singleton
    $Q_i = \set{v_i}$ such that
    $N(v_i) \not\supseteq \bigcup_{j \neq i} Q_j$, then we decompose
    $Q$ into $|Q \setminus N(v_i)| + 1$ many rectangles as
    follows. Denote by $u_1, u_2, \ldots, u_m$ the vertices in $Q$
    that are not a neighbor of $v_i$ and assume that they are in
    blocks $j_1, j_2, \ldots, j_m$. For $\nu = 1, \ldots, m$ we remove
    all tuples that contain the vertex $u_\nu$: let $R^0 = Q$ so we
    can write
    \begin{align}
      Q^\nu
      &=
        \set{u_\nu} \times
        \bigtimes_{j \neq j_\nu} R^{\nu-1}_j \text{\qquad and}
      &R^{\nu}
        =
        \big(R^{\nu-1}_{j_\nu} \setminus u_\nu\big)
        \times
        \bigtimes_{j \neq j_\nu} R^{\nu-1}_j \eqperiod
    \end{align}
    Note that the rectangles $Q^1, \ldots, Q^m, R^m$ partition
    $Q$. Add the $Q^\nu$ to the partition as these are subrectangles
    of edge axioms and recursively decompose $R^m$.
    
  \item \label[case]{case:small} If there is a block $i \in [k]$ of
    size $1 < |Q_i| \le 2s$, then split $Q$ into the $|Q_i|$
    rectangles
    \begin{align}
      \Set{\set{v_i}\times\bigtimes_{j \neq i} Q_j : v_i \in Q_i}
    \end{align}
    and recursively decompose each of these rectangles.

  \item \label[case]{case:error} Let $A$ be the set of blocks of size
    greater than $2s$. Because $G$ is $D$-well-behaved, by
    \cref{it:bounded-error} of \cref{def:well-behaved}, it holds that
    $G$ has $(2s, s, 1/k, p, d)$-bounded error sets. In particular
    $Q_A$ has an error set $U = \set{u_1, \ldots, u_m}$ of size at
    most $s$. Decompose $Q$ into $Q^1, \ldots, Q^m$ and $R^m$ as in
    \cref{case:non-neigh}. By definition the rectangle $R^m$ is
    $(s, 1/k, p, d, [k] \setminus A)$-good and we may thus add it to
    the partition. Recursively decompose the rectangles
    $Q^1, \ldots, Q^m$.
  \end{enumerate}
  
  This completes the description of the decomposition procedure.
  We need to argue that the decomposition $\calQ$ created by above
  procedure is not too large, that is, of size
  $|\calQ| \le 2kn \cdot (2s)^{d}$. Let us start with a few
  observations.

  Because $G$ is $D$-well-behaved it holds that $G$ has
  $(1/k, p, D/4)$-bounded common neighborhoods in every block (see
  \cref{it:expected-neigh} of \cref{def:well-behaved}). Let $Q$ be a
  rectangle with $d$ blocks of size $1$ and with the remaining
  vertices contained in the common neighborhood of these
  singletons. All such rectangles $Q$ are small. Thus the
  decomposition procedure does not need to decompose such rectangles
  $Q$ any further.
  
  Whenever we decompose a rectangle in \cref{case:small,case:error}
  all rectangles that we need to recursively decompose have one more
  singleton. Because we can stop decomposing after identifying $d$
  singletons and in \cref{case:small,case:error} we create at most
  $2s$ many rectangles that require further decomposition we end up
  with at most $2(2s)^d$ many rectangles. We ignored the rectangles
  from \cref{case:non-neigh} so far. But each rectangle that requires
  further decomposition from \cref{case:small,case:error} results in
  at most another $kn$ many rectangles from
  \cref{case:non-neigh}. Thus the size of the family of rectangles is
  bounded by $2kn \cdot (2s)^d$.
\end{proof}

\section{The Measure Is Concentrated on Good Rectangles}
\label{sec:good-rect}

This section is devoted to proving \cref{lem:rect-good} that states
that the measure on good rectangles is very well concentrated. We
restate it here for convenience.

\goodRectangle*

Let us introduce some notation and state \cref{lem:rect-good} once
more in a more convenient form for what follows.
Let $S_i^{\le j}$ be the star with center $i$ and leaves
$[j] \setminus i$, that is, $S_i^{\leq j}$ consists of the vertices
$[k]$ and edges $\Set{\set{i,j'} : j'\in [j], j'\neq i}$. For
simplicity we denote by $S_i$ the star $S_i^{\leq k}$ and for
$I\subseteq [k]$ let $S_I = \cup_{i\in I} S_i$.

\begin{lemma}
  \label{lem:rect-good-full}
  For constants $\eps > 0$ and $\eta < 1/50$, for $n, k, d \in \N$ and
  $p = n^{-2/D} \leq 1/2$ satisfying $d \le \eta D$, $D \le k \le n$
  the following holds.  If $s \ge k^{13}n^{48\eta+\eps} \log n$
  and $G$ is a $D$-well-behaved graph, then any
  $(s, 1/k, p, d, R)$-good rectangle $Q$ for $G$ with $|R| < d$
  satisfies
  \begin{align*}
    \Big|
    p^{-|S_R|} \, |Q| - 
    \sum_{t \in Q} 
    \sum_{H \in \calH_d}
    \chi_{H(t)}(G)
    \Big|
    \le
    O\bigl(
    p^{-|S_R|}\,
    |Q|\,
    n^{-\eps/8}
    \bigr)
    \eqperiod
  \end{align*}
\end{lemma}

Before proving \cref{lem:rect-good-full} let us verify that
\cref{lem:rect-good} indeed follows.

\begin{proof}[Proof of \cref{lem:rect-good}]
  Recall that we defined our measure $\mu_d$ for a rectangle $Q$ as
  \begin{align}
    \mu_d(Q) =
    n^{-k}
    \sum_{t\in Q}
    \sum_{H \in \calH_d}
    \chi_{H(t)}(G) \eqperiod
  \end{align}
  Hence \cref{lem:rect-good-full} implies, for the appropriate
  parameters, that all $(s, 1/k, p, d, R)$-good rectangles $Q$ for
  $G$ satisfy
  \begin{displaymath}
    \mu_d(Q) =
    p^{-|S_R|}\,
    |Q|\,
    n^{-k}
    \big(
    1 \pm
    O(n^{-\eps/8})
    \big)
    \eqperiod
    \qedhere
  \end{displaymath}
\end{proof}

Let us give some intuition for the statement of
\cref{lem:rect-good-full}. Clearly, if $R = \emptyset$, then we are
showing that the measure of a rectangle is tightly concentrated around
the expected value. Let us explain where the $p^{-|S_{[\ell]}|}$
factor comes from.

Consider two blocks, say blocks $1$ and $2$. For vertices $v_1 \in
V_1$ and $v_2 \in V_2$ let us denote by $Q_{v_1v_2}$ the rectangle
consisting of all tuples that contain $v_1$ as well as
$v_2$. According to \cref{lem:small-measure}, if there is no edge
between $v_1$ and $v_2$, then $\mu_d(Q_{v_1v_2}) \approx 0$. Recall
that the measure of the whole space is roughly $1$---hence if there is
an edge between vertices $v_1$ and $v_2$ as above, we expect that
$\mu(Q_{v_1v_2})$ ``compensates'' for the $0$ value rectangles, that
is, we expect that $\mu(Q_{v_1v_2}) \approx (|Q_{v_1v_2}|/n^k)) \cdot
(1/p) = 1/p n^2$ if the edge $v_1v_2$ is present.  More generally, for
$R\subseteq [k]$ we expect to pick up a factor of $1/p$ for each edge
that we condition on being present between a vertex in $R$ and another
block in $Q$. \cref{lem:rect-good-full} establishes that this is
indeed how the measure behaves.

Let us consider an $(s, 1/k, p, d, R)$-good rectangle $Q$. For ease
of exposition let us assume that $R = [\ell]$, the other cases are
analogous. In other words, we assume, for all $i \in [\ell]$, that it
holds that $|Q_i| = 1$, these $\ell$ vertices form a clique and all
edges from the first $\ell$ vertices to any other vertex in $Q$ are
present.

Recall that $\calH_{d}$ is the family of graphs on $k$ labeled
vertices with a vertex cover of size at most~$d$, and that
$\calH_d(e) \subseteq \calH_d$ denotes the set of graphs in the
$(d,e)$-boundary, as defined in \cref{sec:cores-bounds}.
Finally, for two graphs $H$ and $H'$ defined
over the same vertex set $V$ we denote by $H \setminus H'$ the graph
over $V$ containing all edges of $H$ that are not present in $H'$.

We prove \cref{lem:rect-good-full} in three steps. First we split the sum
of Fourier characters into two parts: the \emph{main sum} and some
\emph{boundary sums}. In a second step we show that the boundary sums
are negligible, i.e., that they add up to very little.
In the final
step we then show that the main sum is tightly concentrated around the
expected value.

The following claim splits the sum of Fourier characters into the main
sum and several boundary sums. We postpone the straightforward 
proof until after we prove \cref{lem:rect-good-full}. 

\begin{claim}
  \label{clm:split-sum}
  For $Q$ as in \cref{lem:rect-good-full} and any tuple $t \in Q$ it
  holds that
  \begin{align*}
    \begin{split}
      \sum_{H \in \calH_d}
      \chi_{H(t)}(G)
      =
      p^{-|S_{[\ell]} |}
      &\sum_{
      \substack{
      H \in \calH_d\\
      S_{[\ell]}  \subseteq H
      }
      }
      \chi_{(H\setminus S_{[\ell]})(t)}(G)\\
      &\quad{}+
      \sum_{i\in[\ell]}
      \sum_{j=i+1}^{k}
      p^{-|S_{[i-1]} \cup S_i^{\leq j-1}|}
      \sum_{
      \substack{
      H \in \calH_d(\set{i,j})\\
      S_{[i-1]} \cup S_i^{\leq j-1} \subseteq H 
      }
      }
      \chi_{(H\setminus (S_{[i-1]} \cup S_i^{\leq j-1}))(t)}(G)
      \eqperiod
    \end{split}
  \end{align*}
\end{claim}

The sum with the factor $p^{-|S_{[\ell]} |}$ in \cref{clm:split-sum}
is the so-called \emph{main sum}. Intuitively most of the measure is
in the main sum and it adds up to approximately $p^{-|S_{[\ell]} |}$ 
times the size of $Q$, that is,
\begin{align}
  \label{eq:1}
  p^{-|S_{[\ell]}|}
  \sum_{t \in Q}
  \sum_{
  \substack{
  H \in \calH_d\\
  S_{[\ell]}  \subseteq H
  }
  }
  \chi_{(H\setminus S_{[\ell]})(t)}(G)
  \approx
  p^{-|S_{[\ell]}|} \cdot |Q| \eqperiod
\end{align}
The latter sums in \cref{clm:split-sum} with coefficients
$p^{-|S_{[i-1]} \cup S_i^{\leq j-1}|}$ are the so-called
\emph{boundary sums}. All of these sums turn out to be tiny, that is,
for $i\in [\ell]$ and $j +1 \le i \le k$ it holds that
\begin{align}
  \label{eq:2}
  p^{-|S_{[i-1]} \cup S_i^{\leq j-1}|}
  \Big|
  \sum_{t\in Q}
  \sum_{
  \substack{
  H \in \calH_d(\set{i,j})\\
  S_{[i-1]} \cup S_i^{\leq j-1} \subseteq H
  }
  }
  \chi_{(H\setminus (S_{[i-1]} \cup S_i^{\leq j-1}))(t)}(G)
  \Big|
  \lesssim
  p^{-|S_{[\ell]}|}\cdot
  |Q| \cdot n^{-(d-\ell)} \eqperiod
\end{align}
Both of the bounds corresponding to \refeq{eq:1} and \refeq{eq:2}, and
stated formally below, are proven in \cref{sec:main-sum}.
\cref{clm:split-sum} together with these bounds will allow us to
conclude that the measure $\mu(Q)$ is tightly concentrated around
$p^{-|S_{[\ell]}|} \cdot |Q|$.

Let us first state the bound for the main sum. Intuitively we may
think of the below lemma as a version of \cref{lem:whole-space-again}
that holds on a local part of the graph.

\begin{restatable}{restatablelemma}{betterRectangle}
  \label{lem:rect-better}
  For $G$ as in \cref{lem:rect-good-full} the following holds for
  $p = n^{-2/D}$ and $\ell < d \le \eta D \le 2 \eta \log n$. If
  $ s \ge 10\,k^{13}n^{48\eta+\eps} \log n$, then all
  $(s, 1/k, p, d, [\ell])$-good rectangles $Q$ satisfy
  \begin{align*}
    \Big|
    |Q| -
    \sum_{t \in Q}
    \sum_{\substack{H \in \calH_{d}\\S_{[\ell]} \subseteq H}}
    \chi_{(H\setminus S_{[\ell]})(t)}(G)
    \Big|
    \le
    O\big(
    |Q| \,
    n^{-\eps/4}
    \big)
    \eqperiod
  \end{align*}
\end{restatable}

On the other hand, reminiscent of the edge axioms, we have that the
boundary sums are quite small. Note that, in contrast to
\cref{lem:bound-any-rectangle}, the bound below depends on the size of the
rectangle $Q$.

\begin{restatable}{restatablelemma}{boundarySums}
  \label{lem:boundary}
  For $G$ as in \cref{lem:rect-good-full} the following holds for any
  constant $0 < \eta < 1/50$ and $\eps > 0$, for $p = n^{-2/D}$ and
  $\ell < d \le \eta D \le 2\eta\log n$. Let $D \le k \le n$,
  $s \ge 10\,k^2n^{48\eta + \eps} \log n$ and $i \in [\ell]$. Then all
  $(s, 1/k, p, d, [\ell])$-good rectangles $Q$ satisfy that if
  $j \le d+2$, then
  \begin{align*}
    p^{-|S_{[i-1]} \cup S_i^{\leq j-1}|} \cdot
    \Big|
    \sum_{t \in Q}
    \sum_{
    \substack{
    H \in \calH_d(\set{i,j})\\
    S_{[i-1]} \cup S_i^{\leq j-1} \subseteq H
    }
    }
    \chi_{(H\setminus (S_{[i-1]} \cup S_i^{\leq j-1}))(t)}(G)
    \Big|
    \le
    O
    \big(
    p^{-|S_{[\ell]}|} \,
    |Q| \,
    n^{-\eps(d-\ell)/4}
    \big)\eqcomma
  \end{align*}
  and, on the other hand, if $j \ge d+3$, then the above sum is empty.
\end{restatable}

Assuming these statements the proof of \cref{lem:rect-good-full} boils down
to a sequence of syntactic manipulations.

\begin{proof}[Proof of \cref{lem:rect-good-full}]
  According to \cref{clm:split-sum} the expression
  \begin{align}\label{eq:good-init}
    \Big|
    p^{-|S_{[\ell]} |} \cdot |Q| - 
    \sum_{t \in Q}
    \sum_{H \in \calH_d}
    \chi_{H(t)}(G)
    \Big|
  \end{align}
  is equal to
  \begin{align}
    \begin{split}\label{eq:expanded-main-boundary}
      \Big|
      p^{-|S_{[\ell]} |}  \cdot |Q|
      &-p^{-|S_{[\ell]}|} 
        \sum_{t \in Q}
        \sum_{
        \substack{
        H \in \calH_d\\
        S_{[\ell]}  \subseteq H
        }
        }
        \chi_{(H\setminus S_{[\ell]})(t)}(G)
        \\
      &\qquad{}\quad{}\,
        -
        \sum_{i\in[\ell]}
        \sum_{j=i+1}^{k}
        p^{-|S_{[i-1]} \cup S_i^{\leq j-1}|}
        \sum_{t \in Q}
        \sum_{
        \substack{
        H \in \calH_d(\set{i,j})\\
        S_{[i-1]} \cup S_i^{\leq j-1} \subseteq H
        }
        }
        \chi_{(H\setminus (S_{[i-1]} \cup S_i^{\leq j-1}))(t)}(G)
        \Big| \eqperiod
    \end{split}
  \end{align}
  Appealing to the triangle inequality, \cref{lem:rect-better} and
  \cref{lem:boundary}, we may upper bound
  \refeq{eq:expanded-main-boundary} by
  \begin{align}
    \begin{split}
        p^{-|S_{[\ell]} |} \cdot
        \Big| |Q| -
        \sum_{t \in Q}
        &\sum_{
        \substack{
        H \in \calH_d\\
        S_{[\ell]}  \subseteq H
        }
        }
        \chi_{(H\setminus S_{[\ell]}) (t)}(G)
        \Big|\\
      &\quad{}
        +
        \sum_{i\in[\ell]}
        \sum_{j=i+1}^{d+2}
        p^{-|S_{[i-1]} \cup S_i^{\leq j-1}|}
        \cdot
        \Big|
        \sum_{t \in Q}
        \sum_{
        \substack{
        H \in \calH_d(\set{i,j})\\
        S_{[i-1]} \cup S_i^{\leq j-1} \subseteq H
        }
        }
        \chi_{(H\setminus (S_{[i-1]} \cup S_i^{\leq j-1}))(t)}(G)
        \Big|\\
    &\le
      O\Big(
      p^{-|S_{[\ell]} |}
      \,
      |Q|
      \,
      \big(
      n^{-\eps/4} +
      \sum_{i\in[\ell]}
      \sum_{j=i+1}^{d+2}
      n^{-\eps(d-\ell)/4}
      \big)
      \Big)
    \end{split}
    \\ \label{eq:good-final-bound}
    &\le
      O
      \big(
      p^{-|S_{[\ell]} |}
      \,
      |Q|
      \,
      n^{-\eps/8}
      \big)\eqperiod
  \end{align}
  Putting everything together we may conclude that
  \refeq{eq:good-init} is upper bounded by the expression in
  \refeq{eq:good-final-bound}, as claimed.
\end{proof}

We now proceed to prove \cref{clm:split-sum}.

\begin{proof}[Proof of \cref{clm:split-sum}]
  Suppose $\ell \geq 1$ and let us start by considering the edge
  $\{1,2\}$. We first observe that for any $H$ such that $\{1,2\}\in H$, 
  it holds that
  \begin{equation}
  \chi_{H(t)}(G) = \frac{1-p}{p} \chi_{(H\setminus\{1,2\})(t)}(G)
  \end{equation} 
  by definition of a
  good rectangle as every tuple $t \in Q$ contains the edge
  $\set{1,2}$. Let us write
  \begin{align}
    \sum_{H \in \calH_d}
    \chi_{H(t)}(G)
    &=
    \frac{1-p}{p} \cdot
    \sum_{\substack{H \in \calH_{d}\\ \set{1,2} \in H}}
    \chi_{(H\setminus\{1,2\})(t)}(G)
    +
    \sum_{\substack{H \in \calH_{d}\\ \set{1,2} \not \in H}}
    \chi_{H(t)}(G) \eqperiod
  \end{align}

  Note that we can partition the set
  $\{H \in \calH_{d} : \set{1,2} \not\in H \}$ into two parts: the
  first part contains all the graphs in the boundary of $\{1,2\}$,
  that is, graphs $H \in \calH_d(\set{1,2})$, and the second set
  contains all the remaining graphs. Note that graphs $H$ contained in
  the latter set satisfy that $H\cup\{1,2\} \in \calH_d$. Hence, for
  every $t\in Q$, it holds that
  \begin{align}
    \begin{split}
        \sum_{H \in \calH_d}
        \chi_{H(t)}(G)
      &=
        \frac{1-p}{p} \cdot
        \sum_{
        \substack{
        H \in \calH_{d} \\
    \set{1,2} \in H\\
    }
    }
    \chi_{(H\setminus\{1,2\})(t)}(G)\\
    &\hspace{3.5cm}+ 
    \sum_{
    \substack{
    H \in \calH_{d} \\
    \set{1,2} \in H\\
    }
    }
    \chi_{(H\setminus\{1,2\})(t)}(G)
    +
    \sum_{H \in \calH_d(\set{1,2})}
    \chi_{H(t)}(G)
    \end{split}\\
    &=
      \frac{1}{p} \cdot
      \sum_{
      \substack{
      H \in \calH_{d} \\
      \set{1,2} \in H\\
      }
      }
      \chi_{(H\setminus\set{1,2})(t)}(G)  
      +
      \sum_{H \in \calH_d(\set{1,2})}
      \chi_{H(t)}(G)
      \eqperiod
  \end{align}
  If we continue to rewrite the first sum in the above fashion for edges
  $\set{1,3}, \ldots, \set{1,k}$ we obtain
  \begin{align}
    \label{eq:clique-main}
    \begin{split}
    \sum_{H \in \calH_d}
    \chi_{H(t)}(G)
    =
      p^{-(k-1)} \cdot 
      \sum_{
      \substack{
      H \in \calH_d\\
      S_1\subseteq H
      }
      }
      &\chi_{(H\setminus S_1)(t)}(G)\\
      &\qquad+ \sum_{j=2}^k p^{-(j-2)}  \cdot 
      \sum_{
      \substack{
      H \in \calH_d(\set{1,j})\\
      S_1^{\leq j-1} \subseteq H
      }
      }
      \chi_{(H\setminus S_1^{\leq j-1})(t)}(G) \eqperiod
    \end{split}
  \end{align}
  By iteratively applying the above arguments to vertices
  $i=2, \ldots, \ell$ we conclude that
  \begin{align}
    \begin{split}
      \sum_{H \in \calH_d}
      \chi_{H(t)}(G)
      &=
        p^{-|S_{[\ell]} |} \cdot
        \sum_{
        \substack{
        H \in \calH_d\\
        S_{[\ell]}  \subseteq H
        }}
        \chi_{(H\setminus S_{[\ell]})(t)}(G)\\
      &\qquad{}+
        \sum_{i\in[\ell]}
        \sum_{j=i+1}^{k}
        p^{-|S_{[i-1]} \cup S_i^{\leq j-1}|}
        \sum_{
        \substack{
        H \in \calH_d(\set{i,j})\\
        S_{[i-1]} \cup S_i^{\leq j-1} \subseteq H
        }}
        \chi_{(H\setminus (S_{[i-1]} \cup S_i^{\leq j-1}))(t)}(G) \eqcomma
    \end{split}
  \end{align}
  as claimed.
\end{proof}

\subsection{Concentration of the Main Sum and Bounding the Boundary Sums}
\label{sec:main-sum}

In this section we prove \cref{lem:rect-better} and
\cref{lem:boundary}. We rely on the following lemma which guarantees
bounded character sum on arbitrary good rectangles, regardless of
size. This is in contrast to \cref{it:bounded-char-tight} of
\cref{def:well-behaved} which only holds for small rectangles. The
more general statement follows from the latter by splitting any large
rectangle into many small rectangles while maintaining goodness so
that \cref{it:bounded-char-tight} of \cref{def:well-behaved} holds for
these small rectangles.

\begin{restatable}{restatablelemma}{concentrationLemma}
  \label{lem:good-single-sum}
  Let $n, k \in \N$ and $s, \eta, D, d \in \R^+$ be such that
  $d \le \eta D$, $k \le n$ and $s \ge 9\,k^2 n^{2\eta} d \ln n$. If
  $p = n^{-2/D}$, $G$ is a $D$-well-behaved graph and $F$ is a core
  graph, then all $(s, 1/k, p, d, R)$-good rectangles $Q$ satisfy, for
  $B = [k] \setminus R$, that
  \begin{align*}
    \Big|
		\sum_{t \in Q}
    \sum_{H \in \calH(F, E^\star_F[B])}
    \chi_{H[B](t)}(G)
    \Big|
    =
    O
    \big(
    |Q| \cdot
    p^{-{|E(F[B])|}} \cdot
    (s/10\,k\log n)^{-\vc(F[B])/4}
    \big)
    \eqperiod
  \end{align*}
\end{restatable}

We defer the proof of this lemma to \cref{sec:bound-dep}. Before we
proceed to prove \cref{lem:rect-better} and \cref{lem:boundary} let us
record two claims regarding vertex covers.

\begin{claim}\label{clm:vc-deg}
  Any minimum vertex cover of a graph $H$ contains all vertices of
  degree at least $\vc(H) + 1$.
\end{claim}

\begin{proof}
  Any vertex cover of $H$ that does not contain a vertex $u$ of degree
  at least $\vc(H) + 1$ must contain all neighbors of $u$ and hence is
  of size at least $\vc(H) + 1$.
\end{proof}

\begin{claim}\label{clm:vc-star}
  For $i \in [k-1]$ and a graph $H$ on $k$ vertices it holds that if
  $S_{[i]} \subseteq H$, then $\vc(H \setminus S_{[i]}) = \vc(H) - i$.
\end{claim}
\begin{proof}
  If $i = k-1$, then $H = S_{[i]}$ is the complete graph on $k$
  vertices and the statement readily follows. Otherwise, as every
  vertex in $[i]$ has degree $k - 1 \ge \vc(H)+1$, it follows by
  \cref{clm:vc-deg} that the set $[i]$ is contained in any minimum vertex
  cover $\lexvc$ of $H$. Since the vertices in $\lexvc \setminus [i]$ have to cover
  the edges in $H \setminus S_{[i]}$, and since any vertex cover $W'$
  of $H \setminus S_{[i]}$ is such that $W'\cup [i]$ is a vertex cover of $H$,
  we may conclude
  $\vc(H \setminus S_{[i]}) = \big|\lexvc \setminus [i]\big| = \vc(H) - i$.
\end{proof}

Finally we also need to revisit our map $\core$ as we need a good
bound on the number of cores that the graphs containing $S_{[\ell]}$
are mapped to. Before stating the claim let us recall the construction
of the map $\core$, as done in \cref{sec:core-proofs}: given a graph
$H$ with lexicographically first vertex cover $\lexvc$ we let $U_1$ be
the lex first maximal set of vertices with a matching from $U_1$ to
$\lexvc$ that covers all vertices in $U_1$. Similarly we let $U_2$ be
the lex first maximal set of vertices in $H \setminus U_1$ with a
matching from $U_2$ to $\lexvc$ covering all vertices in $U_2$.
Finally we define $\core(H) = H[\lexvc \cup U_1 \cup U_2]$.

\begin{lemma}\label{clm:F-star}
  Let $k, \ell, w \in \N$. If $\calF_{\ell}(w)$ denotes the set of
  cores $F \in \img(\core)$ such that
  $S_{[\ell]} \subseteq F \cup E^\star_F$ and $\vc(F) = \ell + w$,
  then it holds that
  \begin{align*}
    |
    \calF_{\ell}(w)
    |
    \leq
    2^{3w(\ell + w + \log k)} \eqperiod
  \end{align*}
\end{lemma}

\begin{proof}
  We argue that we can encode elements of the set of cores
  $\calF_{\ell}(w)$ by using few bits.  Fix a core
  $F \in \calF_{\ell}(w)$ and let $\lexvc, U_1, U_2$ as in the
  construction of $F$ (see \cref{sec:core-proofs} and the discussion
  above). By definition we have that $|\lexvc| = \ell + w$ and
  $|U_2| \leq |U_1| \leq \ell + w$.

  According to \cref{clm:vc-deg} every vertex in $[\ell]$ is in any
  minimum vertex cover of $F$ and thus contained in $\lexvc$; the lex
  first one. Hence there are only $w$ vertices in
  $\lexvc \cap \set{\ell+1, \ldots, k}$ to be specified. We spend
  $\log \binom{k}{w}$ bits to specify this set. Once we have specified
  this set we know that the lex smallest $\ell$ vertices outside
  $\lexvc$ are contained in $U_1$ by construction. With another
  $\log \binom{k}{w}$ bits we may thus encode the remaining vertices
  of $U_1$. Similarly, for $U_2$, we know that the lex smallest $\ell$
  vertices outside $\lexvc \cup U_1$ are in $U_2$. Hence
  $\log \binom{k}{w}$ bits suffice to specify $U_2$.

  At this point we know the relevant vertices of the core. It remains
  to encode the edges. Some edges are given: all edges from $[\ell]$
  to the rest of the core are present. As such there are at most
  $w \cdot 3(\ell + w)$ many edges left to be specified. The claim follows.
\end{proof}

With these claims at hand we are ready to prove \cref{lem:rect-better}
and \cref{lem:boundary}. We start with \cref{lem:rect-better} restated
here for convenience.

\betterRectangle*

\begin{proof}
  Let $B = \set{\ell +1, \ldots, k}$ and observe that 
  $H \setminus S_{[\ell]} = H[B]$. Further using that for
  $H = S_{[\ell]}$ it holds that
  $\sum_{t\in Q} \chi_{(H \setminus S_{[\ell]})(t)}(G) = |Q|$ we may
  bound
  \begin{align}
    \Big|
    |Q| -
    \sum_{t \in Q}
    \sum_{\substack{H \in \calH_d\\S_{[\ell]} \subseteq H}}
    \chi_{(H \setminus S_{[\ell]})(t)}(G)
    \Big|
    &=
      \Big| \label{eq:Qsimplify}
      \sum_{t \in Q}
      \sum_{w = \ell+1}^{d}
      \sum_{\substack{F \in \img(\core)\\ \vc(F) = w}}
      \sum_{
      \substack{
        H \in \calH(F, E^\star_F)\\
        S_{[\ell]} \subseteq H
      }
      }
      \chi_{H[B](t)}(G)
      \Big|\\
    &=\label{eq:Qsimplify-1}
      \Big|
      \sum_{w_B = 1}^{d-\ell}
      \sum_{\substack{F \in \img(\core)\\ S_{[\ell]} \subseteq F \cup E^\star_F\\ \vc(F[B]) = w_B}}
      \sum_{t \in Q}
      \sum_{
        H \in \calH(F, E^\star_F[B])
      }
      \chi_{H[B](t)}(G)
      \Big|\\
    &\le\label{eq:Qsimplify-2}
      \sum_{w_B = 1}^{d-\ell}
      \sum_{
      \substack{
      F \in \img(\core)\\
      S_{[\ell]} \subseteq F \cup E^\star_F\\
      \vc(F[B]) = w_B
      }
      }
      \Big|
      \sum_{t \in Q}
      \sum_{
        H \in \calH(F, E^\star_F[B])
      }
      \chi_{H[B](t)}(G)
      \Big|
    \eqcomma
  \end{align}
  where we appeal to \cref{clm:vc-star} to obtain
  \refeq{eq:Qsimplify-1} and to the triangle inequality for
  \refeq{eq:Qsimplify-2}.
  Since by \cref{lem:compression} it holds that
  $\big|V\bigl(E(F)\bigr)\big| \leq 3 \cdot \vc(F)$ we may apply
  \cref{lem:good-single-sum} with
  $\big|E\bigl(F[B]\bigr)\big| \le 3 \cdot \vc(F) \cdot
  \vc\bigl(F[B]\bigr) \le 3d \cdot \vc\bigl(F[B]\bigr)$ to the inner
  expressions to obtain the bound
  \begin{align}
    \Big|
    |Q| -
    \sum_{t \in Q}
    \sum_{\substack{H \in \calH_d\\S_{[\ell]} \subseteq H}}
    \chi_{H[B](t)}(G)
    \Big|
    &\le
      |Q| \cdot
      \sum_{w_B = 1}^{d-\ell}
      \sum_{\substack{F \in \img(\core)\\ S_{[\ell]} \subseteq F \cup E^\star_F\\ \vc(F[B]) = w_B}}
      O\big(
      p^{-3d \cdot w_B} \cdot
      (s/10\,k\log n)^{-w_B/4}
      \big) \eqperiod
  \end{align}
  Applying \cref{clm:F-star} along with the fact that
  $\ell + w_B\leq d \leq \etanew D \leq 2\etanew \log n$ allows us to
  conclude
  \begin{align}
    \Big|
    |Q| - 
    \sum_{t \in Q} 
    \sum_{\substack{H \in \calH_d \\ S_{[\ell]} \subseteq H}}
    \chi_{H[B](t)}(G)
    \Big|
    &\le
      O\big(
      |Q| \cdot
      \sum_{w_B = 1}^{d-\ell}
      2^{3w_B(d + \log k)} \cdot 
      (s/10\,n^{24 \etanew}k\log n)^{-w_B/4}
      \big)\\
    &\le
      O
      \big(
      |Q|
      \cdot
      \sum_{w_B = 1}^{d-\ell}
      (s/10\,k^{13}n^{48 \eta}\log n)^{-w_B/4}
      \big)\\
    &\le
      O
      \big(
      |Q|
      \cdot
      n^{-\eps/4}
      \big)
      \eqcomma
  \end{align}
  where we used that $\ell < d$, the bound
  $ s \ge 10\, k^{13}n^{48\eta+\eps} \log n$ and the fact that the final
  sum is a geometric series. This completes the proof of
  \cref{lem:rect-better} modulo \cref{lem:good-single-sum}.
\end{proof}

\label{sec:boundary-sum}

In order to show that the magnitude of the boundary sums is small, we
need to bound the number of core graphs that contain the edge set
$S_{i-1}\disjointunion S_i^{\leq j-1}$. The bound from
\cref{clm:F-star} turns out to be insufficient as it is with respect
to $d - (i-1)$, that is, the vertex cover size outside the set
$[i-1]$; as we can apply \cref{lem:good-single-sum} with
$B = \set{\ell+1, \ldots, k}$ only, we get concentration with respect
to the size of the vertex cover in the set $\set{\ell+1, \ldots,
  k}$. This is a problem as the vertex cover in the set
$\set{\ell+1, \ldots, k}$ may be considerably smaller than in the set
$\set{i, \ldots, k}$. We overcome this difference in parameters by
leveraging the difference in the exponential factors
$p^{|S_{[i-1]} \cup S_i^{\leq j-1}|}$ and $p^{|S_{[\ell]}|}$ as
follows.

\newcommand{\setB}{B}
\newcommand{\sizeinC}{c}

\begin{claim}\label{claim:bound-F-boundary-sum}
  Let $w_B, d, i, j,\ell \in [k]$, $i\leq \ell$,
  $\setB = \set{\ell + 1, \ldots, k}$,
  $E_0 = S_{[i-1]} \disjointunion S_i^{\leq j-1}$ and
  $E_1 = S_{[\ell]} \setminus E_0$. It holds that
  \begin{align*}
    \sum_{
    \substack{
    F \in \img(\core)\\
    F \in \calH_d(\set{i,j})\\
    E_0 \subseteq F \cup E^\star_F\\
    \vc(F[\setB]) = w_B
    }
    }
    2^{-|E_1 \setminus E^\star_F|}
    &
      \leq
      3d^2
      \cdot
      \binom{k}{d}^3
      \cdot
      2^{- k + 3dw_B+ 4d + \ell} \eqperiod
  \end{align*}
\end{claim}

\begin{proof}
  We first argue that the set over which we sum is not so large and 
  then provide a lower bound on $|E_1 \setminus E^\star_F|$. These
  two bounds will together imply the claim.
  
  Let $C = \set{i, \ldots, k}$. Given parameters $w_C\leq d$ and
  $\sizeinC \leq 3d$, we consider cores $F \in \img(\core)$ such that
  $F \in \calH_d(\set{i,j})$, $E_0 \subseteq F \cup E^\star_F$,
  $\vc(F[\setB]) = w_B$, $\vc(F[C]) = w_C$ and
  $|V\bigl(E(F)\bigr) \cap C| = \sizeinC$.
  We claim that there are at most
  \begin{align}\label{eq:Fcount}
    \binom{k}{d}^3 \cdot 2^{\binom{\sizeinC}{2} - \binom{\sizeinC-w_C}{2}} 
  \end{align}
  such core graphs.  Indeed, it is enough to specify the $w_C\leq d$
  vertices in a minimum vertex cover of $F[C]$, the other at most $2d$
  vertices in $V\bigl(E(F)\bigr) \cap C$, and the edges between
  vertices in the identified minimum vertex cover of $F[C]$ and
  $V\bigl(E(F)\bigr) \cap C$.

  In order to obtain a lower bound on $|E_1 \setminus E^\star_F|$,
  note that the edge set $E^\star_F$ consists only of edges with
  exactly one endpoint in $V\bigl(E(F)\bigr)$. Hence all edges that
  are in $E_1$ and that have both endpoints in $V\bigl(E(F)\bigr)$ are
  also in $E_1 \setminus E^\star_F$. These include edges between
  $V\bigl(E(F)\bigr) \cap \set{i+1, \ldots, \ell}$ and
  $V\bigl(E(F)\bigr) \cap C$.  Note that all edges in
  $F[C] \setminus F[\setB]$ have one endpoint in
  $V\bigl(E(F)\bigr) \cap \set{i, \ldots, \ell}$ and thus
  \begin{align}
        |V\bigl(E(F)\bigr) \cap  \set{i, \ldots, \ell}|
      \ge
        \vc(F[C] \setminus F[\setB])
      \ge
        \vc(F[C]) - \vc(F[\setB])
      =
        w_C - w_B \eqperiod
  \end{align}
  This implies that $ |V\bigl(E(F)\bigr) \cap  \set{i+1, \ldots, \ell}| \ge w_C - w_B - 1$ and therefore
  \begin{align}
    |E_1 \setminus E^\star_F|
    \geq
    \binom{\sizeinC}{2} - \binom{\sizeinC - (w_C-w_B-1)}{2}
    \eqperiod
  \end{align}
  This bound, however, is not good enough for us. We need to also
  count edges in $E_1 \setminus E^\star_F$ that are adjacent to $i$
  and that we have not considered yet.  Note that these include all
  the edges in $(E_1 \setminus E^\star_F) \cap (\set{i} \times B)$.
  Let $\widetilde{N}(i)$ be the neighbors of $i$ in the graph
  $F \cup E^\star_F$.  In particular, we have that
  $[j-1] \subseteq \widetilde{N}(i)$ since
  $S_i^{\leq j-1} \subseteq E_0 \subseteq F \cup E^\star_F$.  This
  implies that
  $\set{i} \times (B\setminus \widetilde{N}(i)) \subseteq E_1
  \setminus E^\star_F$ and therefore
  $|(E_1 \setminus E^\star_F) \cap (\set{i} \times B)| \geq |\set{i}
  \times (B\setminus \widetilde{N}(i))| \geq |B| -
  |\widetilde{N}(i)|$.  Now note that it must be the case that
  $\widetilde{N}(i)$ is in every minimum vertex cover of $F$;
  otherwise we would have a minimum vertex cover of $F$ that either
  contains $i$ (and then $F$ would not be in the
  $(d, \set{i,j})$-boundary) or is not a vertex cover of the graph
  $F \cup E^\star_F$.  This implies that $|\widetilde{N}(i)| \leq d$.
  We may thus conclude that
    \begin{align}\label{eq:E1bound}
    |E_1 \setminus E^\star_F|
    \geq
    \binom{\sizeinC}{2} -
    \binom{\sizeinC - (w_C-w_B-1)}{2} +
    k - \ell -
    d
    \eqperiod
  \end{align}

  Using the bounds in~\cref{eq:Fcount} and~\cref{eq:E1bound} we may
  write
  \begin{align}
    \sum_{\mathclap{
    \substack{
    F \in \img(\core)\\
    F \in \calH_d(\set{i,j})\\
    E_0 \subseteq F \cup E^\star_F\\
    \vc(F[\setB]) = w_B\\
    \vc(F[C]) = w_C\\
    |V(E(F)) \cap C| = \sizeinC
    }
    }
    } \,
    2^{-|E_1 \setminus E^\star_F|}
    &\leq
    \binom{k}{d}^3
    \cdot
    2^{\binom{\sizeinC}{2} - \binom{\sizeinC - w_C}{2}}
    \cdot
    2^{- \binom{\sizeinC}{2} + \binom{\sizeinC - (w_C-w_B-1)}{2} - k +\ell + d}\\
    &=
    \binom{k}{d}^3
    \cdot
    2^{(\sizeinC - w_C)(w_B +1) - \binom{w_B+1}{2} - k + \ell + d} 
    \eqperiod 
  \end{align}
  Since $\sizeinC \leq 3d$, we have that
  $(\sizeinC - w_C)(w_B +1)- \binom{w_B +1}{2} \le 3d(w_B+1)$.  Finally, we need to sum over
  all possible values of $w_C$ and $\sizeinC$. Since $w_C \leq d$ and
  $\sizeinC \leq 3d$, the statement follows.
\end{proof}

In the remainder of this section we prove that boundary sums are
small. For convenience we restate the claim.
\pagebreak
\boundarySums*

\begin{proof}
Let us first consider the case $j \ge d+3$. For all such $j$ it holds
that all graphs $H$ we sum over contain the edges $S_{[i-1]} \cup
S_i^{\le d+2}$. As vertex $i$ has degree at least $d+1$ and $\vc(H) =
d$, according to \cref{clm:vc-deg}, it holds that $i$ is contained in
any minimum vertex cover of $H$. This implies that there are no graphs
$H$ satisfying $S_{[i-1]} \cup S_i^{\le d+2} \subseteq H$ that are
also in the $(d, \set{i,j})$-boundary. Thus the considered sum is
empty for $j \ge d+3$.

\newcommand{\valueT}{T}
\newcommand{\setC}{B}

We now assume $j \le d+2$.
We can also assume $j>i$ as otherwise the edge $\set{i,j}$ is in any
graph we consider and hence cannot be a boundary edge. Recall that
$i\le \ell$, and let $E_0 = S_{[i-1]} \cup S_i^{\leq j-1}$,
$E_1 = S_{[\ell]} \setminus
E_0$ 
and $\setC = \set{\ell + 1, \ldots k}$.
Similar to the proof of
\cref{lem:rect-better}, though this time we rely on
\cref{prop:core-boundary}, we may rewrite the boundary sum as
\begin{align}
    \sum_{t \in Q}
    \sum_{
    \substack{
    H \in \calH_d(\set{i,j})\\
    E_0 \subseteq H 
    }
    }
    \chi_{(H\setminus E_0)(t)}(G) 
    &=
    \sum_{t \in Q}
    \sum_{
    \substack{
    F \in \img(\core)\\
    F \in \calH_d(\set{i,j})
    }
    }
    \sum_{
    \substack{
    H \in \calH(F, E^\star_F)\\
    E_0 \subseteq H 
    }
    }
    \chi_{(H \setminus E_0)(t)}(G)
  \\
  &=
  \sum_{
  \substack{
  F \in \img(\core)\\
  F \in \calH_d(\set{i,j})\\
  E_0\subseteq F \cup E^\star_F
  }
  }
  \sum_{t \in Q}
  \sum_{
  H \in \calH(F, E^\star_F \setminus E_0)
  }
  \chi_{(H \setminus E_0)(t)}(G) 
  \eqperiod
\end{align}
We now focus on the inner two sums.  Since $S_{[\ell]} = E_0
\disjointunion E_1$ we have that $\chi_{(H \setminus E_0)(t)}(G) =
\chi_{(H \cap E_1)(t)}(G) \cdot \chi_{(H \setminus
  S_{[\ell]})(t)}(G)$.  By definition of a good rectangle, for all $t
\in Q$, the edges $S_{[\ell]}(t)$ are present in $G$ thus
$\chi_{(H \cap E_1)(t)}(G)
=
\left((1-p)/p\right)^{|H \cap E_1|}$.
Let $\valueT = |E^\star_F \cap E_1|$.  Using that $H\setminus
S_{[\ell]} = H[\setC]$ and
$E^\star_F\setminus S_{[\ell]} = E^\star_F[\setC]$
we can derive that
\begin{align}
    \sum_{t \in Q}
    \sum_{
    H \in \calH(F, E^\star_F \setminus E_0)
    }
    \chi_{(H \setminus E_0)(t)}(G)
  &=
  \sum_{t \in Q}
  \sum_{
  H \in \calH(F, E^\star_F \setminus E_0)
  }
  \chi_{(H \cap E_1)(t)}(G) \cdot  \chi_{(H \setminus S_{[\ell]})(t)}(G) 
  \\
    &=
	 \sum_{t \in Q}
	 \sum_{
	 H \in \calH(F, E^\star_F \setminus S_{[\ell]})
	 }
	\sum_{\nu = 0}^{\valueT}
	\binom{\valueT}{\nu} \left(\frac{1-p}{p}\right)^{\nu}
	\chi_{H\setminus S_{[\ell]}(t)}(G)
    \\
    &=
	\sum_{t \in Q}
	\sum_{
	H \in \calH(F, E^\star_F[\setC])
	}
	\chi_{H[\setC](t)}(G)
	\sum_{\nu = 0}^{\valueT}
	\binom{\valueT}{\nu} \left(\frac{1-p}{p}\right)^{\nu}
	\\
	&=
	p^{-\valueT}
	\sum_{t \in Q}
	\sum_{
	H \in \calH(F, E^\star_F[\setC])
	}
   \chi_{H[\setC](t)}(G) \eqperiod
\end{align}
Appealing to the triangle inequality we may thus bound
\begin{align}
  \begin{split}
  \Big|
    \sum_{t \in Q}
    \sum_{
    \substack{
    H \in \calH_d(\set{i,j})\\
    E_0 \subseteq H 
    }
    }
    \chi_{(H\setminus E_0)(t)}(G)
  \Big|
  &= 
 \Big|
  \sum_{
  \substack{
  F \in \img(\core)\\
  F \in \calH_d(\set{i,j})\\
  E_0\subseteq F \cup E^\star_F
  }
  }
	p^{-|E^\star_F \cap E_1|}
	\sum_{t \in Q}
	\sum_{
	H \in \calH(F, E^\star_F[\setC])
	}
	\chi_{H[\setC](t)}(G) 
  \Big|
  \end{split}\\
  &\le
  \sum_{
  \substack{
  F \in \img(\core)\\
  F \in \calH_d(\set{i,j})\\
  E_0 \subseteq F \cup E^\star_F
  }
  }
	p^{-|E^\star_F \cap E_1|}
    \Big|
	\sum_{t \in Q}
	\sum_{
	H \in \calH(F, E^\star_F[\setC])
	}
	\chi_{H[\setC](t)}(G) 
    \Big| \eqperiod 
   \label{eq:boundary-restated}
\end{align}
Fix $F$ and let $w_B =  \vc(F[\setC])$.
Note that the characters we sum in
\cref{eq:boundary-restated} solely depend on edges between blocks of
size at least $s$. We may thus apply \cref{lem:good-single-sum} to the
inner expressions and, using that
$p^{-|E(F[\setC])|}\leq p^{-{3 \cdot \vc(F) \cdot \vc(F[\setC])}} = p^{-{3dw_B}}
\leq n^{6\eta w_B}$ and $s \ge 10\,k^2n^{48\eta + \eps} \log n$, we obtain
\begin{align}\label{eq:boundary-single}
  \Big| 
  \sum_{t \in Q}
  \sum_{
  H \in \calH(F, E^\star_F[\setC])
  } 
  \chi_{H[\setC](t)}(G)
  \Big|
  &\le h(w_B) =
    O\big(|Q| \cdot (n^{24\etanew + \eps} k)^{-w_B/4}\big)
    \eqperiod
\end{align}

From \cref{clm:F-star} we know that there are at most
$2^{3(d-(i-1))(d + \log k)}$ many cores we sum over. The problem is
that this bound does not depend on $w_B$ and hence applying
\cref{clm:F-star} to \cref{eq:boundary-restated} does not readily
result in the desired bound.

In order to bound the number of cores, we partition them according to
the size $w_B$ of their minimum vertex cover.  Note that by
\cref{clm:vc-deg}, if
$S_{[i-1]} \subseteq E_0 \subseteq F\cup E^\star_F$ it must be the
case $S_{[i-1]} \subseteq F$ by the definition of core graphs.
Therefore, for such graphs $F$, we have that
$\vc(F[\setC]) = \vc(F \setminus S_{[\ell]}) \leq \vc(F \setminus
S_{[i-1]}) = d - (i-1)$, where the last equality follows from
\cref{clm:vc-star}.  Using that
$|S_{[\ell]}| = |E_0| + |E_1 \cap E^\star_F | + |E_1 \setminus
E^\star_F|$ since $S_{[\ell]} = E_0 \disjointunion E_1$, we can apply
\cref{claim:bound-F-boundary-sum} to obtain that
\begin{align}
 p^{-|E_0|} \cdot \Big|
    \sum_{t \in Q}
    \sum_{
    \substack{
    H \in \calH_d(\set{i,j})\\
    E_0 \subseteq H 
  }
  }
  \chi_{(H\setminus E_0)(t)}(G)
  \Big|
  &\leq 
    p^{-|E_0|}  \cdot
    \sum_{
    \substack{
    F \in \img(\core)\\
  F \in \calH_d(\set{i,j})\\
  E_0 \subseteq F \cup E^\star_F
  }
  }
  p^{-|E^\star_F \cap E_1|} \cdot h(\vc(F[\setC]))\\
  &=
    p^{-|S_{[\ell]} |}
    \sum_{w_B=d-\ell}^{d-(i-1)}
    h(w_B)
    \cdot
    \sum_{
    \substack{
    F \in \img(\core)\\
  F \in \calH_d(\set{i,j})\\
  E_0 \subseteq F \cup E^\star_F\\
  \vc(F[\setC]) = w_B
  }
  }
	p^{|E_1\setminus E^\star_F|}\\
  &\le\label{eq:boundary-final}
    p^{-|S_{[\ell]} |}
    3d^2
    \binom{k}{d}^3
    2^{- k + 5d}
    \sum_{w_B=d-\ell}^{d-(i-1)}
    2^{3dw_B}
    \cdot
    h(w_B) \eqcomma
\end{align}
where
$h(w_B) = O \big( |Q| \cdot (n^{24\etanew + \eps} k)^{-w_B/4} \big)$,
we use that $\ell < d$ and, in order to apply
\cref{claim:bound-F-boundary-sum}, we use that $p\leq 1/2$. Note that,
as $d \leq \eta D\leq \eta k < k/50$, we have that
$\binom{k}{d}^3 \le \big(\frac{4k}{d}\big)^{3d} < 200^{3k/50} < 2^{k/2}$.
This allows us to bound
$3d^2 \binom{k}{d}^3 2^{-k+5d} = O(1)$.

Finally, because $2^{3d} \leq n^{24\eta /4}$ and the sum
in \refeq{eq:boundary-final} is a geometric series with common ratio $2^{3d}
    \cdot
     (n^{24\etanew + \eps} k)^{-1/4} \leq  (n^{\eps} k)^{-1/4}$ and coefficient
$O\big(|Q| \cdot (n^{\eps} k)^{(d-\ell)/4}\big)$ we may conclude
that
\begin{align}
   p^{-|E_0|} \cdot \Big|
    \sum_{t \in Q}
    \sum_{
    \substack{
    H \in \calH_d(\set{i,j})\\
    E_0 \subseteq H 
  }
  }
  \chi_{(H\setminus E_0)(t)}(G)
  \Big|
  \leq
  O\big(
  p^{-|S_{[\ell]} |}
  \,
  |Q|
  \,
  n^{-\eps(d-\ell)/4}
  \big) \eqcomma
\end{align}
as claimed.
\end{proof}


\subsection{Bounds for All Good Rectangles}
\label{sec:bound-dep}

This section is devoted to the proof of \cref{lem:good-single-sum}.
We rely on the following lemma.

\begin{lemma}\label{lem:split-set}
  Let $a,b,c,f \ge 1$ be integers such that $b > 12 \ln (4a)$ and let
  $\gamma$ be a positive real number such that
  $(3a \ln(4af)/c)^{1/2} < \gamma < 1$.
  Let $U$ be a set satisfying $|U| = a \cdot b$ and
  $\calF \subseteq 2^U$ be a family of subsets over $U$ of size
  $|\calF| = f$, where each $F\in\calF$ is of size at least $c$. Then
  there is a partition of $U = \bigcup_{i=1}^{a} U_i$ such that
  \begin{enumerate}
  \item $b/2 \le |U_i| \le 3b/2$ for all $1 \le i \le a$,
    and \label{item:ui}
    
  \item for each set $F \in \calF$ and $1 \le i \le a$ it holds that
    $(1 - \gamma) |F|/a \le |F \cap U_i| \le
    (1+\gamma)|F|/a$. \label{item:F}
  \end{enumerate}
\end{lemma}

\begin{proof}
  We use the probabilistic method to show that a partition 
  as claimed exists.
  Independently color each element in $U$ by a color in
  $\set{1, 2, \ldots, a}$. Let $U_i$ denote color class $i$. 

  Let us argue \cref{item:ui}. By the multiplicative Chernoff bound we
  have that for a single $i\in [a]$, \cref{item:ui} holds except with
  probability $2 \cdot \exp(-b/12)$. A union bound over the $a$ sets
  establishes that except with probability
  $2a \cdot \exp(-b/12) < 1/2$ \cref{item:ui} holds.

  To argue \cref{item:F} we again appeal to the multiplicative
  Chernoff bound to see that for a fixed~$F$ and $i$ the property
  holds except with probability
  $2 \cdot \exp(-\gamma^2 |F|/3a) \leq 2 \cdot \exp(-\gamma^2 c/3a) <
  1/2af $. A union bound over the family $\calF$ and $i \in [a]$ shows
  that \cref{item:F} holds except with probability less than~$1/2$.

  A final union bound shows the existence of a partition as claimed.
\end{proof}

As a corollary we can show that any good rectangle can be partitioned
into many small good rectangles with only a slight loss in
parameters.

\begin{corollary}
  \label{cor:cut-up-C}
  Let $n, k \in \N$ and $s,\eta,D,d \in \R^+$ be such that
  $d \leq \eta D$, $k \le n$ and
  $s = 9\,k^2 n^{2\eta} d \ln n$, let $p=n^{-2/D}$ and
  $R\subseteq [k]$.
  If $Q$ is a $(s, 1/k, p, d, R)$-good rectangle and
  $T \subseteq [k] \setminus R$, then there is a partition of $Q$
  into a set $\calQ$ of
   at most $(n/s)^{|T|}$ rectangles such that each $Q' \in \calQ$ 
  is $(s, 3/k, p, d, R)$-good and  
  satisfies $|Q'_i| \le 4 s$ for all $i \in T$.
\end{corollary}

\begin{proof}
  This is a direct consequence of \cref{lem:split-set}: fix $i \in T$
  such that $|Q_i| > 4s$. Apply \cref{lem:split-set} with $U = Q_i$,
  $b = 2s$, $a = |Q_i|/b$, $\gamma = 1/k$ and $\calF$ being the common
  neighborhoods in $Q_i$ of all tuples of size at most $d$. Thus we
  get the bound $f \le \binom{nk}{\leq d} \le (nk)^d$ and 
  that every set in $\calF$ is of size at least
  $c= (1-1/k) p^{d} |Q_i|$. 
  The bound on $s$ in the lemma statement
  has been chosen so that the conditions of \cref{lem:split-set} are
  satisfied.
  
  Thus we obtain a partition $\calQ_i$ of $Q_i$ such that each set
  $Q_i' \in \calQ_i$ satisfies $s \le |Q'_i| \le 3 s$ as
  required. Furthermore, each small tuple $t$ has a common
  neighborhood of size $(1\pm 1/k)^2p^{|t|}|Q_i'|$ in the set
  $Q_i'$. Here we used \cref{it:good-bounded} of \cref{def:good} which
  states that $Q$ has $(1/k,p)$-bounded common neighborhoods into
  $Q_i$. Because the interval $[(1-1/k)^2,(1+1/k)^2]$ is contained in
  $[1-3/k, 1+3/k]$ we have that each tuple $t$ has a $(3/k,p)$-bounded
  common neighborhood in $Q_i$.
  Applying \cref{lem:split-set} iteratively in the above manner to each
  large block $i \in T$ gives us the family $\calQ$ in the statement.
\end{proof}

We are now ready to prove \cref{lem:good-single-sum} restated here for
convenience.

\concentrationLemma*

\begin{proof}
  We first apply \cref{cor:cut-up-C} to an $(s, 1/k, p, d, R)$-good
  rectangle $Q$ with $T = [k] \setminus R$ to obtain a family $\calQ$
  of $(s, 3/k, p, d, R)$-good rectangles with the additional
  property that each block is bounded in size by $4s$.
  We can therefore write
  \begin{align}
    \Big|
    \sum_{t \in Q} \sum_{H \in \calH(F, E^\star_F[B])}
    \chi_{H[B](t)}(G)
    \Big|
    &\le
      \sum_{\tilde{Q} \in \calQ}
      \Big|
      \sum_{t \in \tilde{Q}}
      \sum_{H \in \calH(F,E^\star_F[B])}
      \chi_{H[B](t)}(G)
      \Big|\\
    &\le  \label{eq:applyProp4}
      \sum_{\tilde{Q} \in \calQ}
      O\big(
      |\tilde{Q}|
      \cdot
      p^{-{|E(F[B])|}}
      \cdot
      (s/10\,k\log n)^{-\vc(F[B])/4}
      \big)\\
    &=
      O
      \big(
      |Q| \cdot
      p^{-{|E(F[B])|}} \cdot
      (s/10\,k\log n)^{-\vc(F[B])/4}
      \big)\eqcomma
  \end{align}
  where \refeq{eq:applyProp4} follows directly 
  from \cref{it:bounded-char-tight} of \cref{def:well-behaved}
  by noting that every $\tilde{Q}\in \calQ$ is an $(s, 3/k, p, d, R)$-good rectangles
  such that $|\tilde Q_i| \le 4s$ for all $i \in [k]$, and 
  that every $|B|$-tuple $t_B \in \tilde{Q}_B$ has a unique $k$-tuple extension
  $t \supseteq t_B$ in $\tilde{Q}$.
\end{proof}

\section{Random Graphs Are Well-Behaved}
\label{sec:random-proof}

In this section we prove that graphs sampled from
$\calG(n,k,n^{-2/D})$ are asymptotically almost surely
$D$-well-behaved. For convenience we recall the precise statement
below.

\Gwellbehaved*

Recall that \cref{def:well-behaved} consists of four properties:
\cref{it:expected-neigh} states that the common neighborhood of small
tuples behave as expected. When we focus our attention on a subgraph
induced by a rectangle, we cannot expect that the common neighborhoods
in this rectangle still behave as expected---we may for example have
an isolated vertex in a rectangle. \Cref{it:bounded-error} of
\cref{def:well-behaved} guarantees thus a slightly weaker property: it
roughly states that every rectangle has only few small tuples whose
common neighborhood is ill-behaved.

While the first two properties of \cref{def:well-behaved} are of
combinatorial nature, the final two properties are more of analytical
nature. \Cref{it:bounded-char-general} states that character sums over
families $\calH(F, E^\star_F)$ are bounded for rectangles with large
enough minimum block size. Once we consider smaller rectangles, things
are not so well-behaved anymore. However, for certain rectangles, we
can still guarantee a similar property: \cref{it:bounded-char-tight}
of \cref{def:well-behaved} states that if the common neighborhoods of
small tuples in a rectangle are well-behaved, then we can still
guarantee that the mentioned character sums are bounded. In fact, this
bound is in some sense tighter than the one in
\cref{it:bounded-char-general} as it actually depends on the size of
the considered rectangle.

In \cref{sec:common-neigh} we start by proving that
\cref{it:expected-neigh,it:bounded-error} of \cref{def:well-behaved}
hold asymptotically almost surely. As previously mentioned, similar
notions have been used in previous papers on the clique formula
\cite{BIS07,BGLR12Parameterized,ABRLNR21}. Our arguments towards
\cref{it:expected-neigh,it:bounded-error} are heavily influenced by
these papers, and some are essentially identical.

Once we established the combinatorial properties, in
\cref{sec:bounded-char}, we then prove the analytical
\cref{it:bounded-char-tight,it:bounded-char-general} of
\cref{def:well-behaved}, modulo a probabilistic bound on the weighted
sums of Fourier characters. \Cref{sec:encode} is dedicated to the
proof of the probabilistic bound which, inspired by~\cite{AMP21norm},
follows via an encoding argument applied to a high moment version of
the Markov inequality.

\subsection{Random Graphs Have Bounded Common Neighborhoods}
\label{sec:common-neigh}

In this section we prove that asymptotically almost surely
$G\sim \calG(n, k, n^{-2/D})$ satisfy
\cref{it:expected-neigh,it:bounded-error} of \cref{def:well-behaved}.

\begin{lemma}
  \label{lem:neighbor}
  For any constant $\delta > 0$, any integer $D,k,n\in \N$ satisfying
  $k \le n$ and $4 \le D \le n^{\delta/2}$ the following holds for
  $\beta \ge n^{-1/4 + \delta/2}$. Except with probability
  $\exp(-\Omega(n^{\delta}))$ a graph $G \sim \calG(n, k, p)$
  has $(\beta ,p,D/4)$-bounded common neighborhoods in every block,
  for $p = n^{-2/D}$.
\end{lemma}

\begin{proof}
  Given a set of vertices $S$, let $X_v$ be the indicator random
  variable of whether a vertex $v$ is in the common neighborhood of
  $S$. Note that $\Pr[X_v = 1] = p^{|S|}$. Thus in expectation we have
  $n p^{|S|}$ many common neighbors per block. Applying the
  multiplicative Chernoff bound we get that
  \begin{align}
    \Prob{
    \big|
    \sum_{i=1}^n{X_i} -
    n\cdot p^{|S|}
    \big|
    \ge \beta n\cdot p^{|S|}
    }
    \le
    2 \cdot \exp\left(-\frac{\beta^2 n}{3 \cdot p^{-|S|}}\right).
  \end{align}
  Taking a union bound over all blocks and all sets $S$ of size less
  than or equal to $D/4$ we get
  \begin{align}
    \Pr[
    \exists S,V_i \text{ such that }
    N(S)\cap V_i \text{ has uncommon size}
    ]
    &\le
      k\cdot
      \sum_{i=1}^{D/4}{\binom{kn}{i}} 2 \cdot
      \exp\left(-\frac{\beta^2 n }{3\cdot p^{-i}}\right) \\
    &\le
      2 k \cdot
      \sum_{i=1}^{D/4}(kn)^i \cdot
      \exp\left(-\frac{n^{1/2+\delta}}{3 \cdot n^{2i/D}}\right) \eqcomma
  \end{align}
  which is exponentially small in $\Omega(n^\delta)$.
\end{proof}

While all small sets have common neighborhoods of expected size in any
block, this is not necessarily true for subsets of a block.  For
example, if $S \subseteq V_i$ consists of all non-neighbors of a
vertex $v$, then clearly $v$ does not have a neighborhood in $S$ of
expected size.  We can, nonetheless, show that for any large enough
set $S$ there is a small set of vertices $W$, which we refer to as
the \emph{error set} of $S$, such that all small sets of vertices not
intersection $W$ have a common neighborhoods in $S$ of expected
size. The following lemma works out the precise dependency between the
size of $W$ and the size of $S$.


\newcommand{\setWS}{W_{S,\ell}}
{
\begin{lemma}\label{lem:common-neigh}
  For all $k, n \in \N^+$, $k \le n$, $\gamma > 0$, and $0<p \leq 1/2$
  the following holds asymptotically almost surely for
  $G \sim \calG(n,k,p)$.
  For all $\ell \in [k]$, for all $i \in [k]$ and for every subset
  $S \subseteq V_i$ of size $s \ge 2 w$, for
  $w = 12 \ell \ln n / p^\ell \gamma^2$, there is a set $\setWS$ of
  size at most $w$ such that every tuple $t$, $|t| \le \ell$, disjoint
  from $V_i$ and $\setWS$ has a common neighborhood size in $S$ of
  expected size, that is, it holds that
  \begin{align*}
    |N^\cap(t, S)| =
    (1\pm \gamma)\E_G[|N^\cap(t,S)|] =
    (1\pm \gamma)p^{|t|}s \eqperiod
  \end{align*}
\end{lemma}
}
\begin{proof}
  Given $\ell \in [k]$ and $S\subseteq V_i$, for some $i\in[k]$, we define
  a set $\setWS$ such that for every tuple $t$ with $|t| \leq \ell$ and
  that is disjoint from $V_i$ and $\setWS$ it holds that
  \begin{align}\label{eq:pf-common-neigh}
    |N^\cap(t, S)| =
    (1\pm \gamma)p^{|t|}s \eqperiod
  \end{align}
  Our goal is then to show that the probability of the event that
  there exists an $\ell \in [k]$, $i\in [k]$ and a set
  $S\subseteq V_i$ of size at least $2w$ such that $|\setWS| > w$ is
  at most $n^{-\Omega(1)}$.
  
  The set $\setWS$ is constructed by including all tuples of that do
  not satisfy \cref{eq:pf-common-neigh}. More formally, we first
  consider each possible tuple size $a = 1, 2, \ldots, \ell $ and
  separately construct sets $W_a$ as follows.  Let
  $t_1, t_2, \ldots, t_{w_a}$ be an arbitrary but fixed maximal
  sequence of pairwise disjoint tuples, each of size~$a$, where each
  tuple either has too many or too few common neighbors in $S$, that
  is, $|N^\cap(t_j, S)| > (1 + \gamma)p^{a}s$ or
  $|N^\cap(t_j, S)| < (1 - \gamma)p^{a}s$ for all $j\in [w_a]$.  We
  define $W_a = \bigcup_{j\in[w_a]} t_j$ and
  $\setWS = \bigcup_{a\in [\ell]} W_a$.
  
  Let $\widehat{w}_a = 6\ln n/p^a\gamma^2$. Note that if
  $w_a \leq \widehat{w}_a$ for all $a\in [\ell]$, then
  $|W_a| \leq a \cdot \widehat{w}_a$ and
  \begin{align}
    |W| \leq  \frac{6\log n}{\gamma^2} \sum_{a\in [\ell]} \frac{a}{p^a} 
    \leq  \frac{6\log n}{\gamma^2} \left(\frac{2\ell}{p^\ell}\right)
    = w \eqcomma
  \end{align}
  where we use that $p\leq 1/2$.  We can therefore bound the
  probability that $|\setWS| > w$ by bounding the probability that
  there exists an $a\in [\ell]$ such that $w_a \geq \widehat{w}_a$.
  
  Fix an $a\in [\ell]$. The probability that a given tuple of $t_j$ of
  size $a$ has too many or too few common neighbors in~$S$ can be
  bounded by the multiplicative Chernoff bound to obtain that
  \begin{align}
    \Prob{\big||N^\cap(t_j,S)| - p^{a}|S|\big| > \gamma \cdot p^{a}|S|}
    &\le 2 \cdot
      \exp(-\gamma^2\cdot p^a \cdot |S|/3) \eqperiod
  \end{align}
  Note that the presence of edges between disjoint tuples $t_j$ and
  $S$ are independent events. 
  Thus we can bound the probability that there
  there are at least $\widehat{w}_a$ many disjoint tuples 
  that have  too many or too few common neighbors in~$S$ by
  \begin{align}
    2 \cdot \exp
    \bigl(
    -
    \gamma^2 \cdot p^a \cdot s \cdot \widehat{w}_a/ 3 
    \bigr)
    =
    2\cdot \exp(-2s\ln n)
    \eqperiod
  \end{align}
  
  By taking a union bound over $\ell \in [k]$, $s = 2w, \ldots, n$ and
  $a \in [\ell]$, and then for each $s$ and $a$ taking a union bound
  over $i\in [k]$, over the choices of $S\subseteq V_i$ of size $s$,
  and over the choices for the $\widehat{w}_a$ tuples (of which there
  are at most
  $\binom{kn}{a}^{\widehat{w}_a} \leq (kn)^{a\widehat{w}_a}$ ), we
  conclude that the probability that there exists an $\ell$ and a set
  $S$ of size at least $2w$ such that $\setWS > w$ is at most
  \begin{align}
    \Prob{\exists S, \ell : |\setWS| > w}
    \leq
    &\sum_{\ell \in [k]}
      \sum_{s = 2w}^{n}
      \sum_{a\in [\ell]}
      2\cdot
      \exp(-2s\ln n + a\widehat{w}_a \ln (kn) + \ln k + s\ln n) \\
    \leq
    &\sum_{\ell \in [k]}
      \sum_{s = 2w}^{n}
      2\cdot
      \exp(-(s - 1)\ln n  + \ell\widehat{w}_\ell (2\ln n) + \ln \ell) \\
    \leq
    &\sum_{\ell \in [k]}
      \sum_{s = 2w}^{n}
      2\cdot
      \exp(-(s-2-w)\ln n ) \leq n^{-\Omega(1)}
      \eqcomma
  \end{align}
  as we wished to prove.
\end{proof}

Given a large rectangle $Q$ we would like to argue that after removing
the union of all error sets $W = \bigcup_{i} W_{Q_i}$ we are left with
a rectangle in which all small tuples have common neighborhoods of
expected size. To this end we rely on the following claim.

\newcommand{\sizealpha}{b}

\begin{claim}
  \label{lem:neigh-remove-set}
  Let $\sizealpha > 0$ and $0 < \gamma < 1$. Suppose we have a
  universe $U$ and a set $S \subseteq U$ satisfying
  $\frac{|S|}{|U|} \in (1 \pm \gamma)\sizealpha$. For any set
  $T \subseteq U$ satisfying
  $\frac{|T|}{|U|} \le \min\set{\gamma/2, \sizealpha\gamma}$ it holds
  that
  \begin{displaymath}
    \frac{|S \setminus T|}{|U \setminus T|}
    \in
    (1 \pm 3\gamma)\sizealpha
    \eqperiod
  \end{displaymath}
\end{claim}

\begin{proof}
  Let us denote the size of $S$ by $s$, the size of $T$ by $t$ and the
  size of $U$ by $u$. It holds that
  \begin{align}
    \frac{|S \setminus T|}{|U \setminus T|} \le \frac{s}{u(1-\gamma/2)}
    \le (1+\gamma)^2\sizealpha \le (1 + 3\gamma)\sizealpha\eqperiod
  \end{align}
  For the lower bound observe that it holds
  \begin{align}
    \frac{|S \setminus T|}{|U \setminus T|}
    \ge \frac{s}{u-t} - \frac{t}{u-t}
    \ge \frac{s}{u} - \frac{t}{u(1-\gamma/2)}
    \ge (1-\gamma)\sizealpha - \frac{t}{u}(1+\gamma)
    \ge (1-3\gamma)\sizealpha \eqcomma
  \end{align}
  where, for the final inequality, we used that
  $t/u \le \sizealpha\gamma$.
\end{proof}

We are now ready to prove that, asymptotically almost surely,
$G\sim \calG(n, k, p)$ for $p=n^{-2/D}$, satisfies
\cref{it:bounded-error} of \cref{def:well-behaved} which states that
for all $\ell \leq D/4$, and
$s \ge C k^4 \ell \ln n / p^{2\ell}$, for a large enough constant $C$,
the graph $G$ has $(2s, s, 1/k, p, \ell)$-{bounded error sets}, i.e.,
that for all rectangles $Q = \bigtimes_{i \in [k]} Q_i$ satisfying
$|Q_i| \geq 2s$ or $|Q_i| = 0$ it holds that there exists a small set
$W$, $|W| \le s$, such that for all $S \subseteq [k]$ of size at most
$\ell$ it holds that all tuples
$t \in \bigtimes_{i \in S} (Q_i\setminus W) $ satisfy
$|N^\cap(t, Q_j\setminus W)| = (1\pm 1/k)p^{|t|}|Q_j \setminus W|$ for
all $j\in [k]\setminus S$.

\begin{corollary}\label{cor:bounded-error}
  There exists a constant $C \in \R^+$ such that the following holds
  for all $D \in \R^+$ and integers $k, n, s \in \N^+$ satisfying
  $k \le n$ and $p = n^{-2/D} \le 1/2$. A graph $G$ sampled from
  $\calG(n, k, p)$ has asymptotically almost surely
  $(2s, s, 1/k, p, \ell)$-bounded error sets for all
  $1 \le \ell \leq D/4$ and $s \ge C k^4 \ell \ln n / p^{2\ell}$.
\end{corollary}

\begin{proof}
  Let $T \subseteq [k]$ and $Q = \bigtimes_{i \in [k]} Q_i$ be such
  that $|Q_i| \geq 2s$ if $i \in T$ and $|Q_i| = 0$ otherwise. Our
  goal is to show that there exists a small set of vertices $W
  \subseteq V(G)$, $|W| \le s$, such that for all $T' \subseteq T$ of
  size at most $\ell$ it holds that all tuples $t \in \bigtimes_{i \in
    T'} (Q_i\setminus W) $ satisfy
  \begin{align*}
    \bigl|N^\cap(t, Q_j\setminus W)\bigr| \in
    (1\pm {1}/{k})p^{|t|}\bigl|Q_j \setminus W\bigr|
  \end{align*}
  for all $j\in [k]\setminus T'$.

  Fix $\ell \leq D/4$, and the parameters $\gamma = 1/3k$ and $w=12
  \ell \ln n/ p^{\ell} \gamma^2 $.  Note that by choosing $C \ge 3^3
  \cdot 12$ it holds that $w \le s\gamma p^\ell /k$.  For each $i \in
  T$, let $W_i$ be the set, guaranteed to exist (asymptotically almost
  surely) by \cref{lem:common-neigh}, of size at most $w$ such that
  for all $T' \subseteq T \setminus i$ of size at most $\ell$ it holds
  that all tuples $t \in \bigtimes_{j \in T'} (Q_j \setminus W_i)$
  satisfy
  \begin{align}
    |N^\cap(t, Q_i)| = (1\pm \gamma)p^{|t|} |Q_i| \eqperiod
  \end{align}
  Note that to apply \cref{lem:common-neigh} we use the fact that
  $|Q_i| \geq 2w$ for all $i \in T$, which follows since $|Q_i| \geq
  2s$ and $w\le s \gamma p^\ell / k \le s$.

  Let $W = \cup_{i\in T} W_i$.  Since $w \le s\gamma p^\ell/k $, we
  have that $|W| \leq kw \le s\gamma p^\ell $. Moreover, for all $T'
  \subseteq T$ of size at most $\ell$ and all $j \in [k] \setminus
  T'$, it holds that all tuples $t \in \bigtimes_{i \in T'}
  (Q_i\setminus W) $ satisfy
  \begin{align}
    |N^\cap(t, Q_j \setminus W)| =
    (1\pm 3\gamma)p^{|t|}|Q_j \setminus W| \eqcomma
  \end{align}
  by \cref{lem:neigh-remove-set} where we use the bound
  $\frac{|W|}{|Q_j|} \leq \frac{|W|}{2s} \leq \frac{\gamma
    p^{\ell}}{2} \leq \min\{\gamma/2,\gamma p^{\ell}\}$.  Since $1/k =
  3\gamma$ and $|W| \leq s\gamma p^\ell \leq s$, the statement
  follows.
\end{proof}

\subsection{Random Graphs Have Bounded Character Sums}
\label{sec:bounded-char}

In order to show that random graphs are asymptotically almost surely
well-behaved, we rely on the following probabilistic bound on the
weighted sums of Fourier characters, inspired by~\cite{AMP21norm}.

\begin{restatable}{restatablelemma}{encodinglemma}
  \label{lem:encode}
  Let $F$ be a non-empty graph over the vertex set $[k]$ with
  $A = V\bigl(E(F)\bigr)$, and let $M \subseteq E(F)$ be any matching in $F$.
  Let $Q = \bigtimes_{u \in A} Q_u$ be a rectangle such that for every
  edge $\set{u,v} \in M$ it holds that $|Q_u \times Q_v| \ge \lbQ$.
  For any even $m \le \lbQ$, any $r \in \R^+$ and any function
  $\xi: Q \rightarrow [-r,r]$ it holds that
  \begin{align*}
    \PROB[{G[V_A]}]
    {\big|\sum_{t \in Q} \chi_{F(t)}(G)\xi(t)\big| > s} 
    \le
    \left(
    \frac{
    r
    \cdot
    p^{-{|E(F)|}} 
    \cdot
    \Bigl(
    \frac{m}{\lbQ}
    \Bigr)^{|M|/2}
    \cdot
    |Q|
    }
    {s}
    \right)^m 
    \eqcomma 
  \end{align*}
  where each edge in $\binom{V_A}{2}$ is sampled independently at
  random with probability $p$.
\end{restatable}

It may be illustrative to consider the case when $F = \set{i,j}$ is a
single edge, $\xi$ is the constant $1$ function,
$|Q_i \times Q_j| = \kappa = |Q| = n$ and $p = 1/2$. In this setting
we obtain a bound on the absolute value of $n$ random variables taking
values in $\pm 1$ uniformly at random. For some setting of parameters
we essentially recover the Chernoff bound: we may, for example, set
$s = \sqrt{n\log n}$ and $m = \frac{\log n}{16}$ to obtain the bound
\begin{align}
  \Prob{\,\big|\sum_{e \in [n]} \chi_e\big| > \sqrt{n \log n}\,} 
    \le
    2^{-\log(n)/16} 
    \eqcomma
\end{align}
which is essentially the bound that Chernoff guarantees.

The proof of \cref{lem:encode} goes through a high moment version of
the Markov inequality. After some standard manipulations of the
resulting expressions we are left to bound the number of ways to map
$m$ copies of the graph $F$ into $Q$ such that every $2$-tuple in $Q$
is either mapped to at at least twice or not at all. We bound this
quantity via an encoding argument.

The main idea of the encoding argument is to only consider the edges
in the matching~$M$. Each such edge can be mapped in at least $\kappa$
ways. As $m \le \kappa$, we argue that it is beneficial to point to
the other copy of $F$ that maps this edge to the same place. This
drives the upper bound on the failure probability.  The details of the
argument are worked out in \cref{sec:encode}.

With this lemma in hand we are ready to prove the main result of this
section.

\begin{proof}[Proof of \cref{lem:Gwellbehaved}]
  \Cref{it:expected-neigh,it:bounded-error} hold asymptotically almost
  surely by \cref{lem:neighbor} and \cref{cor:bounded-error},
  respectively.

  For \cref{it:bounded-char-tight,it:bounded-char-general} we will
  apply \cref{lem:encode} to bound the probability that for a fixed
  core graph $F$ and rectangle $Q$ that satisfy the conditions stated
  $G$ does not have bounded character sums over $Q$ for $F$. We will
  then apply a union bound over all such $F$'s and $Q$'s.

  We start by defining a function $\xi_{F,Q,G,B}$ for every core $F$,
  rectangle $Q$, graph $G$ and subset $B \subseteq [k]$ that maps
  tuples $t_A\in Q_{A}$, for $A=V\bigl(E(F)\bigr) \cap B$, to real numbers as
  \begin{align}
    \xi_{F,Q,G,B}(t_A) =
    \sum_{t \in Q_B : t\supseteq t_A}
    \sum_{E \subseteq E^\star_F[B]}
    \chi_{E(t)}(G) =
    p^ {-|E^\star_F[B]|} \cdot
    \big|\Set{
    t \in Q_B\colon
    t\!\supseteq\! t_A
    \land
    E^\star_F[B](t)\! \subseteq\! E(G) }\big|
    \eqperiod
  \end{align}
  Observe that
  \begin{align}\label{eq:to-xi}
    \sum_{t \in Q_B} \sum_{H \in \calH(F, E^\star_F[B])}
    \chi_{H[B](t)}(G)
    &=
      \sum_{t\in Q_B}
      \chi_{F[B](t)}(G)
      \sum_{E \subseteq E^\star_F[B]}
      \chi_{E(t)}(G)\\
    &=
      \sum_{t_A \in Q_A}
      \chi_{F[B](t_A)}(G)\,
      \xi_{F,Q,G,B}(t_A)
      \eqperiod
  \end{align}
  
  Note that $\xi_{F,Q,G,B}$ only depends on $F$, $Q$ and
  $G[B]\setminus G[V_A]$, that is, it does not depend on the edges
  between vertices in $V_A$. This will be crucial for our analysis as
  it implies that once we fix $B$, $F$, $Q$ and
  $G_{\bar A} = G \setminus G[V_A]$ we can determine whether
  $\xi_{F,Q,G,B}$ is $r$-bounded.

  For a fixed $r$, let $\mathbf{X}_r(F,Q,G,B)$ be the event that
  $\xi_{F,Q,G,B}$ is $r$-bounded. By \cref{lem:encode}, we claim that
  for any $s$ if $Q_B$ is $(\kappa,s)$-\compatible, then for any
  $m \leq \kappa$ it holds that
 \begin{align}\label{eq:bounded-char}
   \Pr_{G}
     \left[
     \Bigl|
     \sum_{t \in Q_B}
     \sum_{H \in \calH(F, E^\star_F[B])}\!\!
     \chi_{H[B](t)}(G)
     \Bigr| > s \;\Bigg\vert\; \mathbf{X}_r(F,Q,G,B)
     \right]
   \le
    \left(
    \frac{
    r\,
    p^{-{|E(F[B])|}} 
    \,
    \Bigl(
    \frac{m}{\lbQ}
    \Bigr)^{\vc(F[B])/4}
    \,
    |Q_A|
    }
    {s}
    \right)^m \!.
 \end{align}
 Indeed, this follows by rewriting the left hand-side according to
 \cref{eq:to-xi} to obtain
 \begin{align}
   \begin{split}
     \Pr_{G}
     \left[
     \vphantom{\Big|\sum_{t \in Q_A}\Big|}\right.
     \Bigl|
     &\sum_{t \in Q_A}
     \chi_{F[B](t)}(G)\,
     \xi_{F,Q,G,B}(t)
     \Bigr|
     >
     s
     \;\Bigg\vert\; \mathbf{X}_r(F,Q,G,B)
     \left.\vphantom{\Big|\sum_{t_A \in Q}\Big|}
       \right]
     \\
     &\:\:\:=
       \sum_{\substack{G_{\bar A} \colon\\
       \mathbf{X}_r(F, Q, G, B) }}
       \Prob[{G \setminus G[V_A]}]{{G}_{\bar{A}} \;\big\vert\;
       \mathbf{X}_r(F,Q,G,B)} \cdot
       \Pr_{G[V_A]}
       \Bigl[
       \bigl|
       \sum_{t \in Q_A} \chi_{F[B](t)}(G)\,\xi_{F,Q,G,B}(t)\bigr| >
       s
       \Bigr] \eqcomma
   \end{split}
 \end{align}
  and applying \cref{lem:encode} for each fixed ${G}_{\bar{A}}$.
  
  We are now ready to prove that asymptotically almost surely $G$
  satisfies \cref{it:bounded-char-tight,it:bounded-char-general}. Fix
  a core $F$ and a rectangle $Q$.

  For \cref{it:bounded-char-general} we may assume that $Q$ is
  $(n/2,F)$-\compatible. Let $B = [n]$ and
  $C = [k] \setminus V\bigl(E(F)\bigr)$. By \cref{it:expected-neigh}
  $G$ has $(1/k, p, D/4)$-bounded common neighborhoods in every block
  hence $\xi_{F,Q,G,B}$ is $r$-bounded for
  $r = (1+1/k)^{|C|} n^{|C|} \le 3 \cdot n^{|C|}$. Applying
  \cref{eq:bounded-char} for this value of $r$ and for
  $s = 6 \cdot p^{-{|E(F)|}} \cdot n^{-\lambda \vc(F)/4} \cdot n^k$,
  $\kappa = n^2/4$ and $m = n^{2-\lambda}/4$, we derive that the
  probability that $G$ does not have $s$-bounded character sums over
  $Q$ for $F$ is at most $2^{-m} = 2^{-n^{2 - \lambda}/4}$. We
  conclude \cref{it:bounded-char-general} by taking a union bound over
  all cores of vertex cover at most $D/4$,
  of which, according to \cref{lem:count-H},
  there are at most
  $
  \sum_{i = 1}^{D/4} 2^{3i(i + \log k)}
  \leq 2\cdot 2^{D(D/4 +  \log k)}
  \ll 2^{ n^{2-\lambda}/4}
  $,
  and all $2^{nk} \ll 2^{n^{2-\lambda}/4}$ rectangles, using
  the fact that $\lambda < 1 - (\log k)/(\log n)$ and $k \le n^{1/3}$.
  
  \Cref{it:bounded-char-tight} follows by a similar argument. Further
  fix an integer $\Lambda \in [n]$, $\Lambda \ge 20\,k \log n$ and a
  subset $B \subseteq [k]$, let $A=V\bigl(E(F)\bigr) \cap B$ and
  $C = B \setminus A$.
  We assume $Q$ is $(4\Lambda)$-small and $(\Lambda,F)$-\compatible.
  Note that if $G$ is such that for every $i\in C$, $G$ has
  $(3/k,p)$-bounded common neighborhoods from $Q_A$ to $Q_i$, then
  $\xi_{F,Q,G,B}$ is $r$-bounded for
  $r =
  (1+3/k)^{|C|}\,\lvert Q_C\rvert <
  30\,\lvert Q_C\rvert$.

  We can therefore apply \cref{eq:bounded-char} with parameters
  $r = 30\,|Q_C|$,
  $\kappa = \Lambda^2$,
  $m = 10\Lambda k \log n \leq \kappa$, and
  $s =
  60\,
  p^{-{|E(F[B])|}} 
  (m/\kappa)^{\vc(F[B])/4} 
  |Q|$,
  to obtain that the probability that $G$
  does not have $s$-bounded character sums over $Q_B$ for $F$ is at most
  $2^{-m} = 2^{-10\Lambda k \log n}$. Applying a union bound over
  all $\Lambda \in [n]$ such that $\Lambda \geq 20\,k\log n$,
  all cores of vertex cover at most $D/4$, all $(4\Lambda)$-small
  rectangles and all subsets $B \subseteq [k]$
  we conclude that the probability that $G$ does not
  satisfy \cref{it:bounded-char-tight} is at most
  $2^{-\Lambda k \log n}$ since, by
  \cref{lem:count-H}, there are at most
  $
  \sum_{i = 1}^{D/4} 2^{3i(i + \log k)}
  \leq 2\cdot 2^{D(D/4 +  \log k)}
  \ll 2^{\Lambda k \log n}
  $ many core graphs,
  $
  \binom{n}{\leq 4 \Lambda}^k
  \leq
  2 \cdot 2^{4\Lambda k \log n}
  $
  rectangles that are $(4\Lambda)$-small and only $2^k$ subsets $B$.
\end{proof}

%
\subsection{Probabilistic Bound on Sums of Fourier Characters}
\label{sec:encode}

The rest of this section is dedicated to the proof of 
\cref{lem:encode}, restated here for convenience.

\encodinglemma*

\begin{proof}
  In order to bound the probability
  $\Prob[G]{|\sum_{t \in Q} \chi_{F(t)}(G)\xi(t)| > s}$ we resort to a
  high moment version of the Markov inequality and conclude that
  \begin{align}\label{eq:encode-markov}
    \PROB[{G[V_A]}]
    {\big|\sum_{t \in Q} \chi_{F(t)}(G)\xi(t)\big| > s}
    \le
    \frac{
    \E_{G[V_A]}
    \Bigl[
    \big(
    \sum_{t \in Q} \chi_{F(t)}(G)\xi(t)
    \big)^m
    \Bigr]
    }
    {s^{m}} \eqcomma
  \end{align}
  for even $m$. For conciseness, let us write $G$ instead of $G[V_A]$
  and rewrite above expectation to obtain that
  \begin{align}
    \label{eq:encode-expectation}
    \E_G
    \Bigl[
    \big(
    \sum_{t \in Q} \chi_{F(t)}(G)\xi(t)
    \big)^m
    \Bigr]
    &=
      \sum_{t_1, \ldots, t_m \in Q}
      \E_G
      \bigl[
      \prod_{i \in [m]}
      \chi_{F(t_i)}(G)
      \bigr]
      \cdot
      \prod_{i \in [m]}
      \xi(t_i)\\
    &\le
      \sum_{t_1, \ldots, t_m \in Q}
      \big|
      \E_G
      \bigl[
      \prod_{i \in [m]}
      \chi_{F(t_i)}(G)
      \bigr]
      \big|
      \cdot
      \big|
      \prod_{i \in [m]}\xi(t_i)
      \big|\\
    &\le
      \max_{t \in Q} |\xi(t)|^m \cdot
      \sum_{t_1, \ldots, t_m \in Q}
      \big|
      \E_G
      \bigl[
      \prod_{i \in [m]}
      \chi_{F(t_i)}(G)
      \bigr]
      \big|
      \eqperiod
  \end{align}
  Observe that the magnitude of each expectation in
  \cref{eq:encode-expectation} is $0$ unless each edge appears at
  least twice in the product, in which case it is at most
  $(1/p)^{|E(F)|\cdot m}$. Indeed, for $\ell \in \N$ it holds that
  %
  \begin{align}
    \big|
    \E_G
    \bigl[
    \chi_{e(t)}(G)^\ell 
    \bigr]
    \big|
    \leq (1-p)\cdot 1 + p\cdot \frac{(1-p)^\ell}{p^\ell}
    \leq (1-p)\cdot \frac{1}{p^{\ell}} + p\cdot \frac{1}{p^{\ell}}
    = \frac{1}{p^{\ell}}\eqperiod
  \end{align}
  Therefore, given $t_1, \ldots, t_m \in Q$ such that each edge
  appears at least twice, we have that
  \begin{align}
    \big|
    \E_G
    \big[
    \prod_{i \in [m]}
    \chi_{F(t_i)}(G)
    \big]
    \big|
    &\leq p^{-|E(F)|\cdot m}\eqperiod
  \end{align}
  
  It remains to bound the number of tuples $t_1, \ldots, t_m$ such
  that each edge appears at least twice. In more detail, we want to
  count the number of ordered sequences $(t_1, \ldots, t_m)$ such that
  if $F$ is mapped to each tuple $t_i$, then in the resulting
  multi-graph $\calF = \cup_{i \in [m]} F(t_i)$ each edge appears at
  least twice.

  To this end, let $\calT$ be the family of ordered sequences
  $(t_i)_{i \in [m]}$ that give rise to such multi-graphs. In what
  follows we bound the cardinality of $\calT$ by an encoding
  argument: we encode each element of $\calT$ with few bits and thus
  establish that this set is small.

  Let us fix one such sequence $(t_i)_{i \in [m]} \in \calT$. Encoding
  each tuple separately is too costly. To minimize the bits needed we
  use the property that in the resulting multi-graph
  $\calF = \cup_{i \in [m]} F(t_i)$ each edge is covered at least
  twice---in fact, we only use this property for edges in the
  matching. We encode the sequence as follows.
  \begin{itemize}
  \item First, for each $i \in [m]$, we go through the edges of the
    $i$th copy of $M$ in some fixed order. Suppose we consider an edge
    $\set{u,v} \in M$ in the $i$th copy of $F$. If this edge has not
    been matched yet, then we write down an integer $0 \le j \le m$
    indicating that the tuples $t_i$ and $t_j$ map the vertices $u$ and
    $v$ to the same place. The edge $\set{u,v}$ of the $j$th copy of $F$
    is then said to be matched. Otherwise, if the edge $\set{u, v}$ of
    the $i$th copy of $F$ has already been matched, then we continue
    with the next edge.
  \item After completing above step we encode the tuples. We use the
    knowledge of the previous step: if we indicated that the $u$th
    vertex of the tuples $t_i$ and $t_j$ are mapped to the same place,
    then we encode the target only once. More precisely, we encode the
    tuples as follows. Iterate over each $i \in [m]$ and
    $u \in V(E(F))$. If the $u$th vertex of $t_i$ has already been
    mapped, then continue with the next vertex. Otherwise, we write
    down an integer $0 \le v \le |Q_u|$ indicating that the $u$th
    vertex of $t_i$ is equal to $v \in Q_u$. If the $u$th vertex of
    $t_i$ has matched edges incident, then record the corresponding
    $u$th vertices as mapped.
  \end{itemize}

  It should be evident that this procedure can be inverted, i.e., from
  the information written we can uniquely decode the sequence
  $(t_i)_{i \in [m]}$, provided we know the graph $F$, the matching
  $M$, the rectangle $Q$ and the integer $m$.

  \newcommand{\numberMatched}{a}
  
  Let us tally the bits used in the above encoding. Suppose that in the
  first step we matched $\numberMatched$ edges.
  \begin{itemize}
  \item For each matched edge we wrote down $\log m$ bits for a total of
    $\numberMatched \cdot \log m$ bits, and
  \item at most
    $m \cdot \sum_{u \in V(E(F))} \log |Q_u| - \numberMatched \cdot \log \lbQ$ many
    bits in the second step: each matched edge reduces the total
    number of vertices to be mapped by $2$ (using that we only match
    edges in $M$). As for every edge $\set{u,v} \in M$ it holds that
    $|Q_u \times Q_v| \ge \lbQ$, we see that the number of bits
    required in the second step of the encoding is reduced by each
    matched edge by at least $\log \lbQ$ bits.
  \end{itemize}

  Hence we need at most
  $m \cdot \sum_{u \in V(E(F))} \log |Q_u| - \min_{\frac{m \cdot |M|}{2}
    \le \numberMatched \le m\cdot|M|} \numberMatched (\log \lbQ - \log m)$ many bits to encode
  such a sequence of tuples. As we assumed that $m \le \lbQ$ we
  conclude that
  \begin{align}
    \E_G
    \bigl[
    \big(
    \sum_{t \in Q} \chi_{F(t)}(G)\xi(t)
    \big)^m
    \bigr]
    &\le
    \max_{t \in Q} |\xi(t)|^m \cdot
    p^{-{|E(F)|\cdot m}} \cdot
    |\calT| \\
    &\le
    \left(
    \numberMatched \cdot
    p^{-{|E(F)|}} \cdot
    \Bigl(
    \frac{m}{\lbQ}
    \Bigr)^{|M|/2}
    \prod_{u \in V(E(F))} |Q_u|
    \right)^m
    \eqperiod
  \end{align}
  Substituting this bound in \cref{eq:encode-markov} results in the
  desired statement.
\end{proof}

\section{Concluding Remarks}
\label{sec:conclusion}

For $k \le n^{1/100}$ we prove an essentially tight average-case
$n^{\Omega(D)}$ coefficient size lower bound on Sherali-Adams
refutations of the $k$-clique formula for Erd\H{o}s-Rényi random
graphs with maximum clique of size $D$. In fact, we obtain a lower
bound on the sum of the magnitude of the coefficients appearing in a
(general) Sherali-Adams refutation.  The obvious problem left open is
to prove an $n^{\Omega(D)}$ \emph{monomial} size lower bound on
Sherali-Adams refutations of the clique formula.

One possible avenue to prove such a monomial size lower bound is to
argue that any Sherali-Adams proof of the clique formula can be
converted into a proof of the same monomial size but with small
coefficients. As shown in~\cite{GHJMPRT22} this is in general not
possible and one would hence have to leverage the structure of the
clique formula to argue that such a conversion exists. In fact, a
slightly weaker statement would suffice: recall that our lower bound
only counts the size of the coefficients of generalized monomials as
well as of monomials multiplied by edge axioms. As such we would just
need to be able to convert a general Sherali-Adams refutation into a
refutation with low coefficients for such monomials.

In contrast to previous lower bounds for clique, our proof strategy is
not purely combinatorial.  It might be fruitful to obtain an explicit
combinatorial description of $\mu_d$---we believe this could
potentially be used to prove average-case clique lower bounds for
other proof systems, including resolution.

A strength of our lower bound approach is that it is quite oblivious
to the encoding: one can introduce all possible extension variables
depending on a \emph{single} block and the lower bound argument still
goes through. This is because the only property we require of a
monomial $m$ is that the set of tuples
$Q_m = \set{t \mid \rho_t(m) \neq 0}$ whose associated assignment
$\rho_t$ sets $m$ to non-zero is a rectangle. By extending $\rho$ in
the natural manner to extension variables it is easy to see that $Q_m$
is still a rectangle.

Our lower bound strategy seems to fail quite spectacularly once the
edge probability is increased well beyond $1/2$. More precisely, once
$D = \omega(\log n)$, we fail to counter exponential in $d^2$ factors
that arise from encoding the core graphs: as long as $D = O(\log n)$
we can counter these with $s^{-d}$ terms, where $s$ is the minimum
block size of a good rectangle. As $s$ is clearly bounded by the block
size $n$, this approach fails once $D = \omega(\log n)$. We leave it
as an open problem to extend our result to the dense setting.

We rely on rather unorthodox pseudorandomness properties of the
underlying graph. It is natural to wonder whether these properties
follow from a previously studied notion of
pseudorandomness. Furthermore, it is wide open whether our lower bound
can be made explicit. In particular, we have not investigated whether
graphs that satisfy our pseudorandomness property can be constructed
deterministically.

Another application of our pseudo-measure $\mu_d$ is in communication
complexity. Suppose we consider the $k$-player number-in-hand model,
where player $i$ obtains a single node $u_i$ from block $V_i$. The
goal of the $k$ players is to find an edge missing in the induced
subgraph by the tuple $(u_1, \ldots, u_k)$. Consider the leaves of
such a communication protocol. Note that each leaf $\ell$ is
associated with a subrectangle $Q_\ell$ of an edge axiom. As the 
family
of these associated rectangles $Q_\ell$ partition the whole space, but
$|\mu_d(Q_\ell)| \le n^{-\Omega(D)}$, there must be at least
$n^{\Omega(D)}$ leaves.

Finally, we have not investigated whether our technique can be used to
obtain lower bounds for other proof systems. For example, is it
possible that with similar ideas one could obtain tree-like cutting
planes lower bounds with bounded coefficients? Possibly even with
unbounded coefficients?  The communication complexity view of the
problem suggests that this may be a viable approach.

\section*{Acknowledgements}
The authors are grateful to Albert Atserias, Per Austrin, Johan
Håstad, Jakob Nordström, Pavel Pudlák, Dmitry Sokolov, Joseph
Swernofsky, Neil Thapen and Marc Vinyals for helpful discussions and
feedback. In particular we would like to thank Albert Atserias who
observed that cores seem to be related to kernels.

We benefited from feedback of the participants of the Oberwolfach
workshop~2413 ``Proof Complexity and Beyond'' and would like to thank
the anonymous FOCS reviewers whose comments helped us improve the
exposition in the paper considerably.

This work was supported by the Approximability and Proof Complexity
project funded by the Knut and Alice Wallenberg Foundation. Part of
this work was carried out while all authors were associated with KTH
Royal Institute of Technology. Other parts were carried out while
taking part in the semester program \emph{Lower Bounds in
  Computational Complexity} in the fall of 2018 and the semster
programs \emph{Meta-Complexity} and \emph{Satisfiability: Extended
  Reunion} in the spring of 2023 at the Simons Institute for the
Theory of Computing at UC Berkeley.

Susanna F. de Rezende received funding from ELLIIT, from Knut and
Alice Wallenberg grants \mbox{KAW 2018.0371} and \mbox{KAW 2021.0307},
and from the Swedish Research Council grant \mbox{2021-05104}. Aaron
Potechin was supported by NSF grant \mbox{CCF-2008920}. Kilian Risse
is supported by Swiss National Science Foundation project
\mbox{200021-184656} “Randomness in Problem Instances and Randomized
Algorithms”.

\appendix
\section{Explicit Characterization of $E^\star_F$}
\label{sec:explicit}

Instead of implicitly describing $E^\star_F$ as done in
\cref{sec:cores-bounds} we can describe this edge set explicitly as
follows. Given a graph $F$ and a set of vertices $\lexvc$, for every
$U \subseteq V\bigl(E(F)\bigr) \setminus \lexvc$ such that there is a
matching from $U$ to $\lexvc$ that covers $U$, we let $A_U\subseteq
\lexvc$ be the set of vertices $w \in \lexvc$ that could potentially
be used to extend one of these matchings.  More formally, $A_U$ is the
set of vertices $w \in \lexvc$ such that there is a matching from $U$
to $\lexvc\setminus \set{w}$ that covers $U$; in this way, if any
vertex $v \in [k] \setminus (\lexvc \cup U)$ is adjacent to $w$, then
there is a matching from~$U \cup \set{v}$ to~$\lexvc$ that covers $U
\cup \set{v}$.  With this definition we can define $E^\star_F$
formally as follows.

\begin{definition}[$E^\star_F$]
  Let $F \in \img(\core)$ and let $\lexvc, U_1, U_2$ denote the
  objects from \cref{alg:core} run on $F$.  For a set
  $U \subseteq V\bigl(E(F)\bigr) \setminus \lexvc$ 
  with a matching between $U$ and $\lexvc$ of size $|U|$, we let
  $A_U \subseteq W$ be the set of vertices
  \begin{align*}
    A_U = 
    \Set{
    w \in \lexvc
    \mid \exists
    \text{ matching of size }
    |U|
    \text{ between }
    \lexvc \setminus \set{w}
    \text{ and }
    U
    } \eqcomma
  \end{align*}
  and define
  \begin{align*}
    E^\star_F =
    \bigcup_{v \in [k] \setminus V(E(F)) }
    \Set{
    \set{v,w}\mid
    w \in \lexvc
    \setminus
    (A_{U_1 \cap [v-1]} \cup A_{U_2 \cap [v-1]})
    }\eqperiod
  \end{align*}
\end{definition}

\begin{lemma}
  For every $F \in \img(\core)$ it holds that $\core(H) = F$ if and
  only if $H = F \disjointunion E$ for some $E \subseteq E^\star_F$.
\end{lemma}

\begin{proof}
  We first argue that given $F \in \img(\core)$ then for any
  $E \subseteq E^\star_F$ it holds that
  $\core( F \disjointunion E) = F$.  Let $H = F \disjointunion E$ for
  some $E \subseteq E^\star_F$ and denote by $\lexvc, U_1$ and $U_2$
  the sets obtained by \cref{alg:core} on~$F$. Note that $H$ has the
  same edges as $F$ on the vertices $\lexvc \cup U_1 \cup U_2$. Let us
  argue that a run of \cref{alg:core} on $H$ results in the same sets
  $\lexvc$, $U_1$ and $U_2$.
  \begin{enumerate}
    
  \item All edges in $E^\star_F$ are incident to a vertex in $\lexvc$
    and hence $\lexvc$ is a vertex cover of $H$. Since $\lexvc$ is the
    lexicographically first minimum vertex cover of $F$ it must be the
    lexicographically first minimum vertex cover of $H$.

  \item We now argue that the set $U_1$ is the lexicographically first
    maximal subset of $[k] \setminus \lexvc$ with a matching in $H$ of
    size $|U_1|$ from $U_1$ to $\lexvc$. To see this, let
    $U'_1 \subseteq [k] \setminus \lexvc$ be the lexicographically
    first maximal set such that there is a matching $M$ in $H$ between
    $U'_1$ and $\lexvc$ that covers $U_1'$, and suppose for the sake
    of contradiction that $U'_1 \neq U_1$.  This implies that either
    $U'_1 \supsetneq U_1$ or $U'_1$ is lexicographically smaller than
    $U_1$.  Let $v$ be the lexicographically first vertex in the
    symmetric difference of $U'_1$ and $U_1$.  Note that it must be
    the case that $v\in U'_1$ (since either $U'_1 \supsetneq U_1$ or
    $U'_1$ is maximal and lexicographically smaller than $U_1$).  Let
    $w$ be the vertex in $\lexvc$ that $v$ is matched to in $M$. Note
    that $\{v,w\}$ cannot be in $F$ since the set
    $(U_1'\cap U_1)\cup \set{v}$ would contradict the choice of $U_1$.
    Moreover, it holds that
    $w \in A_{U'_1 \cap [v-1]} = A_{U_1 \cap [v-1]}$ hence $\{v,w\}$
    is not in $E^\star_F$. But $\set{v,w} \not\in E^\star_F \cup F$
    contradicts $\set{v,w} \in H \subseteq E^\star_F \cup F$.

  \item Similarly, we argue that the set $U_2$ is the
    lexicographically first maximal subset disjoint of
    $\lexvc \cup U_1$ with a matching in $H$ of size $|U_2|$ from
    $U_2$ to $\lexvc$. To see this, let
    $U'_2 \subseteq [k] \setminus (\lexvc \cup U_1)$ be the
    lexicographically first maximal set such that there is a matching
    $M$ in $H$ between $U'_2$ and $\lexvc$ that covers $U_2'$, and
    suppose for the sake of contradiction that $U'_2 \neq U_2$.  This
    implies that either $U'_2 \supsetneq U_2$ or $U'_2$ is
    lexicographically smaller than $U_2$.  Let $v$ be the
    lexicographically first vertex in the symmetric difference of
    $U'_2$ and $U_2$.  Note that it must be the case that $v\in U'_2$
    (since either $U'_2 \supsetneq U_2$ or $U'_2$ is maximal and
    lexicographically smaller than $U_2$).  Let $w$ be the vertex in
    $\lexvc$ that $v$ is matched to in $M$. Note that $\{v,w\}$ cannot
    be in $F$ since the set $(U_2'\cap U_2)\cup \set{v}$ would
    contradict the choice of $U_2$.  Moreover, it holds that
    $w \in A_{U'_2 \cap [v-1]} = A_{U_2 \cap [v-1]}$, hence $\{v,w\}$
    is not in $E^\star_F$. But $\set{v,w} \not\in E^\star_F \cup F$
    contradicts $\set{v,w} \in H \subseteq E^\star_F \cup F$.
  \end{enumerate}
  
  For the other direction, let $\lexvc, U_1$ and $U_2$ be the sets obtained by
  \cref{alg:core} on $H$, and let $F = H[\lexvc \cup U_1 \cup U_2]$.
  Assume, for sake of contradiction, that $H$ contains an edge
  $e \notin F \cup E^\star_F$.
  We analyse four cases depending on where the endpoints of $e$ are located.
  \begin{enumerate}
  \item If $e$ is not incident to a vertex in $\lexvc$, then this contradicts 
  the fact that $\lexvc$ is a vertex cover of~$H$. 

  \item If both endpoints of $e$ are in $V\bigl(E(F)\bigr) = \lexvc \cup U_1 \cup U_2$ 
    then $H$ has different edges than $F$ on the vertices $V\bigl(E(F)\bigr)$,
   contradicting that $F = H[\lexvc \cup U_1 \cup U_2]$.
  \end{enumerate}
    We are left to analyse the case when $e = \{v,w\}$ for some $w \in \lexvc$ and
    $v \in [k] \setminus V\bigl(E(F)\bigr)$. 
    Note that since $e \notin E^\star_F$ it must hold that 
    $w\in \lexvc \cap A_{U_1 \cap [v-1]}$ or $w\in \lexvc \cap A_{U_2 \cap [v-1]}$.

  \begin{enumerate}
\addtocounter{enumi}{2}
  \item If $e = \{v,w\}$ for some $w \in W\cap A_{U_1 \cap [v-1]}$ and
    $v \in [k] \setminus V\bigl(E(F)\bigr)$ (hence $v\not\in U_1$), 
    then there is a matching
    between $U=(U_1 \cap [v-1]) \cup \{v\}$ and $\lexvc$ that covers $U$. 
    Since $U$ precedes $U_1$ lexicographically,
    this contradicts our choice of $U_1$.

  \item If $e = \{v,w\}$ for some $w \in W\cap A_{U_2 \cap [v-1]}$ and
    $v \in [k] \setminus V\bigl(E(F)\bigr)$ (hence $v\not\in U_2$),  
    then there is a matching
    between $U=(U_2 \cap [v-1]) \cup \{v\}$ and $\lexvc$ that covers $U$. 
    Since $U$ precedes $U_2$ lexicographically,
    this contradicts our choice of $U_2$.
  \end{enumerate}
  In each case, we obtained a contradiction and therefore we conclude that if
  $\core(H) = F$ then $H$ cannot contain an edge $e \notin F \cup E^\star_F$. 
  This completes the proof of the lemma.
\end{proof}


\bibliographystyle{alpha}
\bibliography{references}

\end{document}